\documentclass{article}
\usepackage{fullpage}
\usepackage{amsthm}
\usepackage{amsfonts}
\usepackage{amsmath}
\usepackage{enumitem}
\usepackage{graphicx}
\usepackage{thm-restate}
\usepackage[hidelinks]{hyperref}
\usepackage{xcolor}

\newtheorem{theorem}{Theorem}
\newtheorem{definition}[theorem]{Definition}
\newtheorem{lemma}[theorem]{Lemma}
\newtheorem{fact}[theorem]{Fact}

\newtheorem{corollary}[theorem]{Corollary}
\newtheorem{criterion}[theorem]{Criterion}
\newtheorem{open}[theorem]{Open Problem}

\DeclareMathOperator{\diam}{diam}
\newcommand{\semisize}{\operatorname{semi-size}_n}
\DeclareMathOperator{\parent}{parent}
\DeclareMathOperator{\double}{double}

\DeclareMathOperator{\polylog}{polylog}
\DeclareMathOperator{\poly}{poly}
\DeclareMathOperator{\predecessor}{predecessor}
\DeclareMathOperator{\size}{size}

\title{A well-separated pair decomposition for low density graphs}
\author{
    Joachim Gudmundsson\footnote{University of Sydney, Australia. joachim.gudmundsson@sydney.edu.au. Joachim Gudmundsson was funded in part by the Australian Government through the Australian Research Council (project number DP240101353).}
    ~and 
    Sampson Wong\footnote{University of Copenhagen, Denmark. sampson.wong123@gmail.com. Sampson Wong was funded in part by the European Union's Marie Skłodowska-Curie Actions Postdoctoral Fellowship (project number 101146276), and in part by the Independent Research Fund Denmark's Sapere Aude research career programme (project number 1054-00032B).}
}
\date{}

\begin{document}

\maketitle

\begin{abstract}
Low density graphs are considered to be a realistic graph class for modelling road networks. It has advantages over other popular graph classes for road networks, such as planar graphs, bounded highway dimension graphs, and spanners. We believe that low density graphs have the potential to be a useful graph class for road networks, but until now, its usefulness is limited by a lack of available tools.

In this paper, we develop two fundamental tools for low density graphs, that is, a well-separated pair decomposition and an approximate distance oracle. We believe that by expanding the algorithmic toolbox for low density graphs, we can help provide a useful and realistic graph class for road networks, which in turn, may help explain the many efficient and practical heuristics available for road networks.
\end{abstract}

\section{Introduction}
\label{section:introduction}

\subsection{Motivation}
\label{subsection:motivation}

Advances in technology have made the collection of geographic data easier than ever before. Nowadays, continent-sized road networks are stored in graphs with up to a billion vertices and edges. To analyse these large graphs, highly efficient algorithms and data structures are required.

Many heuristics for analysing road networks are highly efficient in practice. One explanation as to why many heuristics are efficient on road networks but inefficient on general graphs is that these heuristics can exploit the underlying properties of road networks. For example, experiments show that road networks have desirable underlying properties such as small separators~\cite{DBLP:conf/ipps/DellingGRW11,DBLP:journals/jea/DibbeltSW16} and bounded maximum degree~\cite{DBLP:journals/networks/BoyaciDL22,DBLP:conf/gis/EppsteinG08}. A common structure that exists for road networks that does not exist for general graphs is an efficient shortest path data structure~\cite{DBLP:conf/alenex/BastFMSS07,DBLP:conf/aips/BlumS18,DBLP:conf/wea/GeisbergerSSD08,DBLP:conf/alenex/Gutman04}. 

Many attempts have been made to explain why road networks have desirable properties such as small separators and efficient shortest path data structures. In the theory community, a popular approach is to argue that road networks belong to a certain \emph{graph class}, and then to prove a set of desirable properties for the graph class. Numerous graph classes have been proposed for road networks. When assessing these graph classes, it is natural to consider two criteria~\cite{DBLP:conf/gis/Eppstein017}.

\begin{criterion}[Usefulness]
\label{critieria:usefulness}
How well does this graph class explain the desirable properties of road networks? Do graphs in this graph class have small separators, or efficient shortest path data structures?
\end{criterion}

\begin{criterion}[Realism]
\label{criteria:realism}
How well does this graph class model real-world road networks? Do all, or most, real-world road networks belong to this graph class?
\end{criterion}

Planar graphs, bounded highway dimension graphs, and spanners are among the most popular graph classes for road networks\footnote{In Section~\ref{subsection:other_graph_classes}, we provide a more extensive overview of existing graph classes for modelling road networks.}. Next, we will assess these popular graph classes using Criteria~\ref{critieria:usefulness} and Criteria~\ref{criteria:realism}. Then, we will introduce low density graphs, which is our preferred graph class for road networks. 

\textbf{Planar graphs.} Planar graphs are graphs with no edge crossings. An advantage of planar graphs is that it has a broad algorithmic toolbox, including separators~\cite{DBLP:conf/focs/LiptonT77}, shortest path separators~\cite{DBLP:conf/soda/Klein02,DBLP:journals/jacm/Thorup04}, cycle separators~\cite{DBLP:journals/jcss/Miller86}, $r$-divisions~\cite{DBLP:journals/siamcomp/Frederickson87}, recursive divisions~\cite{DBLP:journals/jcss/HenzingerKRS97}, abstract Voronoi diagrams~\cite{DBLP:conf/soda/Cabello17}, tree covers~\cite{DBLP:conf/icalp/BartalFN19,DBLP:conf/focs/ChangCLMST23}, and low treewidth embeddings~\cite{DBLP:conf/soda/ChangCC0PP25,DBLP:conf/focs/Cohen-AddadLPP23}, to name a few. Shortest path separators have been used to construct efficient approximate shortest path data structures in planar graphs~\cite{DBLP:conf/soda/Klein02,DBLP:journals/jacm/Thorup04}. 

Unfortunately, \mbox{real-world} road networks are far from planar~\cite{boeing2020planarity,DBLP:conf/gis/EppsteinG08}. Non-planar features of road networks include bridges, tunnels and overpasses. 

\textbf{Bounded highway dimension\footnote{In 2025, an alternative definition of highway dimension was proposed~\cite{DBLP:conf/soda/FeldmannF25}, which we will discuss further in Section~\ref{subsection:other_graph_classes}.}.} Highway dimension~\cite{DBLP:journals/jacm/AbrahamDFGW16}, introduced in 2010, is graph parameter for road networks. A graph has highway dimension~$h$ if, for every~$r$ and every ball of radius~$4r$, there exists~$h$ vertices that cover all shortest paths with length at least~$r$ inside the ball. Bounded highway dimension guarantees provable running times for shortest path data structure heuristics, such as Contraction Hierarchies~\cite{DBLP:conf/wea/GeisbergerSSD08}, Transit Nodes~\cite{DBLP:conf/alenex/BastFMSS07}, and Hub Labelling~\cite{DBLP:journals/siamcomp/CohenHKZ03}. For many researchers, bounded highway dimension is the preferred graph class for modelling road networks~\cite{DBLP:conf/soda/ColletteI24,DBLP:conf/soda/FeldmannF25,DBLP:conf/icalp/FeldmannFKP15,DBLP:conf/icalp/0002KV19,DBLP:conf/soda/JayaprakashS22,DBLP:conf/soda/KosowskiV17}.

Unfortunately, it is unknown whether real-world road networks have bounded highway dimension. It is NP-hard to compute the highway dimension of a graph~\cite{DBLP:conf/iwpec/000119}, making the model difficult to verify. Experiments suggest that even a lower bound for the highway dimension is non-constant, and may be superlogarithmic~\cite{DBLP:conf/cocoon/BlumS18}. Moreover, road networks that follow a grid layout~\cite{wiki:gridplan} have a large highway dimension~\cite{DBLP:journals/jco/0001FS21,DBLP:conf/soda/FeldmannF25}. For example, in Figure~\ref{figure:grids_and_trees} (left), the $\sqrt n \times \! \sqrt n$-grid has a highway dimension of $h = \Theta(\sqrt n)$. 

\textbf{Spanners.} Spanners are a popular class of geometric networks. For any pair of vertices in a geometric \mbox{$t$-spanner}, the shortest path distance in the graph between the two vertices is at most~$t$ times the Euclidean distance between the vertices. Spanners represent good network design, through the absence of large detours. Geometric spanners have been studied extensively, see the reference textbook~\cite{DBLP:books/daglib/0017763}. Efficient approximate shortest path data structures have been constructed for spanners~\cite{DBLP:journals/talg/GudmundssonLNS08}. 

Unfortunately, many road networks are not spanners. Rivers, mountains and valleys impose physical constraints on road networks~\cite{rodrigue2020geography}, which can lead to large detours~\cite{DBLP:conf/gis/HuaXT18}. Even without physical constraints, the suburban cul-de-sac layout~\cite{wiki:suburbanisation} prioritises buildable area over connectivity~\cite{grammenos2002residential,zhang2013cul}, which can also lead to large detours. In Figure~\ref{figure:grids_and_trees}~(right), the $\sqrt n \times \! \sqrt n$ comb models a suburban cul-de-sac network, and has a spanning ratio of $t = \Theta(\sqrt n)$.

\begin{figure}[ht]
    \centering
    \includegraphics[width=0.7\textwidth]{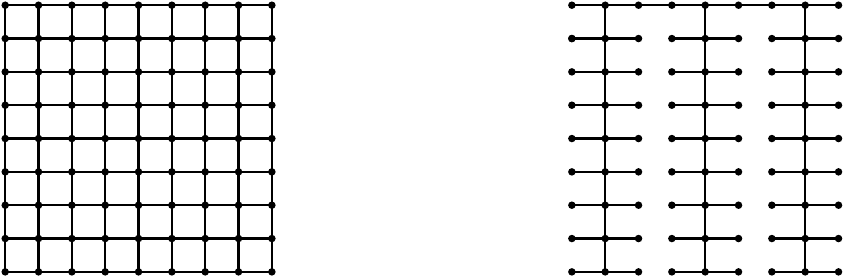}
    \caption{Left: A $\sqrt n \times \! \sqrt n$ grid has a highway dimension, skeleton dimension, and treewidth of $\Theta(\sqrt n)$. Right: A $\sqrt n \times \! \sqrt n$ comb has a spanning ratio, doubling constant, and highway dimension of $\Theta(\sqrt n)$.}
    \label{figure:grids_and_trees}
\end{figure}

\textbf{Low density graphs.} Low density~\cite{DBLP:books/daglib/0084325}, introduced in 1994, is a versatile geometric property that has been used to model obstacles~\cite{DBLP:journals/dcg/StappenOBV98}, trajectories~\cite{DBLP:journals/dcg/DriemelHW12} and road networks~\cite{DBLP:conf/alenex/ChenDGNW11}. Low density states that not too many large objects can be in any given area. We formally define low density in Section~\ref{section:definitions}. 

Low density graphs are considered to be a realistic graph class for modelling road networks. An advantage of low density graphs is that realistic features of road networks, such as edge crossings, grid-like structures (Figure~\ref{figure:grids_and_trees}, left) and comb-like structures (Figure~\ref{figure:grids_and_trees}, right), can all be modelled with low density graphs. Experiments support the claim that road networks are low density. It has been verified that the city-sized road networks of San Francisco, Athens, Berlin and San Antonio are all low density graphs~\cite{DBLP:conf/alenex/ChenDGNW11}.

Unfortunately, low density graphs only have a narrow algorithmic toolbox, limiting its usefulness as a graph class. Nonetheless, we believe that low density graphs have the potential to be a useful graph class for road networks. It was shown in~2022 that low density graphs have small separators~\cite{DBLP:conf/soda/LeT22}. Map matching under the Fr\'echet distance has been studied on low density graphs~\cite{DBLP:conf/compgeom/BuchinBG0W24,DBLP:conf/alenex/ChenDGNW11}. But little else is known. In particular, low density graphs lack an efficient shortest path data structure. Shortest path data structures are useful in both theory and practice, and are motivated by applications such as navigation systems~\cite{blog:uber,geisberger2015route,wiki:graphhopper}.

\begin{open}
    \label{open}
    Is there an efficient shortest path data structure for low density graphs?
\end{open}

Resolving Open Problem~\ref{open} would help bridge the gap, in terms of usefulness, between low density graphs and other popular graph classes.

\subsection{Contributions}
\label{subsection:contributions}

We resolve Open Problem~\ref{open} in the affirmative. We construct an approximate shortest path data structure for low density graphs, which has $O(n \varepsilon^{-4} \log n)$ size and can answer $(1+\varepsilon)$-approximate shortest path queries in $O(1)$ time. Along the way, we construct a~$(1/\varepsilon)$-well-separated pair decomposition of~$O(n \varepsilon^{-4} \log n)$ size. 

The well-separated pair decomposition is a highly influential structure in computational geometry, and has led a wide range of approximation algorithms for proximity problems in Euclidean spaces~\cite{DBLP:conf/soda/CallahanK93,DBLP:journals/jacm/CallahanK95}, doubling spaces~\cite{DBLP:journals/siamcomp/Har-PeledM06,DBLP:conf/stoc/Talwar04}, and unit disc graphs~\cite{DBLP:conf/stoc/GaoZ03}. We hope that our results will pave the way for more tools to be developed for low density graphs, which would further increase the usefulness of low density as a graph class. 

\subsection{Previous work}
\label{subsection:previous_work}

Gao and Zhang~\cite{DBLP:conf/stoc/GaoZ03} construct a well-separated pair decomposition for unit disc graphs. We regard this to be one of the most relevant results to our work. To the best of our knowledge, their result is the only previous work that constructs a well-separated pair decomposition for a graph class. Gao and Zhang~\cite{DBLP:conf/stoc/GaoZ03} show how to transform a well-separated pair decomposition of a graph into an approximate shortest path data structure. In Section~\ref{subsection:approximate_distance_oracle}, we apply essentially the same transformation.

The well-separated pair decomposition is a highly influential structure. Applications of the well-separated pair decomposition include spanners~\cite{DBLP:conf/soda/CallahanK93}, approximate nearest neighbour queries~\cite{DBLP:journals/jacm/AryaMNSW98}, tree covers~\cite{DBLP:conf/stoc/AryaDMSS95}, and approximate distance oracles~\cite{DBLP:conf/stoc/GaoZ03}. For a thorough treatment of the topic, refer to the textbook chapters by Har-Peled~\cite[Chapter~3]{har2011geometric} and by Narasimhan and Smid~\cite[Chapter~9]{DBLP:books/daglib/0017763}. Since the well-separated pair decomposition leads to efficient approximation algorithms for a wide range of proximity problems, a natural open question~\cite[Chapter~9.6]{DBLP:books/daglib/0017763} is: for which metrics spaces does there exist a well-separated pair decomposition of subquadratic size?

To the best of our knowledge, the only metrics in which the well-separated pair decomposition has previously been studied are Euclidean spaces, unit disc graphs, doubling spaces, and $c$-packed graphs. For Euclidean spaces, Callahan and Kosaraju~\cite{DBLP:journals/jacm/CallahanK95} constructed a well-separated pair decomposition with $O(n\varepsilon^{-d})$ pairs and separation constant~$(1/\varepsilon)$, for point sets in~$d$-dimensional Euclidean space. For unit disc graphs, Gao and Zhang~\cite{DBLP:conf/stoc/GaoZ03} constructed a well-separated pair decomposition with $O(n \varepsilon^{-4} \log n)$ pairs and a separation constant of~$(1/\varepsilon)$. Har-Peled, Raichel and Robson~\cite{har2025well} provided an improved construction with $O(n \varepsilon^{-2} \log n)$ pairs. For doubling spaces, Talwar~\cite{DBLP:conf/stoc/Talwar04}  provided a well-separated pair decomposition with $O(n \varepsilon^{-ddim} \log \Delta)$, where $\Delta$ is the spread of the point set, $ddim$ is the doubling dimension, and $(1/\varepsilon)$ is the separation constant. Har-Peled and Mendel~\cite{DBLP:journals/siamcomp/Har-PeledM06} provided an improved, randomised construction, with~$O(n\varepsilon^{-ddim})$ pairs. For $c$-packed graphs, Deryckere, Gudmundsson, van Renssen, Sha and Wong~\cite{DBLP:conf/wads/DeryckereGRSW25} constructed a well-separated pair decomposition with $O(c^3 n \varepsilon^{-1})$ pairs.

Low density was among the first realistic input models to be considered in computational geometry, see the introduction of~\cite{DBLP:conf/compgeom/BergKSV97}. In van der Stappen's thesis~\cite{DBLP:books/daglib/0084325}, he introduced low density and proved that non-intersecting fat objects formed a low density environment. A $\lambda$-low-density environment has at most~$\lambda$ large objects in any given area, where an object is large if it has greater diameter than that of the given area. Schwarzkopf and Vleugels~\cite{DBLP:journals/ipl/SchwarzkopfV96} constructed a range searching data structure for low density environments. Van der Stappen, Overmars, de Berg and Vleugels~\cite{DBLP:journals/dcg/StappenOBV98} studied motion planning in low density environments. Berretty, Overmars and van der Stappen~\cite{DBLP:journals/comgeo/BerrettyOS98} considered moving obstacles in dynamic low density environment. Har-Peled and Quanrad~\cite{DBLP:journals/siamcomp/Har-PeledQ17} studied the intersection graphs of low density environments. It is worth noting that the definition of low density graphs used in~\cite{DBLP:journals/siamcomp/Har-PeledQ17} different from our definition, but similar in spirit.

Driemel, Har-Peled and Wenk~\cite{DBLP:journals/dcg/DriemelHW12} introduced a definition of low density for a set of edges. A set of edges is $\lambda$-low-density if there are at most $\lambda$ long edges intersecting any given disc, refer also to Section~\ref{section:definitions}. Driemel, Har-Peled and Wenk~\cite{DBLP:journals/dcg/DriemelHW12} used low density to study polygonal curves under the Fr\'echet distance. They provided a near-linear time algorithm to compute the Fr\'echet distance between a pair of low density curves, by showing that after simplifying both curves, the free space diagram between the two curves has linear complexity. The $c$-packedness property was also introduced in Driemel, Har-Peled and Wenk~\cite{DBLP:journals/dcg/DriemelHW12}. 

Chen, Driemel, Guibas, Nguyen and Wenk~\cite{DBLP:conf/alenex/ChenDGNW11} were the first to consider low density graphs. They considered the map matching problem under the Fr\'echet distance on low density graphs. They additionally required that the matched trajectory is $c$-packed. Under these assumptions, they  provided a near-linear time algorithm for the map matching problem. Buchin, Buchin, Gudmundsson, Popov and Wong~\cite{DBLP:conf/compgeom/BuchinBG0W24} constructed a data structure for map matching under the Fr\'echet distance. Their data structure has subquadratic size, and answers map matching queries in sublinear time, provided that the road network is a low density spanner.

De Berg, Katz, van der Stappen and Vleugels~\cite{DBLP:conf/compgeom/BergKSV97} were first to consider the problem of computing the $\lambda$-low-density value. Given a set of $n$ non-intersecting objects, their algorithm computes the parameter $\lambda$ in $O(n \log^3 n + \lambda n \log^2 n + \lambda^2 n)$ time. Chen, Driemel, Guibas, Nguyen and Wenk~\cite{DBLP:conf/alenex/ChenDGNW11} provide an efficient implementation for computing the $\lambda$-low-density value of a graph, and verified that $\lambda \leq 28$ for the city-sized road networks of San Francisco, Athens, Berlin and San Antonio. Gudmundsson, Huang and Wong~\cite{DBLP:conf/cocoon/GudmundssonHW23} provided an $O(n \log n + \lambda n/\varepsilon^{4})$ time, $(2+\varepsilon)$-approximation algorithm for computing~$\lambda$. 

Low density graphs have small separators. Le and Than~\cite{DBLP:conf/soda/LeT22} introduced the $\tau$-lanky property, which is similar to and is implied by the $\lambda$-low-density property. See Definitions~\ref{definition:low-density} and~\ref{definition:lanky}. They proved that, in $\mathbb R^d$, any $\tau$-lanky graph contains a separator of size $O(n^{1-1/d})$.  Their result implies that a low density graph embedded in the Euclidean plane has an $O(\sqrt n)$-sized separator.

Approximate distance oracles are highly efficient shortest path data structures. Typically, an approximate distance oracle has a constant approximation ratio, subquadratic size, and constant query time. See Definition~\ref{definition:ado}. For general graphs, Thorup and Zwick~\cite{DBLP:conf/stoc/ThorupZ01} constructed an approximate distance oracle with an approximation ratio of $2k-1$, a size of $O(kn^{1+1/k})$, and a query time of $O(k)$. They also showed that for general graphs, any approximate distance oracle with an approximation ratio of $2k+1$ requires $\Omega(kn^{1+1/k})$ size. For planar graphs, Thorup~\cite{DBLP:journals/jacm/Thorup04} and Klein~\cite{DBLP:conf/soda/Klein02} independently constructed a $(1+\varepsilon)$-approximate distance oracle with $O(n \varepsilon^{-1} \log n)$ size and $O(\varepsilon^{-1})$ query time. Follow up works~\cite{DBLP:journals/algorithmica/ChanS19,DBLP:journals/tcs/GuX19,DBLP:conf/icalp/KawarabayashiKS11,DBLP:conf/soda/KawarabayashiST13,DBLP:conf/soda/Le23,DBLP:conf/soda/Wulff-Nilsen16} improved the size and query time trade-off. The state of the art result of Le and Wulff-Nilsen~\cite{DBLP:conf/focs/LeW21} is a $(1+\varepsilon)$-approximate distance oracle with~$O(n \varepsilon^{-2})$ size and~$O(\varepsilon^{-2})$ query time. For $K_r$-minor-free graphs, Chang, Conroy, Le, Milenkovi\'c, Solomon and Than~\cite{DBLP:conf/soda/ChangCLMST24} provided a $(1+\varepsilon)$-approximate distance oracle with~$O(n2^{\varepsilon^{-1}r^{O(r)}})$ size and~$O(2^{\varepsilon^{-1}r^{O(r)}})$ query time.
For geometric $t$-spanners, Gudmundsson, Levcopoulos, Narasimhan and Smid~\cite{DBLP:journals/talg/GudmundssonLNS08} constructed a $(1+\varepsilon)$-approximate distance oracle with $O(n t^4 \varepsilon^{-4} \log n)$ size and $O(t^6 \varepsilon^{-4})$ query time, see also \cite[Chapter~17.3]{DBLP:books/daglib/0017763}. For unit disc graphs, Gao and Zhang~\cite{DBLP:conf/stoc/GaoZ03} constructed a $(1+\varepsilon)$-approximate distance oracle with $O(n \varepsilon^{-4} \log n)$ size and $O(1)$ query time. Chan and Skrepetos~\cite{DBLP:journals/jocg/ChanS19a} improved the construction time of the approximate distance oracle for unit disc graphs. For bounded highway dimension graphs, Abraham, Delling, Fiat, Goldberg and Werneck~\cite{DBLP:journals/jacm/AbrahamDFGW16} constructed a distance oracle with $O(nh \log \Delta)$ size and $O(h \log \Delta)$ query time, where~$h$ is the highway dimension and~$\Delta$ is the spread.

\subsection{Other graph classes}
\label{subsection:other_graph_classes}

In Section~\ref{subsection:motivation}, we compared low density graphs to three popular graph classes, that is, planar graphs, bounded highway dimension graphs, and spanners. In this section, we will discuss some of the other graph classes that have also been considered for modelling road networks. 

The \emph{treewidth} of a graph measures how close to a tree the graph is. Specifically, a graph has treewidth~$tw$ if the size of the largest vertex set in its tree decomposition is~$tw$. Graphs with treewidth~$tw$ have separators of size~$O(tw)$~\cite{DBLP:journals/jct/DvorakN19} and have a shortest path data structures of $O(tw^3 n)$ size and $O(tw^3 \alpha(n))$ query time~\cite{DBLP:journals/algorithmica/ChaudhuriZ00}. Here, $\alpha(n)$ denotes the inverse Ackermann function. Unfortunately, the treewidth of the grid graph in Figure~\ref{figure:grids_and_trees} (left) is $\Theta(\sqrt n)$, and experiments suggest that $tw = \Omega(\sqrt[3] n)$ on real-world road networks~\cite{DBLP:journals/jea/DibbeltSW16,DBLP:conf/icdt/ManiuSJ19}.

The \emph{doubling dimension} of a metric measures its dimensionality by using a packing property. Formally, a metric has doubling dimension $ddim$ if any ball of radius~$r$ can be covered with at most $2^{ddim}$ balls of radius $r/2$. Talwar~\cite{DBLP:conf/stoc/Talwar04} and Har-Peled and Mendel~\cite{DBLP:journals/siamcomp/Har-PeledM06} constructed $(1+\varepsilon)$-approximate distance oracles for doubling metrics, with $O(n \log \Delta /\varepsilon^{ddim})$ size and $O(n /\varepsilon^{ddim})$ size, respectively. However, verifying the doubling dimension of a graph is NP-hard~\cite{DBLP:journals/siamdm/GottliebK13}, and the comb graph in Figure~\ref{figure:grids_and_trees} (right) has a doubling constant of $2^{ddim} = \Theta(\sqrt n)$~\cite[Chapter~5]{DBLP:phd/basesearch/Blum23}.  

The \emph{skeleton dimension} of a graph is a closely related parameter to the highway dimension. The skeleton dimension is the maximum width of the so-called ``skeleton'' of all shortest path trees. Kosowski and Viennot~\cite{DBLP:conf/soda/KosowskiV17} proved that $k < h$, where $k$ is the skeleton dimension and $h$ is the highway dimension. They also constructed a shortest path data structure of size~$O(nk \log \Delta)$ and query time of $O(k \log \Delta)$, where $\Delta$ is the spread of the point set. Unfortunately, the skeleton dimension of a grid graph in Figure~\ref{figure:grids_and_trees} (left) is $k = \Theta(\sqrt n)$. Blum and Storandt~\cite{DBLP:conf/cocoon/BlumS18} computed the skeleton dimension for real-world road networks and found that $k \leq 114$, which for their road network sizes corresponded to $\sqrt n \gg k \gg \log n$.

The \emph{$c$-packedness} of a graph is the smallest value of~$c$ such that, for any radius~$r$ and any disc of radius~$r$, the total length of edges inside the disc is at most $c \cdot r$. Deryckere, Gudmundsson, van Renssen, Sha and Wong~\cite{DBLP:conf/wads/DeryckereGRSW25} constructed, for $c$-packed graphs, a well-separated pair decomposition of size $O(c^3 n \varepsilon^{-1})$, an exact distance oracle of size $O(cn \log n)$, a tree cover with $O(c^6 \varepsilon^{-2})$ trees, and an approximate distance oracle of size $O(c^6 n \varepsilon^{-2})$. A $c$-packed graph has $O(c)$ treewidth and $O(c)$ doubling dimension, which makes it one of the most useful but least realistic graph classes for road networks. 

It is worth noting that low density graphs can have non-constant highway dimension, non-constant skeleton dimension and non-constant treewidth, see Figure~\ref{figure:grids_and_trees} (left). They can also have non-constant spanning ratio and non-constant doubling dimension, see Figure~\ref{figure:grids_and_trees} (right). Low density graphs can be non-planar. Therefore, any results for these other graph classes do not immediately extend to low density graphs.

Previous works have considered the shortcomings of popular graph classes. Eppstein and Goodrich~\cite{DBLP:conf/gis/EppsteinG08} noted that road networks are non-planar, and empirically measured the number of edge crossings to be proportional to~$\sqrt n$. Blum, Funke and Storandt~\cite{DBLP:journals/jco/0001FS21} and Feldmann and Filtser~\cite{DBLP:conf/soda/FeldmannF25} claimed that road networks contain grid graphs, and therefore a realistic model for road networks should also contain grid graphs. Eppstein and Gupta~\cite{DBLP:conf/gis/Eppstein017} stated the need for a graph class that is realistic, i.e. accurately models real-world road networks, and is useful, i.e. leads to efficient algorithms. 

Eppstein and Goodrich~\cite{DBLP:conf/gis/EppsteinG08} proposed a realistic alternative to planar graphs. They defined a \emph{disc neighbourhood system} to be a subgraph of a disc intersection graph where the set of discs have at most constant ply. They showed that disc neighbourhood systems admit $O(\sqrt n)$-sized separators, similar to planar graphs and real-world road networks. They also verified that on real-world road networks, the maximum ply is indeed constant, provided that large enough discs may be ignored. 

Eppstein and Gupta~\cite{DBLP:conf/gis/Eppstein017} proposed another realistic alternative to planar graphs. The \emph{crossing degeneracy} of a graph is the maximum minimum degree over all subgraphs of its crossing graph. They proved that the graph class admits $O(\sqrt n)$-sized separators, and verified that on real-world road networks, the crossing degeneracy is at most~6. We note that low density implies constant crossing degeneracy.

Blum, Funke and Storandt~\cite{DBLP:journals/jco/0001FS21} proposed a realistic alternative to highway dimension, skeleton dimension, and treewidth. Its main advantage is that it includes grid graphs. A graph has \emph{bounded growth} if, for all vertices~$u$ and all radii~$r$, the number of vertices within distance~$r$ of~$u$ is at most~$cr^2$ for some constant~$c$. The authors showed that bounded growth graphs have a shortest path data structure of size~$O(n \sqrt n)$ and a query time of~$O(\sqrt n)$. Funke and Storandt~\cite{DBLP:conf/isaac/FunkeS15} verified experimentally that the German road network is bounded growth for~$c=1$.

Feldmann and Filtser~\cite{DBLP:conf/soda/FeldmannF25} proposed a more general definition of highway dimension that includes grid graphs and doubling metrics. The main difference is, for every~$r$ and every ball of radius $\approx 4r$, instead of covering the shortest path between any pair of points at distance $>r$, the hitting set covers at least one approximate shortest path between the same pairs of points. The authors develop an extensive algorithmic toolbox using this definition. Unfortunately, it is unknown how realistic it is. There is currently no polynomial time algorithm for computing the highway dimension, moreover, it is merely assumed that the hitting sets can be computed in polynomial time. Since the highway dimension of a graph currently can not be measured with a polynomial time algorithm, it is difficult to verify the model empirically. We also note that the highway dimension of the comb graph in Figure~\ref{figure:grids_and_trees} (right) is $\Theta(\sqrt n)$.

Finally, we summarise our comparison between graph classes. Numerous graph classes have been proposed, some of these graphs are useful while others are realistic. Many of the graph classes have a shortest path data structure of $O(n \polylog n)$ size and $O(\polylog n)$ query time, which matches the practical performance of the most efficient heuristics~\cite{DBLP:conf/aips/BlumS18}. Some of the graph classes admit an~$O(\sqrt n)$ sized separator, which is a well-known property of real-world road networks~\cite{DBLP:conf/ipps/DellingGRW11,DBLP:journals/jea/DibbeltSW16}. For some graph classes, there exists a polynomial time algorithm to compute the graph parameter on real-world road networks, with empirical studies showing that the parameter is small in practice. Features such as grid graphs, comb graphs, and edge crossings can be modelled by many graph classes, but only a few can model all three features. We propose a model that unifies all three features. The SELG graph class is a generalisation of both the grid and comb graphs when~$\ell=1$, and includes edge crossings for $\ell \geq \sqrt 2$.

\begin{definition}[SELG]
\label{definition:SELG}
A short-edged lattice graph, or SELG, is a graph where all vertices are lattice vertices in $\mathbb Z \times \mathbb Z$, and all edges have length at most~$\ell$, where~$\ell$ is a constant.
\end{definition}

To the best of our knowledge, low-density graphs are the only known graph class that simultaneously $(i)$ admits an approximate distance oracle of size $O(n \polylog n)$ with $O(\polylog n)$ query time, $(ii)$ admits separators of size $O(\sqrt n)$, $(iii)$ has had the relevant graph parameter empirically measured on real-world graphs using a polynomial time algorithm, with results confirming that the parameter is small in practice, and $(iv)$ subsumes the SELG graph class and therefore models grid, comb and edge-crossing features.

\begin{table}[ht]
    \centering
    \renewcommand{\arraystretch}{1.1}
    \begin{tabular}{|c||c|c||c|c|c|c|c|}
    \hline
    & \multicolumn{2}{c||}{Useful} & \multicolumn{5}{c|}{Realistic} \\ \hline
    & Distance oracle & Separator & Empirical & SELG  & Grid & Comb & Cross \\ 
    \hline
    Planar & \checkmark & \checkmark &&& \checkmark & \checkmark & \\ \hline
    Highway dimension & \checkmark && && ~\checkmark && \checkmark \\ \hline
    Spanners &\checkmark &&&& \checkmark & & \checkmark \\ \hline
    Treewidth &\checkmark &\checkmark &&&& \checkmark & \checkmark \\ \hline
    Doubling dimension &\checkmark &&&&\checkmark && \checkmark\\ \hline
    Skeleton dimension &\checkmark &&&&&\checkmark& \checkmark \\ \hline
    $c$-packed &\checkmark & \checkmark &&&&& \checkmark \\ \hline
    Disc neighbourhood &&\checkmark  &\checkmark &\checkmark&\checkmark &\checkmark & \checkmark  \\ \hline
    Crossing degeneracy &&\checkmark  &\checkmark &\checkmark&\checkmark &\checkmark & \checkmark  \\ \hline
    Bounded growth &&&\checkmark &\checkmark &\checkmark &\checkmark & \checkmark \\ \hline
    Low density &\checkmark &\checkmark &\checkmark &\checkmark &\checkmark &\checkmark & \checkmark \\ \hline
    \end{tabular}   
    \caption{Comparing graph classes for road networks, using Criterion~\ref{critieria:usefulness} and Criterion~\ref{criteria:realism}.}
    \label{table:benchmarks} 
\end{table}

\section{Definitions}
\label{section:definitions}

Low density is a geometric property that states that not too many large objects can be in any given area. In the case of graphs, the objects are edges, the size of an object is its Euclidean length, and the given area is a disc. To define low density, we first define the density of a disc with respect to a graph. Given a graph and a disc in the plane, we define the density of the disc to be the number of long edges in the graph that intersect it, where an edge is defined as long if its length is at least the radius of the disc. Given a graph, the graph is $\lambda$-low-density if every discs in the plane has density $\leq \lambda$. Refer to Figure~\ref{figure:low_density_definition}.

\begin{figure}[ht]
    \centering
    \includegraphics[width=0.8\textwidth]{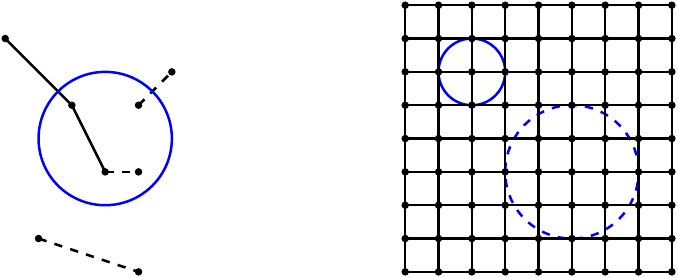}
    \caption{Left: The solid edges count towards the density of the blue disc, but the dashed edges do not. Right: A unit grid graph is a low density graph, since balls of radius~$r$ have constant density for all $r \leq 1$, and have zero density for all $r > 1$.}
    \label{figure:low_density_definition}
\end{figure}

We formally state the low density property in Definition~\ref{definition:low-density}. 

\begin{definition}[Low density graph]
    \label{definition:low-density}
    Let $G=(V,E)$ be a graph with a straight line embedding in $\mathbb R^2$, and let $\lambda \in \mathbb N$. We say that~$G$ is $\lambda$-low-density if, for all $r \in \mathbb R^+$ and all balls $B$ of radius $r$, there are at most $\lambda$ edges in~$E$ that intersect~$B$ and have length at least~$r$. If $\lambda = O(1)$, we say that~$G$ is a low density graph.
\end{definition}

Note that a subgraph of a low density graph is also low density. Next, we state the $\tau$-lanky property, which is similar to and is implied by $\lambda$-low-density.

\begin{definition}[Lanky graph~\cite{DBLP:conf/soda/LeT22}]
    \label{definition:lanky}
    Let $G=(V,E)$ be a graph embedded in $\mathbb R^2$, and let $\tau \in \mathbb N$. We say that $G$ is $\tau$-lanky if, for all $r \in \mathbb R^+$ and all balls $B$ of radius $r$ centred at a vertex~$v \in V$, there are at most $\tau$ edges in~$E$ that have length at least~$r$, have one endpoint inside~$B$, and have one endpoint outside~$B$.
\end{definition}

The result of Le and Than~\cite{DBLP:conf/soda/LeT22} states that $\tau$-lanky graphs have a separator of size~$O(\tau \sqrt n)$, see Fact~\ref{fact:balanced_separator_lanky}. Any~$\lambda$-low-density graph is~$\lambda$-lanky, so the separator result also applies to $\lambda$-low-density graphs. 

Recall from Section~\ref{subsection:contributions} that one of our contributions is to construct a well-separated pair decomposition for low density graphs. Next, we define the well-separated pair decomposition for a general metric~$M$. In this paper, we have two types of well-separated pair decompositions, those where the metric is the Euclidean metric, and those where the metric is the low density graph metric. See Figure~\ref{figure:wspd}. 

\begin{definition}[Well-separated pair]
    Let $A$ and $B$ be point sets in a metric space $M$. Define the distance between sets to be $d_M(A,B) = \min_{a \in A, b \in B} d_M(a,b)$, and define the diameter of~$A$ to be $diam_M(A) = \max_{a,a' \in A} d_M(a,a')$. Then the pair $(A,B)$ is a $(1/\varepsilon)$-well-separated pair in $M$ if 
    $$(1/\varepsilon) \cdot \max(\diam_M(A),\diam_M(B)) \leq d_M(A,B).$$
\end{definition}

\begin{definition}[Pair decomposition]
    Let $V$ be a point set. A collection of pairs $\mathcal P = \{(A_i, B_i)\}_{i=1}^k$ is a pair decomposition of $V$ if for any two distinct points $u, v \in V$, there is exactly one pair $(A_i,B_i) \in \mathcal P$ such that either $(i)$ $u \in A_i$ and $v \in B_i$, or $(ii)$ $v \in A_i$ and $u \in B_i$. 
\end{definition}

\begin{definition}[Well-separated pair decomposition, WSPD]
    \label{definition:wspd}
    Let $V$ be a point set in a metric space $M$. A well-separated pair decomposition in $M$ with separation constant $1/\varepsilon$ (alternatively, a \mbox{$(1/\varepsilon)$-WSPD} in $M$) is a pair decomposition $\mathcal P = \{(A_i, B_i)\}_{i=1}^k$ such that every pair $(A_i, B_i)$ is $(1/\varepsilon)$-well-separated in $M$. 
\end{definition}

\begin{figure}[ht]
    \centering
    \includegraphics[width=0.8\textwidth]{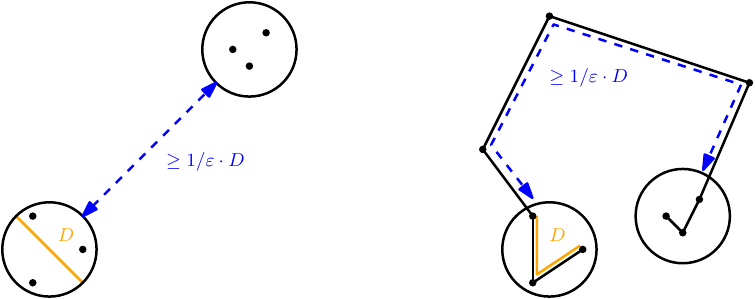}
    \caption{A well-separated pair under the Euclidean metric (left) and under a graph metric (right).
    }
    \label{figure:wspd}
\end{figure}

Recall also from Section~\ref{subsection:contributions} that another of our contributions is to construct an approximate distance oracle for low density graphs. Next, we define an approximate distance oracle over a set of vertices~$V$ and a general metric~$M$. 

\begin{definition}[Approximate distance oracle]
    \label{definition:ado}
    Let $V$ be a point set in a metric space~$M$. A $(1+\varepsilon)$-approximate distance oracle is a data structure that, given a pair of points $u,v \in V$, returns in $O(1)$ query time a $(1+\varepsilon)$-approximation of the distance $d_M(u,v)$. 
\end{definition}

For approximate distance oracles with $O(1)$ query time, it is common to assume the word RAM model of computation~\cite{DBLP:conf/stoc/GaoZ03,DBLP:journals/siamcomp/Har-PeledM06,DBLP:conf/soda/LeT22,DBLP:journals/jacm/Thorup04,DBLP:conf/stoc/ThorupZ01}. In this paper, we assume the unit-cost floating-point word RAM model of Har-Peled and Mendel~\cite{DBLP:journals/siamcomp/Har-PeledM06}. The floating-point model allows us to efficiently represent point sets that have large spread.

\begin{definition}[Unit-cost floating-point word RAM~\cite{DBLP:journals/siamcomp/Har-PeledM06}]
    \label{definition:word_ram}
    Suppose the input consists of $\poly(n)$ real numbers in $[-\Phi,-\Phi^{-1}] \cup [\Phi^{-1},\Phi]$ and let $t \in \mathbb N$ be an accuracy parameter. The unit cost floating-point word RAM model has words of length $O(\log n + \log \log \Phi + t)$ which can represent any floating point number of the form $u \cdot x \cdot 2^y$, where $u \in \{1,-1\}$, $x 2^{-t} \in [2^t,2^{t+1}) \cap \mathbb N$, and $y \in [-n^{O(1)} \log^{O(1)} \Phi, n^{O(1)} \log^{O(1)} \Phi] \cap \mathbb Z$. The model allows arithmetic, floor, logarithm, exponent, and table lookup operations in unit time.
\end{definition}

Similar to~\cite{DBLP:journals/siamcomp/Har-PeledM06}, we avoid rounding errors in the floating point model by assuming that our algorithm's approximation error is large relative to the model's accuracy parameter. In particular, we will assume that~$\varepsilon^{O(1)} > 2^{-t}$ and ignore subsequent numerical issues.

\section{Technical overview}
\label{section:technical_overview}

Our high-level algorithm is relatively simple. It consists of two steps.

The first step is to consider the vertices of the graph as points in $\mathbb R^2$, and to construct their WSPD under the $\mathbb R^2$ metric. We observe that well-separated pairs under the $\mathbb R^2$ metric may not be well-separated under the graph metric because the graph diameter of a set may be much larger than its Euclidean diameter, causing the distances within a set to be large relative to the distance between the sets. See Figure~\ref{figure:no_longer_well_separated} (left). 

The second step is to consider a pair~$(A_i,B_i)$ in the Euclidean WSPD, and to partition the sets~$A_i$ and~$B_i$ into clusters. In particular, we partition~$A_i$ into sets~$\{C_{ij}\}$ so that the graph diameter of $C_{ij}$ is at most the Euclidean diameter of $A_i$. We partition $B_i$ into sets $\{D_{ik}\}$ analogously. See Figure~\ref{figure:no_longer_well_separated} (right). Then, any pair of clusters $C_{ij}$ and $D_{ik}$ are well-separated with respect to the graph metric, since
$
    (1/\varepsilon) \cdot \max(\diam_G(C_{ij}),\diam_G(D_{ik}))
    \leq (1/\varepsilon) \cdot \max(\diam_2(A_i),\diam_2(B_i))
    \leq d_2(A_i,B_i) 
    \leq d_G(C_{ij},D_{ik})
$.
Here, $d_2$ refers to the Euclidean metric and $\diam_2$ refers to the Euclidean diameter. Collecting all pairs of clusters over all $(A_i, B_i)$ gives us the pairs $\{(C_{ij}, D_{ik})\}$. We return $\{(C_{ij}, D_{ik})\}$ as our Graph WSPD. This completes the overview of our algorithm. 

\begin{figure}[ht]
    \centering
    \includegraphics[width=0.8\textwidth]{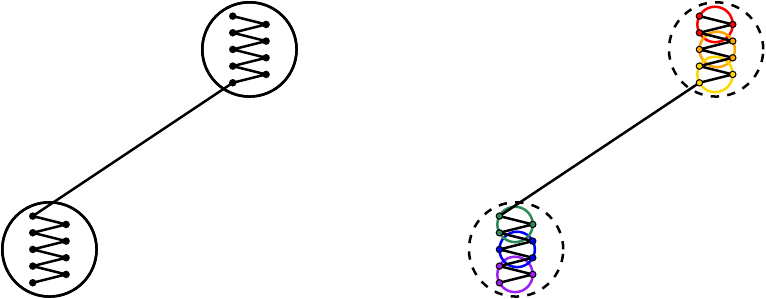}
    \caption{Left: A pair of sets that are well-separated in $\mathbb R^2$ may not be well-separated in the graph metric. Right: Clustering the sets in the Euclidean WSPD to reduce their graph diameter.}
    \label{figure:no_longer_well_separated}
\end{figure}

Surprisingly, this simple algorithm is all that is required to efficiently construct a WSPD of $O(n \lambda^2 \log n)$ size, provided that the graph is $\lambda$-low-density. The difficulty lies in the analysis. Note that so far, we have not yet used the fact that the graph is $\lambda$-low-density. In fact if~$G$ is not low density, then the number of pairs in~$\{(C_{ij}, D_{ik})\}$ may still be~$\Omega(n^2)$.

In Section~\ref{section:wspd_size}, we analyse the size of the Graph WSPD. In our analysis, we require a bound on the number of clusters~$\{C_{ij}\}$ that~$A_i$ partitions into. We obtain such a bound by combining a low density length bound by Driemel, Har-Peled and Wenk~\cite{DBLP:journals/dcg/DriemelHW12} with a clustering bound by Gudmundsson, Seybold and Wong~\cite{gudmundsson2024map}. See Fact~\ref{fact:low_density_implies_packed_weak} and Lemma~\ref{lemma:clustering_weak}, respectively. Summing up the bound on the number of clusters over all Euclidean WSPD pairs~$(A_i,B_i)$ and using the AM-GM inequality~\cite{wiki:amgm}, we get that the size of the Graph WSPD is at most $\lambda^2 \sum_i (|6A_i| + |6B_i|)$. Here, $|6A_i|$ denotes the number of points inside the square~$6S$, where~$S$ is the quadtree cell associated with the set $A_i$, and $6 S$ is the square~$S$ expanded about its centre by a factor of 6.

Then, in Section~\ref{subsection:bounded_spread}, we analyse the size of the Graph WSPD in the case where the spread~$\Phi$ is bounded, see Definition~\ref{definition:spread}. In this case, there exists a Euclidean WSPD~$\{(A_i,B_i)\}_{i=1}^k$ with a weight bound of $\sum_i (|A_i| + |B_i|) = O(n \log \Phi)$, see Har-Peled~\cite[Lemma~3.29]{har2011geometric}. In Theorem~\ref{theorem:bounded_spread_wspd}, we adapt the weight bound to also bound $\sum_i (|6A_i| + |6B_i|) = O(n \log \Phi)$, so that the overall size of the Graph WSPD is~$O(n \lambda^2 \log \Phi)$. 

In Section~\ref{subsection:unbounded_spread}, we remove the dependence on the spread~$\Phi$ when bounding the size of the Graph WSPD. Our approach is to define a new measure for Euclidean WSPD that can be bounded independent of the spread. In Definition~\ref{definition:semi-weight}, we define the semi-weight. Replacing weight with semi-weight in our approach introduces several technical challenges. The first challenge is that we require a new Euclidean WSPD construction and we require a bound on its semi-weight. We resolve these in Lemma~\ref{lemma:unbounded_spread_euclidean_wspd} by constructing a so-called semi-compressed quadtree, and with it, a new Euclidean WSPD with bounded semi-weight. The second challenge is that, although the expression $\lambda^2 \sum (|6A_i| + |6B_i|)$ can be upper bounded by the weight of a WSPD, the same expression can not be bounded by the semi-weight of a WSPD. We resolve this in Lemma~\ref{lemma:low_density_implies_packed_strong} by strengthen the low density length bound of Driemel, Har-Peled and Wenk~\cite{DBLP:journals/dcg/DriemelHW12}, which in turn reduces the number of clusters that $A_i$ can partition into, so as to match the semi-weight of the new Euclidean WSPD. Putting these ideas together, and with some additional work, we obtain the theorem below.

\begin{theorem}
[restate of Theorem~\ref{theorem:unbounded_spread_wspd}]
Let $G = (V,E)$ be a $\lambda$-low-density graph with $n$ vertices. For all $\varepsilon > 0$, there exists a $(1/\varepsilon)$-WSPD for $V$ in the graph metric $G$ with $O(n \lambda^2 \varepsilon^{-4} \log n)$ pairs.
\end{theorem}

In Section~\ref{section:wspd_algorithm}, we analyse the running time of our Graph WSPD algorithm. We start by analysing the running time to construct the Euclidean WSPD of bounded semi-weight in Section~\ref{subsection:semi_compressed_quadtree}. This step is relatively straightforward, and the overall construction time to compute both the semi-compressed quadtree and the Euclidean WSPD is~$O(n \log n)$ time. See Lemmas~\ref{lemma:semi-compressed_construction} and~\ref{lemma:euclidean_wspd_construction}.

In Section~\ref{subsection:net_tree}, we fill in some algorithmic gaps in our clustering algorithm, and then we analyse its running time. Our approach is to compute a global clustering which can later be refined into local clusterings for each of the Euclidean WSPD sets~$(A_i,B_i)$. Our global clustering combines the net-tree construction of Har-Peled and Mendel~\cite{DBLP:journals/siamcomp/Har-PeledM06}, and the approximate greedy graph clustering algorithm of Eppstein, Har-Peled and Sidiropoulos~\cite{DBLP:journals/jocg/EppsteinHS20}. In Lemma~\ref{lemma:net_tree}, we show how to adapt the algorithm of~\cite{DBLP:journals/jocg/EppsteinHS20} to construct a net-tree for a low density graph in $O(\lambda n \log^2 n)$ expected time.

In Section~\ref{subsection:putting_it_together}, we put it all together to construct the Graph WSPD. For each set~$A_i$ in the Euclidean WSPD, we obtain a local clustering for the node~$A_i$ by locating the nodes in the net-tree that, are on the $i(A_i)$-level, and, lie inside the bounding box of~$2A_i$. Here, the~$i(A_i)$-level of the net-tree guarantees that the clusters have radius at most $\diam(A_i)$ and each cluster centre is at least $\diam(A_i)/2$ apart. We can perform the local clustering efficiently by preprocessing the net-tree for three-dimensional orthogonal range searching queries, see Theorem~\ref{theorem:wspd_construction}. This completes the overview of the running time analysis.

\begin{theorem}[restate of Theorem~\ref{theorem:wspd_construction}]
    \label{theorem:wspd_construction_0}
    Let $G = (V,E)$ be a $\lambda$-low-density graph with $n$ vertices. For all $\varepsilon > 0$, there is a $(1/\varepsilon)$-WSPD for $V$ in the graph metric $G$ with $O(n \lambda^2 \varepsilon^{-4} \log n)$ pairs. The construction takes $O(n \lambda^2 \varepsilon^{-4} \log^2 n)$ expected time.
\end{theorem}

In Section~6, we construct a membership oracle that determines which Graph WSPD pair a pair of vertices belongs to, and an approximate distance oracle that determines the approximate graph distance between a pair of vertices. For our oracles, we require the word RAM model of computation, see Definition~\ref{definition:word_ram}.

In Section~\ref{subsection:membership_oracle}, we construct a membership oracle for the Graph WSPD. Given vertices~$a$ and~$b$, we first compute the Euclidean WSPD pair~$(A,B)$ where~$a \in A$ and~$b \in B$. We can query~$(A,B)$ from our semi-compressed quadtree in~$O(1)$ time, see Corollary~\ref{corollary:congruent_and_maximal}. Next, we compute the cluster~$C_j \subseteq A$ that~$a$ belongs to. The cluster centre $c_j \in C_j$ is the $i(A)$-level ancestor of~$a$ in the net-tree. We perform a weighted $i(A)$-level ancestor query of~$a$ by solving a new special case of the weighted ancestor query problem, see Lemma~\ref{lemma:weighted_ancestor_special_case}. Note that general weighted level ancestor queries are too slow, whereas other special cases~\cite{DBLP:conf/icalp/AlstrupH00,DBLP:journals/tcs/BadkobehCKP22,DBLP:conf/swat/BilleNP24,DBLP:journals/talg/KociumakaKRRW20,DBLP:conf/focs/LeW21} can not be applied. Putting it together, we can query the clusters~$a \in C_j$ and~$b \in D_k$ in~$O(1)$ query time.

Finally, in Section~\ref{subsection:approximate_distance_oracle}, we construct an approximate distance oracle for the low density graph. We follow the approach of Gao and Zhang~\cite{DBLP:conf/stoc/GaoZ03}, which is to construct an exact distance oracle and a lookup table. We use the separator theorem of Le and Than~\cite{DBLP:conf/soda/LeT22} to construct the exact distance oracle, and we use our membership oracle to construct the lookup table. The size of the data structure is equal to the size of the lookup table, which is~$O(n \lambda^2 \log n)$. All query operations take~$O(1)$ time. We state our final result below.

\begin{theorem}[restate of Theorem~\ref{theorem:ado_construciton}]
    \label{theorem:ado_construciton_0}
    Let $G=(V,E)$ be a $\lambda$-low-density graph with~$n$ vertices. For all $\varepsilon > 0$, there is an approximate distance oracle of $O(n \lambda^2 \varepsilon^{-4} \log n)$ size, that, given vertices $u,v \in V$, returns in $O(1)$ time a $(1+\varepsilon)$-approximation of $d_G(u,v)$. The preprocessing time is $O(n \sqrt n \lambda ^3 \varepsilon^{-4} \log n)$ expected.
\end{theorem}

This completes the overview of the main results of our paper.

\section{WSPD of near-linear size}
\label{section:wspd_size}

In this section, we will prove that for a low density graph metric, there exists a WSPD of near-linear size. First, we prove a bound on the WSPD size that has a logarithmic dependence on the spread, and then we improve the bound by removing the dependence on the spread.

\subsection{Bounded spread case}
\label{subsection:bounded_spread}

In this subsection, we assume the graph vertices have bounded spread.

\begin{definition}[Spread]
    \label{definition:spread}
    Let $V$ be a set of $n$ points in $\mathbb R^2$, and let $d_2(\cdot,\cdot)$ denote the Euclidean metric. The spread $\Phi$ of $V$ is defined to be
    \[
        \Phi = \frac {\max_{u,v \in V} d_2(u,v)} {\min_{u,v \in V} d_2(u,v)}.
    \]
\end{definition}

Recall that the first step of our algorithm is to construct a WSPD of the graph vertices with respect to the metric~$\mathbb R^2$. For any point set in~$\mathbb R^2$, there is a Euclidean WSPD of linear size~\cite{DBLP:journals/jacm/CallahanK95}. Moreover, if the point set has bounded spread, then there is a Euclidean WSPD with low weight~\cite{har2011geometric}. Next, we state the properties of the Euclidean WSPD with low weight.

\begin{definition}[Weight of a WSPD]
    \label{definition:weight}
    Let $\mathcal P = \{(A_i, B_i)\}_{i=1}^k$ be a pair decomposition. The weight of $\mathcal P$ is defined as $\sum_{i=1}^k (|A_i| + |B_i|)$, where $|X|$ denotes the number of elements in the set $X$.
\end{definition}

\begin{restatable}[Lemma~3.29 in Har-Peled's book~\cite{har2011geometric}]{fact}{bswspd}
\label{fact:bounded_spread_euclidean_wspd}
Let $V$ be a set of $n$ points in $\mathbb R^2$ with spread $\Phi$. For all $\varepsilon > 0$, there is a \mbox{$(1/\varepsilon)$-WSPD} for $V$ in the metric $\mathbb R^2$ with $O(n \varepsilon^{-2})$ pairs. Moreover,
\begin{enumerate}[noitemsep,label=(\alph*)]
    \item An uncompressed quadtree on~$V$ has $O(n)$ nodes and $O(\log \Phi)$ levels.
    \item The WSPD consists of pairs of well-separated quadtree nodes.
    \item Each quadtree node appears at most $O(\varepsilon^{-2})$ times in the WSPD.
    \item Each point in $V$ participates in at most one quadtree node per level, and in at most $O(\log \Phi)$ quadtree nodes in total.
    \item The WSPD has $O(n \varepsilon^{-2} \log \Phi)$ weight.
\end{enumerate}
\end{restatable}

Recall that the second step is to partition the sets~$A_i$ and~$B_i$ into clusters, for all pairs $(A_i,B_i)$ in the Euclidean WSPD. In particular, we partition~$A_i$ into sets~$\{C_{ij}\}$ so that the graph diameter of $C_{ij}$ is at most the Euclidean diameter of $A_i$. Similarly, we partition~$B_i$ into sets~$\{D_{ik}\}$ so that the graph diameter of~$D_{ik}$ is at most the Euclidean diameter of~$B_i$. By the end of this subsection, we will prove that~$C_{ij}$ and~$D_{ik}$ are~$(1/\varepsilon)$-well-separated. Collecting all pairs of clusters over all $(A_i, B_i)$ gives us the pairs $\{(C_{ij}, D_{ik})\}$. We return $\{(C_{ij}, D_{ik})\}$ as our Graph WSPD. 

Next, we analyse the size of our Graph WSPD. We require a bound on the number of clusters~$\{C_{ij}\}$ that~$A_i$ partitions into. Our bound on the number of clusters consists of two ingredients. The first ingredient is Fact~\ref{fact:low_density_implies_packed_weak}, which is a length bound for low density graphs. Informally, the length bound relates the sum of the lengths of the edges in~$A_i$ to the number of nearby vertices. See Fact~\ref{fact:low_density_implies_packed_weak} and Figure~\ref{figure:length_bound}. Note that in Lemma~\ref{lemma:low_density_implies_packed_strong} of this paper, we will prove a stronger version of Fact~\ref{fact:low_density_implies_packed_weak}.

\begin{fact}[Lemma~4.9 in Driemel, Har-Peled and Wenk~\cite{DBLP:journals/dcg/DriemelHW12}]
\label{fact:low_density_implies_packed_weak}
Let $G = (V,E)$ be a $\lambda$-low-density graph in $\mathbb R^2$. Let $S$ be a square with side length $\ell$ and let $3S$ be the scaling of $S$ by a factor of 3 about its centre. Let $L$ be the total length of $E \cap S$ and let $m$ be the number of vertices in $V \cap 3S$. Then $(\frac L {\lambda \ell})^2 = O(m)$. 
\end{fact}

\begin{figure}[ht]
    \centering
    \includegraphics[width=0.7\textwidth]{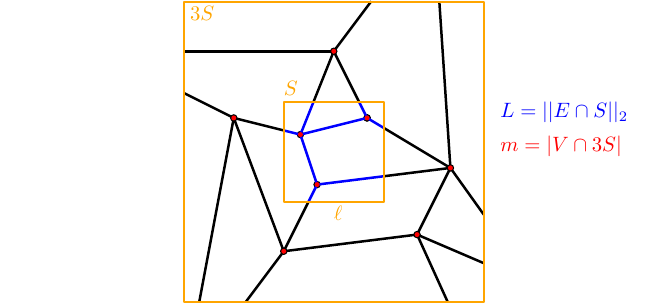}
    \caption{The length bound relates the sum of the lengths of edges in~$S$ with the number of points in~$3S$.}
    \label{figure:length_bound}
\end{figure}

The second ingredient is a clustering bound which we will prove in Lemma~\ref{lemma:clustering_weak}. This clustering bound is similar to the bound of Gudmundsson, Seybold and Wong~\cite[Lemma~17]{gudmundsson2024map}, however, their analysis applies to~$c$-packed graphs whereas our analysis applies to the broader class of~$\lambda$-low-density graphs. We can show that the number of clusters~$\{C_{ij}\}$ that~$A_i$ partitions into is at most the square root of the number of points nearby to~$A_i$. Our proof relies on the fact that the diameters of the clusters are at most the side length of~$A_i$, the distance between cluster centres are at least half the side length of~$A_i$, and on the length bound given in Fact~\ref{fact:low_density_implies_packed_weak}. We state our clustering bound in Lemma~\ref{lemma:clustering_weak}, see also Figure~\ref{figure:clustering_bound}.

\begin{lemma}
\label{lemma:clustering_weak}
Let $G = (V,E)$ be a connected $\lambda$-low-density graph. Let $S$ be a square with side length~$\ell$ and let~$6S$ be the scaling of $S$ by a factor of 6 about its centre. There exists a partition of the vertices $V\cap S$ into $K = O(\lambda \sqrt{|V \cap 6S|})$ clusters $\{C_1, C_2, \ldots, C_K\}$ so that $\diam_G(C_i) \leq \ell$ for all $C_i$. Here, $\diam_G(C_i)$ denotes the diameter of $C_i$ under the graph metric~$G$.
\end{lemma}

\begin{figure}[ht]
    \centering
    \includegraphics[width=0.7\textwidth]{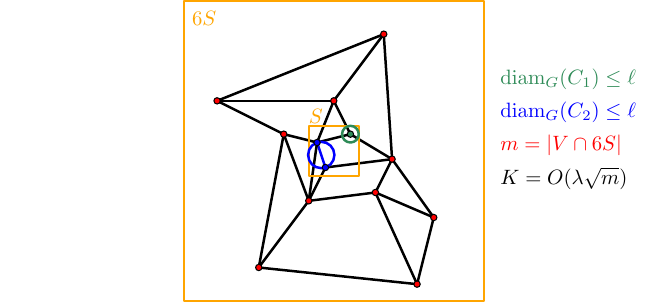}
    \caption{The vertices of $V \cap S$ can be divided into $O(\lambda \sqrt m)$ clusters, where $m = |V \cap 6S|$ and each cluster has diameter at most~$\ell$.}
    \label{figure:clustering_bound}
\end{figure}

\begin{proof}
Iteratively construct the clusters $C_i$. Pick an arbitrary unassigned vertex in $V \cap S$ and assign it as a cluster centre $c_i$. Assign $c_i$ to the cluster $C_i$, plus all unassigned vertices $u \in V \cap S$ satisfying $d_G(u,c_i) \leq \ell/2$. Continue until all vertices of~$V \cap S$ are assigned. Let the clusters be $\{C_1, \ldots, C_K\}$. Next, we bound~$K$.

Fix an arbitrary vertex~$v \in V$. For $1 \leq i \leq K$, define~$\pi_i$ to be the shortest path in $G$ from~$v$ to~$c_i$. Note that~$\pi_i$ is guaranteed to exist since~$G$ is connected. Define~$\sigma_i$ to be the portion of $\pi_i$ within distance~$\ell/4$ to~$c_i$. Note that~$\sigma_i$ is a subpath and does not necessarily end at vertices. We have~$\sigma_i \subset 2S$ because $c_i \in V \cap S$ and the length of~$\sigma_i$ is at most~$\ell/4$. For~$i < j$, we have~$\sigma_i$ and~$\sigma_j$ disjoint, otherwise we would have $d_G(c_i,c_j) \leq \ell/2$ and $c_j \in C_i$, which is impossible. The length of~$\sigma_i$ is~$\ell/4$ unless $d_G(v,c_i) \leq \ell/4$, which can happen at most once. Therefore, the sum of the lengths of edges in $2S$ is at least the sum of the lengths of $\sigma_i$ for $1 \leq i \leq K$, which is at least $(K-1) \ell /4$. By Fact~\ref{fact:low_density_implies_packed_weak}, $(\frac{(K-1)\ell}{4 \lambda \ell})^2 = O(|V \cap 6S|)$. Rearranging, we get $K = O(\lambda \sqrt{|V \cap 6S|})$, as required.
\end{proof}

Equipped with a bound on the number of clusters that each Euclidean WSPD set partitions into, we are now ready to bound the size of the Graph WSPD. We formally state the WSPD algorithm and its size bound in Theorem~\ref{theorem:bounded_spread_wspd}.

\begin{theorem}
\label{theorem:bounded_spread_wspd}
Let $G=(V,E)$ be a $\lambda$-low-density graph with $n$ vertices, and let $\Phi$ be the spread of~$V$. For all $\varepsilon > 0$, there exists a $(1/\varepsilon)$-WSPD for $V$ in the graph metric $G$ with $O(n \lambda^2 \varepsilon^{-2} \log \Phi)$ pairs. 
\end{theorem}

\begin{proof}
Using Fact~\ref{fact:bounded_spread_euclidean_wspd}, construct a Euclidean WSPD with separation constant~$(1/\varepsilon)$ to obtain a pair decomposition $\{(A_i, B_i)\}$, where $A_i$ and $B_i$ are point sets with associated quadtree nodes.  Lemma~\ref{lemma:clustering_weak} states that $A_i$ can be partitioned into $O(\lambda \sqrt{|6A_i|})$ clusters $\{C_{ij}\}$, and that $B_i$ can be partitioned into $O(\lambda \sqrt{|6B_i|})$ clusters $\{D_{ik}\}$. Collecting all pairs of clusters over all $(A_i, B_i)$ gives us the pairs $\{(C_{ij}, D_{ik})\}$. Since $\{C_{ij}\}$ and $\{D_{ik}\}$ are partitions of~$A_i$ and $B_i$, we get that $\{(C_{ij}, D_{ik})\}$ is a pair decomposition of $V \times V$. It remains to prove that the pair decomposition $\{(C_{ij}, D_{ik})\}$ is well-separated and has linear size.

We will prove that $(C_{ij}, D_{ik})$ is well-separated. Let $\ell(A_i)$ denote the side length of the quadtree node containing $A_i$. Then,
\begin{align*}
    (1/\varepsilon) \cdot \max(\diam_G(C_{ij}),\diam_G(D_{ik}))
    &\leq (1/\varepsilon) \cdot \max(\ell(A_i),\ell(B_i)) \\ 
    &\leq d_2(A_i,B_i) \\
    &\leq d_G(C_{ij},D_{ik})
\end{align*}
where the first inequality comes from Lemma~\ref{lemma:clustering_weak}, the second inequality comes from the quadtree nodes for~$A_i$ and~$B_i$ being well-separated, and the third inequality comes from $d_2(u,v) \leq d_G(u,v)$, $C_{ij} \subseteq A_i$ and $D_{ik} \subseteq B_i$.

We will prove that the pair decomposition $\{(C_{ij}, D_{ik})\}$ has linear size. Since $A_i$ is partitioned into $O(\lambda \sqrt{|6A_i|})$ clusters and $B_i$ is partitioned into $O(\lambda \sqrt{|6B_i|})$ clusters, the total number of pairs is bounded above by $\sum_i O(\lambda^2 \sqrt{|6A_i|\cdot |6B_i|})$. But $\sqrt{|6A_i|\cdot |6B_i|} \leq  \frac 1 2 (|6A_i|+ |6B_i|)$ by the AM-GM~\cite{wiki:amgm} inequality, so the total number of pairs is at most $\sum_i O(\lambda^2 (|6A_i| + |6B_i|))$. We will bound the expression $\sum_i (|6A_i| + |6B_i|)$ in the paragraph below. Our bound is, up to a constant factor, the same as the weight bound for the Euclidean WSPD in Fact~\ref{fact:bounded_spread_euclidean_wspd}, which is $\sum_i (|A_i| + |B_i|) = O(n \varepsilon^{-2} \log \Phi)$. 

It remains to bound $\sum_i (|6A_i| + |6B_i|)$. We rely heavily on Fact~\ref{fact:bounded_spread_euclidean_wspd} and its proof, see Lemma~3.29 in~\cite{har2011geometric}. Recall from Fact~\ref{fact:bounded_spread_euclidean_wspd}b that $A_i$ and $B_i$ are quadtree nodes, and from Fact~\ref{fact:bounded_spread_euclidean_wspd}c that each quadtree node appears at most~$O(\varepsilon^{-2})$ times in the WSPD. Therefore, $6A_i$ can appear at most~$O(\varepsilon^{-2})$ times in the summation $\sum_i(|6A_i| + |6B_i|)$. Recall from Fact~\ref{fact:bounded_spread_euclidean_wspd}d that each point in~$V$ participates in at most one quadtree node per level. Expanding each quadtree node by a factor of~6, each point may now participate in at most~36 expanded quadtree nodes per level. Recall from Fact~\ref{fact:bounded_spread_euclidean_wspd}a that there are $O(\log \Phi)$ quadtree levels. So each point can participate in $O(\log \Phi)$ expanded quadtree nodes, each of which can appear at most $O(\varepsilon^{-2})$ times in the expression~$\sum_i(|6A_i| + |6B_i|)$. Therefore, each point contributes at most $O(\varepsilon^{-2} \log \Phi)$ times to the summation $\sum_i(|6A_i| + |6B_i|)$. As there are $n$ points, we get $\sum_i(|6A_i| + |6B_i|) = O(n \varepsilon^{-2} \log \Phi)$. 

As a result, we obtain a final bound of $\sum_i O(\lambda^2 \sqrt{|6A_i|\cdot |6B_i|}) = O( n \lambda^2 \varepsilon^{-2} \log \Phi)$ on the total number of pairs in the WSPD, as required.
\end{proof}

\subsection{Unbounded spread case}

\label{subsection:unbounded_spread}

In this subsection, we bound the size of the WSPD independently of the spread~$\Phi$. In particular, we replace the $O(\log \Phi)$ dependence in Theorem~\ref{theorem:bounded_spread_wspd} with an $O(\varepsilon^{-2} \log n)$ dependence. See Theorem~\ref{theorem:unbounded_spread_wspd} for the final result of this section. Note that under our input model, the spread~$\Phi$ may be unbounded with respect to the size of the input measured in words, see Definition~\ref{definition:word_ram}.

The main challenge in the unbounded spread case is that the Euclidean WSPD's weight in Fact~\ref{fact:bounded_spread_euclidean_wspd} can be~$\Omega(n^2)$, so it can no longer be used. In particular, the Euclidean WSPD's weight bound is critical to the step $\sum_i \lambda^2  (|6A_i| + |6B_i|) = O(n \lambda^2 \varepsilon^{-2} \log \Phi)$ when bounding the size of the Graph WSPD. If we want to follow a similar approach to Section~\ref{subsection:bounded_spread}, we will need an alternative to the weight bound in Fact~\ref{fact:bounded_spread_euclidean_wspd}.
 
To this end, we introduce a new measure: the semi-weight of a Euclidean WSPD. Unlike the weight, we can always construct a Euclidean WSPD with low semi-weight, for any point set in the Euclidean plane. Next, we define the semi-size of a set and semi-weight of a WSPD.

\begin{definition}[Semi-size]
\label{definition:semi-size}
Let $A$ be a set of points in $\mathbb R^2$ and let $n \geq |A|$. We say $\semisize(A) \leq k$ if there exist $k$ discs that cover all points in $A$, and each disc has a diameter of $\frac 1 n \diam_2(A)$.
\end{definition}

\begin{definition}[Semi-weight]
\label{definition:semi-weight}
Let $V$ be a set of $n$ points in $\mathbb R^2$ and let $\mathcal P = \{(A_i,B_i)\}_{i=1}^k$ be a pair decomposition of $V$. The semi-weight of $\mathcal P$ is defined as $\sum_{i=1}^k (\semisize(A_i) + \semisize(B_i))$. 
\end{definition}

\begin{figure}[ht]
    \centering
    \includegraphics[width=0.28\textwidth]{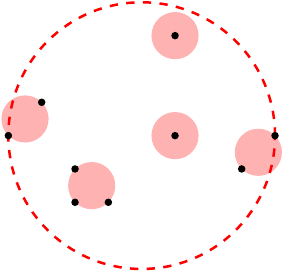}
    \caption{We have $\semisize(A) \leq k$ if~$A$ can be covered by~$k$ discs, all with diameter at most $\frac 1 n \diam_2(A)$.}
    \label{figure:semisize}
\end{figure}

Note that the semi-size of a set is at most its size, since each point can be covered by an individual disc. See Figure~\ref{figure:semisize}. Similarly, the semi-weight of a WSPD is at most its weight. For our purposes, we do not need to compute the semi-weight of a set exactly, as it suffices to upper bound the semi-size of a set. In particular, by computing only an upper bound, we can avoid computing a minimum geometric set cover which is known to be NP-hard.

Our goal is to adapt the approach in Section~\ref{subsection:bounded_spread} to use the semi-weight of a Euclidean WSPD instead of its weight. We summarise our approach. We start by constructing a new Euclidean WSPD with bounded \mbox{semi-weight} in Lemma~\ref{lemma:unbounded_spread_euclidean_wspd} to replace Fact~\ref{fact:bounded_spread_euclidean_wspd}. Then we prove a stronger length bound for low density graphs, and we replace Fact~\ref{fact:low_density_implies_packed_weak} with Lemma~\ref{lemma:low_density_implies_packed_strong}. The main technical contributions of this section are Lemmas~\ref{lemma:unbounded_spread_euclidean_wspd} and~\ref{lemma:low_density_implies_packed_strong}. In Lemma~\ref{lemma:clustering_strong} we improve the clustering bound of Lemma~\ref{lemma:clustering_weak} to depend on the semi-size of the set instead of its size, and finally in Theorem~\ref{theorem:unbounded_spread_wspd} we prove the bound on the WSPD size. Refer to Table~\ref{table:adapting}.

\begin{table}[ht]
    \centering
    \renewcommand{\arraystretch}{1.5}
    \setlength{\tabcolsep}{10pt}
    \begin{tabular}{|c|c|c|c|}
        \hline
        & Bounded spread & Unbounded spread & New proof\\
        \hline
        \hline
        Euclidean WSPD & Fact~\ref{fact:bounded_spread_euclidean_wspd} & Lemma~\ref{lemma:unbounded_spread_euclidean_wspd} & Yes
        \\
        \hline
        Length bound for low density graphs & Fact~\ref{fact:low_density_implies_packed_weak} & Lemma~\ref{lemma:low_density_implies_packed_strong} & Yes
        \\
        \hline
        Clustering of quadtree nodes& Lemma~\ref{lemma:clustering_weak} & Lemma~\ref{lemma:clustering_strong} & -
        \\
        \hline
        Graph WSPD & Theorem~\ref{theorem:bounded_spread_wspd} & Theorem~\ref{theorem:unbounded_spread_wspd} & -
        \\
        \hline
    \end{tabular}
    \caption{Adapting the approach from the bounded spread case to the unbounded spread case.}
    \label{table:adapting}
\end{table}

As stated previously, we construct a new Euclidean WSPD with bounded semi-weight. For comparison, we restate the Euclidean WSPD in Fact~\ref{fact:bounded_spread_euclidean_wspd}, and then discuss our approach. 

\bswspd*

A natural approach would be to replace Fact~\ref{fact:bounded_spread_euclidean_wspd}a with a compressed quadtree. However, doing so creates several issues for Facts~\ref{fact:bounded_spread_euclidean_wspd}b-d. The first issue is that, instead of having $O(\log \Phi)$ levels, the compressed quadtree may have $\Omega(n)$ levels. Therefore, the bound of $O(\log \Phi)$ in Fact~\ref{fact:bounded_spread_euclidean_wspd}d no longer applies. The second issue is that in a compressed quadtree, a child node may be much smaller than its parent. As a result, the well-separated pairs in Fact~\ref{fact:bounded_spread_euclidean_wspd}b may consist of pairs of quadtree nodes of vastly different sizes. The packing argument to prove Fact~\ref{fact:bounded_spread_euclidean_wspd}c requires that paired quadtree nodes have similar sizes. Therefore, the previous proof of Fact~\ref{fact:bounded_spread_euclidean_wspd}c no longer holds on compressed quadtrees, and a quadtree node may appear $\Omega(n)$ times in the WSPD.

We resolve these issues in Lemma~\ref{lemma:unbounded_spread_euclidean_wspd}. To resolve the first issue, we leverage the semi-size and semi-weight definitions. We cover each quadtree node with a disc. Rather than count the number of times a point contributes to the weight $\sum_i (|A_i| + |B_i|)$, we instead count the number of times its covering disc contributes to the semi-weight $\sum_i(\semisize(A_i) + \semisize(B_i))$. To resolve the second issue, we augment the compressed quadtree to force the WSPD pairs to be of the same size. We call this new quadtree a semi-compressed quadtree. The packing argument can now be applied to the semi-compressed quadtree to prove that each quadtree node appears at most $O(\varepsilon^{-2})$ times in the WSPD. 

We construct the new Euclidean WSPD in Lemma~\ref{lemma:unbounded_spread_euclidean_wspd}. Note that Lemmas~\ref{lemma:unbounded_spread_euclidean_wspd}a, \ref{lemma:unbounded_spread_euclidean_wspd}b, \ref{lemma:unbounded_spread_euclidean_wspd}c, \ref{lemma:unbounded_spread_euclidean_wspd}d, \ref{lemma:unbounded_spread_euclidean_wspd}e are analogous to Facts~\ref{fact:bounded_spread_euclidean_wspd}a, \ref{fact:bounded_spread_euclidean_wspd}b, \ref{fact:bounded_spread_euclidean_wspd}c, \ref{fact:bounded_spread_euclidean_wspd}d, \ref{fact:bounded_spread_euclidean_wspd}e. 

\begin{lemma}
\label{lemma:unbounded_spread_euclidean_wspd}
Let $V$ be a set of $n$ points in $\mathbb R^2$. For all $\varepsilon > 0$, there is a $(1/\varepsilon)$-WSPD for $V$ in the metric $\mathbb R^2$ with $O(n\varepsilon^{-2})$ pairs. Moreover,
\begin{enumerate}[noitemsep,label=(\alph*)]
    \item There is a semi-compressed quadtree with $O(n \varepsilon^{-2})$ nodes.
    \item The WSPD consists of pairs of well-separated quadtree nodes.
    \item Each quadtree node appears at most $O(\varepsilon^{-2})$ times in the WSPD.
    \item Each quadtree node has an associated covering disc. Each covering disc contributes to the semi-size of at most one quadtree node per level, and to at most $O(\log n)$ levels in total.
    \item The WSPD has $O(n\varepsilon^{-4}\log n)$ semi-weight.
\end{enumerate}
\end{lemma}

\begin{proof}
First, we will construct the semi-compressed quadtree in Lemma~\ref{lemma:unbounded_spread_euclidean_wspd}a as follows. Construct a regular compressed quadtree on~$V$. Construct a $(2/\varepsilon)$-WSPD using this quadtree. For the remainder of this proof, when we say ``$(2/\varepsilon)$-WSPD'' we refer to this WSPD. The textbook~\cite{har2011geometric} provides an algorithm for constructing a WSPD from a compressed quadtree. Next, we will use the \mbox{$(2/\varepsilon)$-WSPD} to uncompress and recompress the quadtree and obtain a semi-compressed quadtree. Let~$(A_i,B_i)$ be a well-separated pair in the \mbox{$(2/\varepsilon)$-WSPD}. Let~$\parent(A_i)$, $\parent(B_i)$ be the parents of $A_i$ and $B_i$, respectively, in the compressed quadtree. For now, uncompress the paths $\{\parent(A_i) \to A_i\}$ and $\{\parent(B_i) \to B_i\}$. Let $A'$ and $B'$ be any uncompressed quadtree nodes along the paths $\{\parent(A_i) \to A_i\}$ and $\{\parent(B_i) \to B_i\}$. See Figure~\ref{figure:recompress}. Define $\double(A')$ to be the uncompressed quadtree node that contains $A'$ and has two times its side length. Define $\double(B')$ similarly. We call $A'$ and $B'$ are congruent if they have the same side length. We call $A'$ and $B'$ $(1/\varepsilon)$-maximal if~$(A',B')$ is $(1/\varepsilon)$-well-separated, but $(\double(A'),\double(B'))$ is not $(1/\varepsilon)$-well-separated. We claim that along the paths $\{\parent(A_i) \to A_i\}$ and $\{\parent(B_i) \to B_i\}$, there always exists some~$A'$ and~$B'$ that are congruent and $(1/\varepsilon)$-maximal. We will prove the claim in the next paragraph. We recompress the path to $\{\parent(A_i) \to A' \to A_i\}$ and $\{\parent(B_i) \to B' \to B_i\}$. Note that the only difference to the compressed quadtree is that $A'$ and $B'$ are now included as quadtree nodes. Repeat this process for all pairs $(A_i,B_i)$ in the $(2/\varepsilon)$-WSPD. The total number of nodes in the semi-compressed quadtree is now $O(n \varepsilon^{-2})$, as there are~$O(n \varepsilon^{-2})$ pairs in the $(2/\varepsilon)$-WSPD. This completes the construction of the semi-compressed quadtree in Lemma~\ref{lemma:unbounded_spread_euclidean_wspd}a.  

\begin{figure}[ht]
    \centering
    \includegraphics[width=0.55\textwidth]{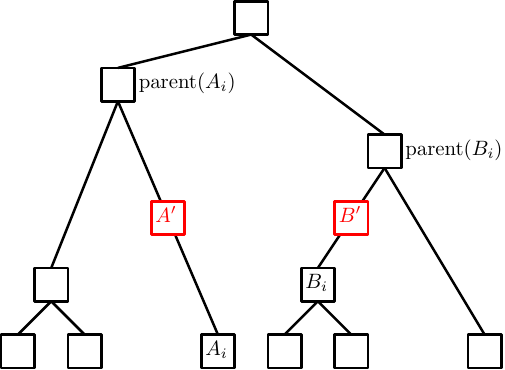}
    \caption{For all pairs~$(A_i,B_i)$ in the $(2/\varepsilon)$-WSPD, we reintroduce node~$A'$ on the path $\{A_i \to \parent(A_i)\}$ and node~$B'$ on the path $\{B_i \to \parent(B_i)\}$, where $(A',B')$ are congruent and $(1/\varepsilon)$-maximal.}
    \label{figure:recompress}
\end{figure}

To prove the correctness of Lemma~\ref{lemma:unbounded_spread_euclidean_wspd}a, it remains to show that there always exists~$A'$ and $B'$ that are congruent and $(1/\varepsilon)$-maximal. For now, uncompress the paths $\{\parent(A_i) \to A_i\}$ and $\{\parent(B_i) \to B_i\}$. Define $A_{\min}$ and $B_{\min}$ to be quadtree nodes that are on the paths $\{\parent(A_i) \to A_i\}$ and $\{\parent(B_i) \to B_i\}$ respectively, that are congruent to the largest of~$A_i$ and~$B_i$. Define $A_{max}$ and $B_{max}$ to be quadtree nodes on the paths $\{\parent(A_i) \to A_i\}$ and $\{\parent(B_i) \to B_i\}$ respectively, that are congruent to the smallest of~$\parent(A_i)$ and~$\parent(B_i)$. Suppose for the sake of contradiction that the claim is false. The negation of the claim implies that for any pair of congruent quadtree nodes~$A'$ and~$B'$ on the paths~$\{A_{max} \to A_{min}\}$ and~$\{B_{max} \to B_{min}\}$ respectively, if $A'$ and $B'$ are $(1/\varepsilon)$-well-separated, then $\double(A_i)$ and $\double(B_i)$ are also $(1/\varepsilon)$-well-separated. In the first bullet point below, we will prove that $A_{min}$ and~$B_{min}$ are~$(1/\varepsilon)$-well-separated. By the negation of the claim, $\double(A_{min})$ and~$\double(B_{min})$ are also~$(1/\varepsilon)$-well-separated. This continues inductively up the uncompressed paths~$\{A_{max} \to A_{min}\}$ and~$\{B_{max} \to B_{min}\}$ until $\double(A_{max})$ and $\double(B_{max})$ are~$(1/\varepsilon)$-well-separated. In the second bullet point below, we will prove that $\double(A_{max})$ and $\double(B_{max})$ are not~$(1/\varepsilon)$-well-separated, which is a contradiction, and completing the proof. It remains only to show that $A_{min}$ and~$B_{min}$ are~$(1/\varepsilon)$-well-separated (first bullet point) and that $\double(A_{max})$ and $\double(B_{max})$ are not~$(1/\varepsilon)$-well-separated (second bullet point). 

\begin{itemize}
    \item Proof that $A_{min}$ and~$B_{min}$ are~$(1/\varepsilon)$-well-separated. Recall that $(A_i,B_i)$ are $(2/\varepsilon)$-well-separated. Recall that $A_{\min}$ and $B_{\min}$ are defined to be congruent to the largest of~$A_i$ and~$B_i$, and are quadtree nodes on the paths $\{\parent(A_i) \to A_i\}$ and $\{\parent(B_i) \to B_i\}$ repectively. Therefore,
\begin{align*}
    (1/\varepsilon) \cdot &\max(\diam_2(A_{min}),\diam_2(B_{min})) \\
    &< ((2/\varepsilon) - 1) \cdot \max(\diam_2(A_{min}),\diam_2(B_{min}))  \\
    &\leq (2/\varepsilon) \cdot \max(\diam_2(A_{i}),\diam_2(B_{i})) - \max(\diam_2(A_{min}),\diam_2(B_{min}))  \\
    & \leq d_2(A_i,B_i) - \max(\diam_2(A_{min}),\diam_2(B_{min})) \\
    & \leq d_2(A_{min},B_{min}),
\end{align*}
    where the first inequality comes from~$0 < \varepsilon < 1$, the second inequality comes from $A_i \subseteq A_{min}$ and $B_i \subseteq B_{min}$, and the third inequality comes from $A_i$ and $B_i$ being $(2/\varepsilon)$-well-separated. The fourth inequality comes from the observation that at most one quadtree node $\{A_i,B_i\}$ needs to expand to obtain $\{A_{min},B_{min}\}$. The other quadtree node stays the same. As one of the quadtree nodes expands, the distance between the two sets decreases by at most the diameter of the expansion, which is at most $\max(\diam_2(A_{min}),\diam_2(B_{min}))$. Putting this together, we get that $A_{min}$ and $B_{min}$ are $(1/\varepsilon)$-well-separated, as required.
    
    \item Proof that $\double(A_{max})$ and $\double(B_{max})$ are not~$(1/\varepsilon)$-well-separated. We will first show that $A_{max}$ and $B_{max}$ are not~$(2/\varepsilon)$-well-separated. Prior to $(A_i, B_i)$ being considered as a candidate pair for the $(2/\varepsilon)$-WSPD, the WSPD construction algorithm must have either considered $(\parent(A_i),B_i)$ or $(A_i,\parent(B_i))$. Without loss of generality, suppose that $(\parent(A_i),B_i)$ was considered. It follows that~$\parent(A_i)$ cannot be larger than~$\parent(B_i)$, so $\parent(A_i) = A_{max}$. Moreover, since $(\parent(A_i),B_i)$ was considered but not added to the~$(2/\varepsilon)$-WSPD, we must have that~$\parent(A_i)$ and~$B_i$ are not $(2/\varepsilon)$-well-separated. Therefore,
\begin{align*}
    d_2(A_{max},B_{max}) &= d_2(parent(A_i),B_{max}) \\
    &\leq d_2(parent(A_i),B_i) \\
    &\leq (2/\varepsilon) \cdot \max(\diam_2(parent(A_i)), \diam_2(B_i)) \\
    &= (2/\varepsilon) \cdot \diam_2(parent(A_i)) \\
    &= (2/\varepsilon) \cdot \max(\diam_2(A_{max}), \diam_2(B_{max})),
\end{align*}
    where the first equality comes from $parent(A_i) = A_{max}$, the second inequality comes from $B_i \subseteq B_{max}$, the third inequality comes from $\parent(A_i)$ and $B_i$ not being $(2/\varepsilon)$-well-separated, the fourth equality comes from $\parent(A_i)$ being no smaller than $B_i$, and the fifth equality comes from $A_{max}$ and $B_{max}$ being congruent to $\parent(A_i)$. Therefore, $A_{max}$ and $B_{max}$ are not~$(2/\varepsilon)$-well-separated. But now,
\begin{align*}
    d_2(\double(A_{max}),\double(B_{max}))  
    &\leq d_2(A_{max}, B_{max}) \\
    &\leq (2/\varepsilon) \cdot \max(\diam_2(A_{max}), \diam_2(B_{max})) \\
    &= (1/\varepsilon) \cdot \max(\diam_2(\double(A_{max})), \diam_2(\double(B_{max}))),
\end{align*}
    where the first inequality comes from $A_{max} \subset \double(A_{max})$ and $B_{max} \subset \double(B_{max})$, the second inequality comes from $A_{max}$ and $B_{max}$ not being~$(2/\varepsilon)$-well-separated, and the third equality comes from the definition of $\double(\cdot)$. Therefore, $\double(A_{max})$ and $\double(B_{max})$ are not~$(1/\varepsilon)$-well-separated, as required.
\end{itemize}
This completes the proof of the claim that there always exists~$A'$ and $B'$ that are congruent and $(1/\varepsilon)$-maximal, thereby completing the proof of correctness of Lemma~\ref{lemma:unbounded_spread_euclidean_wspd}a. 

Second, while we construct the semi-compressed quadtree, we will obtain the new WSPD in Lemma~\ref{lemma:unbounded_spread_euclidean_wspd}b. The new WSPD is essentially the same as the $(2/\varepsilon)$-WSPD, except that we force WSPD pairs to be of the same size. In particular, the new WSPD consists of the congruent and $(1/\varepsilon)$-maximal pairs $(A',B')$ along the paths $\{\parent(A_i) \to A_i\}$ and $\{\parent(B_i) \to B_i\}$, for every well-separated pair $(A_i,B_i)$ in the $(2/\varepsilon)$-WSPD. Even though~$A'$ may be significantly larger than~$A_i$, we retain only the points in~$A_i$. In other words, the points associated with~$A'$ are exactly the points that were associated with~$A_i$. In particular, this implies that the pairs $(A',B')$ will remain a pair decomposition of $V \times V$. By the definition of~$A'$ and~$B'$, we have that $(A',B')$ is a $(1/\varepsilon)$-well-separated pair. This completes the construction of the new WSPD in Lemma~\ref{lemma:unbounded_spread_euclidean_wspd}b.

Third, we will apply a packing argument to prove Lemma~\ref{lemma:unbounded_spread_euclidean_wspd}c. Fix a quadtree node $A'$ in the semi-compressed quadtree, and suppose that $B'$ is some other quadtree node so that $(A',B')$ is a $(1/\varepsilon)$-well-separated pair in the WSPD in Lemma~\ref{lemma:unbounded_spread_euclidean_wspd}b. Then by our construction, $A'$ and $B'$ are the same size, and $\double(A')$ and $\double(B')$ are not $(1/\varepsilon)$-well-separated. Define $\ell(A') = \ell(B')$ to be the side length of $A'$ and $B'$. Then $d_2(A',B') \leq d_2(\double(A'),\double(B')) + 2 \ell(A') \leq (2/\varepsilon + 2) \cdot \ell(A')$. Therefore, $B'$ must be of side length $\ell(A')$ and within a distance of $O(\varepsilon^{-1} \cdot \ell(A'))$ from $A'$, so there are $O(\varepsilon^{-2})$ possible choices for~$B'$. Therefore, the quadtree node~$A'$ appears at most $O(\varepsilon^{-2})$ times in the new WSPD, completing the proof of Lemma~\ref{lemma:unbounded_spread_euclidean_wspd}c. Note that Lemma~\ref{lemma:unbounded_spread_euclidean_wspd}c immediately implies that the number of pairs in the new WSPD is~$O(n \varepsilon^{-2})$.

Fourth, we will perform a counting argument to prove Lemma~\ref{lemma:unbounded_spread_euclidean_wspd}d. For a semi-compressed quadtree node~$A'$ we can bound its semi-size by traversing down the semi-compressed tree. Create a priority queue which initially consists only of $A'$. Pop the first element~$Q$ from the priority queue and construct its circumscribed disc. If the circumscribed disc has diameter at most $\frac 1 n \diam_2(A')$, add~$Q$ to the set $Cover(A')$. If the circumscribed disc has diameter greater than $\frac 1 n \diam_2(A')$, add the children of~$Q$ to the priority queue. Continue until the priority queue is empty. This completes the construction of $Cover(A')$. By Definition~\ref{definition:semi-size}, we have $\semisize(A') \leq |Cover(A')|$. Consider a quadtree node $Q$. If $Q \in Cover(A')$, then~$Q$ must be a descendent of~$A'$. Therefore,~$Q$ contributes to at most one quadtree node~$A'$ per quadtree level. Moreover, if $Q \in Cover(A')$, then $\parent(Q)$ was considered for $Cover(A')$ but not included. Therefore, $\diam(\parent(Q)) \geq \frac 1 n \diam(A')$. Therefore, $A'$ is an ancestor of $Q$ satisfying $\diam(Q) \leq \diam(A') \leq n \cdot \diam(\parent(Q))$. However, there are only $O(\log n)$ ancestors of $Q$ satisfying this size requirement, since the size of an ancestor at least doubles per level. Therefore, every quadtree node~$Q$ contributes to the semi-size of at most one quadtree node per level, and to at most $O(\log n)$ levels in total, as required by Lemma~\ref{lemma:unbounded_spread_euclidean_wspd}d. 

Fifth, we prove Lemma~\ref{lemma:unbounded_spread_euclidean_wspd}e by bounding the semi-weight of the new WSPD. By Definition~\ref{definition:semi-weight}, the semi-weight of the WSPD is $\sum_i (\semisize(A_i) + \semisize(B_i))$, where $A_i$ and $B_i$ ranges over all WSPD pairs $(A_i, B_i)$. By Lemma~\ref{lemma:unbounded_spread_euclidean_wspd}c, each quadtree node~$A'$ appears at most $O(\varepsilon^{-2})$ times in the above expression, so the semi-weight is upper bounded by $O(\varepsilon^{-2} \cdot \sum \semisize(A'))$, where $A'$ ranges over all quadtree nodes in the semi-compressed quadtree. We upper bound $\semisize(A')$ by charging its cost to the elements of $Cover(A')$. Each quadtree node~$Q$ is charged at most $O(\log n)$ times by Lemma~\ref{lemma:unbounded_spread_euclidean_wspd}d, and there are~$O(n \varepsilon^{-2})$ quadtree nodes by Lemma~\ref{lemma:unbounded_spread_euclidean_wspd}a. Therefore, $\sum \semisize(A') = O(n \varepsilon^{-2} \log n)$, so the semi-weight of the new WSPD is $O(n \varepsilon^{-4} \log n)$, completing the proof of Lemma~\ref{lemma:unbounded_spread_euclidean_wspd}e.
\end{proof}

As stated previously, the next step is replace Fact~\ref{fact:low_density_implies_packed_weak} with Lemma~\ref{lemma:low_density_implies_packed_strong}. Note that Lemma~\ref{lemma:low_density_implies_packed_strong} generalises Fact~\ref{fact:low_density_implies_packed_weak}, and it provides length bound that depends only on the semi-size of the set of nearby points, rather than the number of nearby points. See Figure~\ref{figure:length_the_second}.

\begin{lemma}
\label{lemma:low_density_implies_packed_strong}
Let $G = (V,E)$ be a $\lambda$-low-density graph in $\mathbb R^2$ with $n$ vertices. Let $S$ be a square with side length $\ell$ and let $3S$ be the scaling of $S$ by a factor of 3 about its centre. Let $L$ be the total length of $E \cap S$ and let $m = \semisize(V \cap 3S)$. Then $(\frac L {\lambda \ell})^2 = O(m)$. 
\end{lemma}

\begin{figure}[ht]
    \centering
    \includegraphics[width=0.78\textwidth]{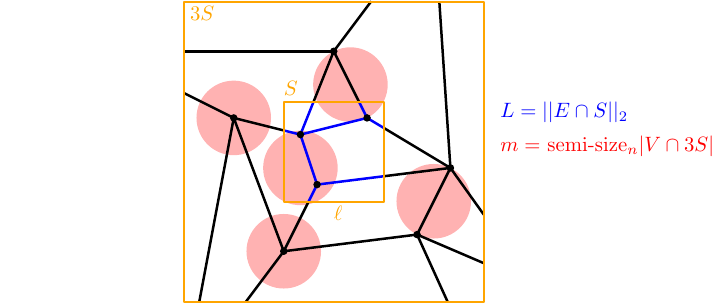}
    \caption{We prove a stronger version of Fact~\ref{fact:low_density_implies_packed_weak}, where the~$L$ of edges in $E \cap S$ is bounded relative to the semi-size of $V \cap 3S$, instead of its size.}
    \label{figure:length_the_second}
\end{figure}

\begin{proof}
For $e \in E$, define $w_S(e)$ to be the length of the subsegment $e \cap S$. Our goal is to bound $\sum_{e \in E} w_S(e) = L$. Note that for all $e$, we have $w_S(e) \leq \sqrt 2 \ell = \diam(S)$. We divide~$E$ into two subset, i.e.~$E = E' \cup E''$.
\begin{align*}
    E' &= \{e \in E: \mbox{both endpoints of $e$ lies in $V \cap 3S$}\} \\
    E'' &= \{e \in E: \mbox{at least one endpoint of $e$ lies outside $3S$}\}
\end{align*}

We bound the sum of the lengths of edges in $E''$ first. If an edge~$e$ does not intersect $S$, then $w_S(e) = 0$, so the only relevant edges are those that intersect $S$. Since edges in $E''$ have an endpoint outside $3S$, all relevant edges have length at least~$\ell$. Cover the square $S$ with four discs, each with radius~$\ell$. Each relevant edge intersects a covering disc. Each covering disc can only intersect $\lambda$ relevant edges, by Definition~\ref{definition:low-density}. Therefore, there are at most $4 \lambda$ relevant edges. Each relevant edge satisfies $w_S(e) \leq \sqrt 2 \ell$. Putting this together, $\sum_{e \in E''} w_S(e) \leq (4 \lambda) \cdot (\sqrt 2 \ell) = O(\lambda \ell)$. 

Next, we bound the sum of the edges in $E'$. To this end, we partition $E'$ into three subsets, i.e.~$E' = E_1 \cup E_2 \cup E_3$. 
\begin{align*}
    E_1 &= \{e \in E': 0 \leq w_S(e) \leq  5\ell / n \} 
        \\
    E_2 &= \{e \in E': 5\ell / n \leq w_S(e) \leq  \ell / \sqrt{\semisize(3S)} \} 
        \\
    E_3 &= \{e \in E': \ell / \sqrt{\semisize(3S)} \leq w_S(e) \leq \sqrt 2 \ell \}
\end{align*}
We will bound $\sum_{e \in E_i} w_S(e)$ separately for $i=1,2,3$.

\begin{itemize}
    \item \emph{Bound for $E_1$.} There are $n$ vertices in $V$, and the degree of each vertex is at most $\lambda$. Therefore, $|E_1| \leq |E| \leq \lambda n$. But $w_S(e) \leq 5 \ell /n$ for all $e \in E_1$. So $\sum_{e \in E_1} w_S(e) \leq (\lambda n) \cdot (5 \ell / n) = O(\lambda \ell)$. 
    
    \item \emph{Bound for $E_2$.} We will use the fact that $E_2 \in E'$, so both endpoints of $e$ lie in~$V \cap 3S$. 
    
    By Definition~\ref{definition:semi-size}, all points in $V \cap 3S$ can be covered by $\semisize(V \cap 3S)$ number of discs, where each disc has diameter $\frac 1 n \diam(3S) = 3 \sqrt 2 \ell / n < 5 \ell / n$. Note that these discs cover all endpoints of edges in $E_2$. 
    
    Each edge $e \in E_2$ intersects a covering disc. Each semi-size disc can intersect at most $\lambda$ edges in $E_2$, by Definition~\ref{definition:low-density}. Therefore, there are at most $\semisize(V \cap 3S) \lambda$ edges in $E_2$. But each edge $e \in E_2$ satisfies $w_S(e) \leq \ell / \sqrt{\semisize(3S)}$. Therefore, 
    \[
        \sum_{e \in E_2} w_S(e) \leq (\semisize(3S) \lambda) \cdot (\ell / \sqrt{\semisize(3S)}) = O(\sqrt{\semisize(3S)} \lambda \ell).
    \]
    \item \emph{Bound for $E_3$.} We further divide $E_3$ into subsets $E_{3,j} = \{e \in E_3: \ell/2^j \leq w_S(e) \leq 2\ell/2^j\}$. Here, $0 \leq j \leq k = \lceil \log_2(\sqrt{\semisize(3S)}) \rceil$, so that $E_3 \subseteq \cup_{j=0}^k E_{3,j}$.
    
    Cover the square $S$ with $2^{2j+2}$ discs, each with radius $\ell / 2^j$. Each edge in $E_{3,j}$ intersects a covering disc. Each covering disc can only intersect $\lambda$ edges in $E_{3,j}$, by Definition~\ref{definition:low-density}. Therefore, there are at most $2^{2j+2} \lambda$ edges in $E_{3,j}$. But $w_S(e) \leq 2 \ell / 2^j$ for $e \in E_{3,j}$. Therefore, we have $\sum_{e \in E_{3,j}} w_S(e) \leq (2^{2j+2} \lambda) \cdot (2\ell / 2^j) = 2^{j+3} \lambda \ell$. 

    Let $L_j = \sum_{e \in E_{3,j}} w_S(e)$. Then we are required to bound $\sum_{e \in E_3} w_S(e) = \sum_{j=0}^k L_j$, where $k = \lceil \log_2(\sqrt{\semisize(3S)}) \rceil$. 
    
    But,
    \[
        \sum_{j=0}^k L_j \leq (2^4 + 2^5 + \ldots + 2^{k+3}) \lambda \ell \leq 2^{k+4} \lambda \ell = O(\sqrt{\semisize(3S)} \lambda \ell).
    \]
    So $\sum_{e \in E_{3}} w_S(e) = O(\sqrt{\semisize(3S)} \lambda \ell)$. 
\end{itemize}
Combining the bounds for $E''$, $E_1$, $E_2$ and $E_3$, we get $\sum_{e \in E} w_S(e) = O(\sqrt{\semisize(3S)} \lambda \ell)$. Substituting in the variables $L = \sum_{e \in E} w_S(e)$ and $m = \semisize(3S)$, we get $L = O(\sqrt m \lambda \ell)$. Therefore, $\frac L {\lambda \ell} = O(\sqrt m)$ and $(\frac L {\lambda \ell})^2 = O(m)$, as required.
\end{proof}

This completes the proofs of Lemmas~\ref{lemma:unbounded_spread_euclidean_wspd} and~\ref{lemma:low_density_implies_packed_strong}, which recall from Table~\ref{table:adapting} are the new proofs in this section. We complete the proof for the unbounded spread case by following the same approach as the bounded spread case. We will adapt the approach by replacing Fact~\ref{fact:bounded_spread_euclidean_wspd} with Lemma~\ref{lemma:unbounded_spread_euclidean_wspd}, and Fact~\ref{fact:low_density_implies_packed_weak} with Lemma~\ref{lemma:low_density_implies_packed_strong}. 

Given the new length bound, we can cluster any quadtree node in an analogous way to Lemma~\ref{lemma:clustering_weak}. 

% \begin{figure}[ht]
%     \centering
%     \includegraphics[width=0.75\textwidth]{figures/clustering_the_second.pdf}
%     \caption{Caption}
%     \label{fig:enter-label}
% \end{figure}

\begin{lemma}
\label{lemma:clustering_strong}
Let $G = (V,E)$ be a connected $\lambda$-low-density graph. Let $S$ be a square with side length~$\ell$ and let~$6S$ be the scaling of $S$ by a factor of 6 about its centre. There exists a partition of the vertices $V\cap S$ into $O(\lambda \sqrt{\semisize(V \cap 6S)})$ clusters $\{C_1, C_2, \ldots\}$ so that $\diam_G(C_i) \leq \ell$ for all $C_i$. Here, $\diam_G(C_i)$ denotes the diameter of $C_i$ under the graph metric~$G$.
\end{lemma}

\begin{proof}
The proof is the same as Lemma~\ref{lemma:clustering_weak}, but we replace Fact~\ref{fact:low_density_implies_packed_weak} with Lemma~\ref{lemma:low_density_implies_packed_strong}. Due to the better length bound in Lemma~\ref{lemma:low_density_implies_packed_strong}, we now only have $O(\lambda \sqrt{\semisize(V \cap 6S)})$ clusters instead of $O(\lambda \sqrt{|V \cap 6S|})$ clusters.
\end{proof}

Finally, we bound the total size of the Graph WSPD in an analogous way to Theorem~\ref{theorem:bounded_spread_wspd}.

\begin{theorem}
\label{theorem:unbounded_spread_wspd}
Let $G = (V,E)$ be a $\lambda$-low-density graph with $n$ vertices. For all $\varepsilon > 0$, there exists a $(1/\varepsilon)$-WSPD for $V$ in the graph metric $G$ with $O(n \lambda^2 \varepsilon^{-4} \log n)$ pairs.
\end{theorem}

\begin{proof}
The proof is the same as Theorem~\ref{theorem:bounded_spread_wspd}, but we~$(i)$ replace $|6A_i|$ with $\semisize(6A_i)$, $(ii)$ replace Fact~\ref{fact:bounded_spread_euclidean_wspd} with Lemma~\ref{lemma:unbounded_spread_euclidean_wspd}, and $(iii)$ replace Lemma~\ref{lemma:clustering_weak} with Lemma~\ref{lemma:clustering_strong}. This yields a total number of pairs that is $\sum_i O(\lambda^2 \sqrt{\semisize(6A_i) \cdot \semisize(6B_i)})$. Next, we bound the square root term in the same way as Theorem~\ref{theorem:bounded_spread_wspd}. The sum of the square root terms is at most the semi-weight of the Euclidean WSPD, which by Lemma~\ref{lemma:unbounded_spread_euclidean_wspd}e is~$O(n \varepsilon^{-4} \log n)$. Putting this together, the total number of pairs in the Graph WSPD is~$O(n \lambda^2 \varepsilon^{-4} \log n)$, as required. 
\end{proof}

\section{Constructing the WSPD}
\label{section:wspd_algorithm}

In the previous section, we showed the existence of a Graph WSPD and bounded its size. In this section, we will provide an algorithm for constructing the Graph WSPD. We will follow the same outline as Section~\ref{subsection:unbounded_spread}, while filling in the algorithmic gaps. Note that for our analysis, we assume the floating-point word RAM model of computation in Definition~\ref{definition:word_ram}.

In Section~\ref{subsection:semi_compressed_quadtree}, we construct a semi-compressed quadtree. In Section~\ref{subsection:net_tree}, we construct a net-tree. Finally, in Section~\ref{subsection:putting_it_together}, we put these pieces together to construct the Graph WSPD.  Sections~\ref{subsection:semi_compressed_quadtree},~\ref{subsection:net_tree} and~\ref{subsection:putting_it_together} correspond to algorithms for Lemmas~\ref{lemma:unbounded_spread_euclidean_wspd},~\ref{lemma:clustering_strong}, and~\ref{theorem:unbounded_spread_wspd} respectively.

\subsection{Semi-compressed quadtree}
\label{subsection:semi_compressed_quadtree}

We previously stated the construction of the semi-compressed quadtree in the proof of Lemma~\ref{lemma:unbounded_spread_euclidean_wspd}a. We restate the construction for convenience. Construct a compressed quadtree on $V$, and with it, a Euclidean $(2/\varepsilon)$-WSPD. For each well-separated pair $A_i$ and $B_i$ in the Euclidean $(2/\varepsilon)$-WSPD, construct a pair of congruent and $(1/\varepsilon)$-maximal quadtree nodes $(A',B')$ so that $A_i \subseteq A' \subseteq \parent(A_i)$, and $B_i \subseteq B' \subseteq \parent(B_i)$. Recall that $A'$ and $B'$ are congruent if they have the same side length, and are $(1/\varepsilon)$-maximal if they are $(1/\varepsilon)$-well-separated but $\double(A')$ and $\double(B')$ are not $(1/\varepsilon)$-well-separated, where $\double(A')$ is the quadtree node that contains $A'$ and has double its side length. In Lemma~\ref{lemma:unbounded_spread_euclidean_wspd}, we proved that a pair of congruent and $(1/\varepsilon)$-maximal quadtree nodes~$(A',B')$ are guaranteed to exist, but we did not provide an algorithm to compute~$(A',B')$. We will fill this gap in Lemma~\ref{lemma:quadtree_ancestor}. Given~$A'$ and~$B'$, we simply insert the quadtree node $A'$ into the compressed quadtree between $\parent(A_i)$ and $A_i$, and similarly for $B'$. This completes our restatement of the semi-compressed quadtree construction.

To complete the algorithm, it remains only to fill the gap of computing the congruent and $(1/\varepsilon)$-maximal quadtree nodes~$(A',B')$, given a well-separated pair~$(A_i,B_i)$ in the Euclidean $(2/\varepsilon)$-WSPD.

\begin{lemma}
    \label{lemma:quadtree_ancestor}
    Given a quadtree and a pair of $(2/\varepsilon)$-well-separated quadtree nodes $A_i$ and $B_i$, one can compute in $O(1)$ time a pair of congruent and $(1/\varepsilon)$-maximal quadtree nodes $A'$ and $B'$ that are ancestors of $A_i$ and $B_i$ respectively.
\end{lemma}

\begin{proof}
    We already know that $A'$ and $B'$ are guaranteed to exist, by Lemma~\ref{lemma:unbounded_spread_euclidean_wspd}a. Next, we claim that the side lengths of $A'$ and $B'$ must be in the range $[c_1 \varepsilon \cdot d_2(A_i,B_i), c_2 \varepsilon \cdot d_2(A_i,B_i)]$, for constants~$c_1$ and~$c_2$. As there are only a constant number of quadtree sizes in this range, there are only $O(1)$ ancestors of $A_i$ and $B_i$ to check. Constructing these ancestors takes $O(1)$ time, as we can perform logarithm and floor operations in unit time. Checking if a candidate pair of ancestors is $(1/\varepsilon)$-well-separated takes $O(1)$ time. Finally, we return the largest of these pairs of congruent ancestors that are $(1/\varepsilon)$-well-separated.

    It remains only to prove the claim that $\ell(A') \in [c_1\varepsilon \cdot d_2(A_i,B_i), c_2\varepsilon \cdot d_2(A_i,B_i)]$, where $\ell(\cdot)$ denotes the side length of a quadtree, $c_1 = 1/(6\sqrt 2)$ and $c_2 = 1/\sqrt 2$. We have:
    \begin{align*}
        c_1 \varepsilon \cdot d_2(A_i,B_i) 
        &\leq c_1 \varepsilon \cdot (d_2(\double(A'),\double(B')) + 2 \cdot \max(\diam(\double(A')),\diam(\double(B')))\\ 
        &\leq (c_1 + 2 c_1 \varepsilon) \cdot \max(\diam(\double(A')),\diam(\double(B'))) \\
        &= (1/(6 \sqrt 2)) \cdot (1+2\varepsilon) \cdot (2 \sqrt 2) \cdot \ell (A') \\
        &\leq \ell (A') \\
        &= (1/\sqrt 2) \cdot \max(\diam(A'),\diam(B')) \\
        &\leq (1/\sqrt 2) \cdot \varepsilon \cdot d_2(A',B') \\
        &\leq c_2 \varepsilon \cdot d_2(A_i,B_i)
    \end{align*}
    where the first line comes from $A_i \subseteq A'$ and $B_i \subseteq B'$, the second line comes from $\double(A')$ and $\double(B')$ not being $(1/\varepsilon)$-well-separated, the fourth line comes from $0 < \varepsilon < 1$, the sixth line comes from $A'$ and $B'$ being $(1/\varepsilon)$-well-separated, and the seventh line comes from $A_i \subseteq A'$ and $B_i \subseteq B'$.
\end{proof}

Finally, we will analyse the construction time of the semi-compressed quadtree. Constructing the compressed quadtree takes $O(n \log n)$ time and constructing the Euclidean $(2/\varepsilon)$-WSPD takes an additional $O(n \varepsilon^{-2})$ time~\cite{har2011geometric}. For each of the $O(n \varepsilon^{-2})$ well-separated pairs $(A_i,B_i)$, we construct congruent and $(1/\varepsilon)$-maximal quadtree nodes $A'$ and $B'$, which takes $O(1)$ time by Lemma~\ref{lemma:quadtree_ancestor}. In total, all pairs $(A',B')$ can be computed in $O(n \varepsilon^{-2})$ time. The quadtree node~$A'$ can be inserted in between~$A_i$ and~$\parent(A_i)$ in $O(1)$ time. Therefore, all~$A'$ and~$B'$ can be inserted in $O(n \varepsilon^{-2})$ time. The overall running time is $O(n \log n + n \varepsilon^{-2})$. Thus, we obtain Lemma~\ref{lemma:semi-compressed_construction}.

\begin{lemma}
    \label{lemma:semi-compressed_construction}
    The semi-compressed quadtree in Lemma~\ref{lemma:unbounded_spread_euclidean_wspd}a can be constructed in $O(n \log n + n \varepsilon^{-2})$ time.
\end{lemma}

The Euclidean WSPD in Lemma~\ref{lemma:unbounded_spread_euclidean_wspd}b can be obtained by simply enumerating the well-separated pairs $(A',B')$. We obtain Lemma~\ref{lemma:euclidean_wspd_construction}.

\begin{lemma}
    \label{lemma:euclidean_wspd_construction}
    The Euclidean WSPD in Lemma~\ref{lemma:unbounded_spread_euclidean_wspd}b can be constructed in $O(n \log n + n \varepsilon^{-2})$ time.
\end{lemma}

This completes the construction of the semi-compressed quadtree, and the Euclidean WSPD with $O(n \varepsilon^{-2})$ pairs and $O(n \varepsilon^{-4} \log n)$ \mbox{semi-weight}.

\subsection{Clustering via a net-tree}
\label{subsection:net_tree}

The next algorithmic step is to cluster the Euclidean WSPD. In particular, for each pair~$(A,B)$ in the Euclidean WSPD, recall that the approach in Section~\ref{subsection:unbounded_spread} is to partition~$A$ into $O(\lambda \sqrt{\semisize(V \cap 6A)})$ clusters and~$B$ into $O(\lambda \sqrt{\semisize(V \cap 6B)})$ clusters so that each cluster has a small graph diameter relative to its quadtree node size. Lemma~\ref{lemma:clustering_strong} guarantees the existence of such a clustering, and Lemma~\ref{lemma:clustering_weak} provides a na\"ive algorithm to compute the clustering. 

We restate the clustering algorithm in Lemma~\ref{lemma:clustering_weak} for convenience. Let~$S$ be a square in~$\mathbb R^2$ with side length~$\ell$. We iteratively construct the clusters~$C_i$. Pick an arbitrary unassigned vertex in $V \cap S$ and assign it as a cluster centre $c_i$. Assign to the cluster $C_i$ all unassigned vertices $u \in V \cap S$ satisfying $d_G(u,c_i) \leq \ell/2$. This completes the restatement of the clustering algorithm. To prove Lemmas~\ref{lemma:clustering_weak} and~\ref{lemma:clustering_strong}, the only requirements of the clustering algorithm are that each cluster has diameter at most~$\ell$, and the distance between different cluster centres is at least~$\ell/2$.

Applying the na\"ive clustering algorithm in Lemma~\ref{lemma:clustering_weak} for every quadtree node in the Euclidean WSPD results in an algorithm with a running time at least $O(\sum_{i=1}^k (|A_i| + |B_i|))$, where $\{(A_i,B_i)\}_{i=1}^k$ is the Euclidean WSPD. This summation is equal to the weight of the Euclidean WSPD, which may be~$\Omega(n^2)$. To obtain a faster algorithm, we will instead compute a set of global clusterings. Then we can retrieve the local clustering of a quadtree node by clipping the relevant global clustering to the quadtree node's square. We describe how to retrieve the local clusterings from the global clusterings in Section~\ref{subsection:putting_it_together}. In this subsection we focus on computing the global $r$-clustering, where $r$ ranges over all quadtree node sizes.

A net-tree can be thought of as a heirarchy of $r$-clusterings for exponential values of~$r$. Har-Peled and Mendel~\cite{DBLP:journals/siamcomp/Har-PeledM06} introduced the net-tree and used it to obtain a wide range of results in doubling metrics, for example, approximate nearest neighbour data structures, spanners, and approximate distance oracles. The net-tree has since been used to construct coresets for clustering in doubling metrics~\cite{DBLP:conf/focs/HuangJLW18}, fault-tolerant spanners in doubling metrics~\cite{DBLP:conf/stoc/Solomon14}, and approximate distance oracles for planar graphs~\cite{DBLP:conf/focs/LeW21}.

Eppstein, Har-Peled and Sidiropoulos~\cite{DBLP:journals/jocg/EppsteinHS20} use a set of $r$-clusterings to approximate the greedy clustering of a graphs. Their set of $r$-clusterings is inherently hierarchical, so it is straightforward to modify their $r$-clustering into a net-tree. The advantage of using the net-tree structure, instead of directly using the $r$-clusterings, is that we will be able to search the net-tree more efficiently. The remainder of Section~\ref{subsection:net_tree} will be dedicated to modifying the $r$-clustering in~\cite{DBLP:journals/jocg/EppsteinHS20} into a net-tree. 

We divide the remainder of the section into five steps. The first step is to review the clustering in Eppstein, Har-Peled and Sidiropoulos~\cite{DBLP:journals/jocg/EppsteinHS20} if~$\Phi$ is bounded. The second step is to modify this clustering into a net-tree. The third step is to review the clustering in Eppstein, Har-Peled and Sidiropoulos~\cite{DBLP:journals/jocg/EppsteinHS20} if~$\Phi$ is unbounded. The fourth step is to modify this clustering into a net tree. The fifth step is to state the properties of our net-tree.

First, we review the clustering in~\cite{DBLP:journals/jocg/EppsteinHS20} in the case where the spread~$\Phi$ is bounded. Let~$r_i = \Delta / (1+\varepsilon_1)^{i-1}$, where $1 \leq i \leq \lceil \log_{1+\varepsilon_1} \Phi \rceil$, $\Delta$ is the diameter of~$G$, and $\varepsilon_1$ is a constant. Later, we will set $\varepsilon = 1$, as we are interested in~$r_i$'s that halve per quadtree level. For sparse graphs, Eppstein, Har-Peled and Sidiropoulos~\cite{DBLP:journals/jocg/EppsteinHS20} provide a randomised, near-linear time algorithm for computing a sequence of cluster centres~$\{N_i\}$  satisfying the following three properties; see Fact~\ref{fact:net_tree_bounded_spread}.

\begin{itemize}[noitemsep]
    \item (Covering) The distance from any vertex $v \in G$ to a cluster centre in $N_i$ is at most $r_i$.
    \item (Packing) The distance between any two cluster centres in $N_i$ is at least $r_i$.
    \item (Inheritance) $N_i \subseteq N_{i+1}$ for all $1 \leq i \leq \lceil \log_{1+\varepsilon_1} \Phi \rceil - 1$. 
\end{itemize}

\begin{fact}[Lemma~2.5-2.7 in Eppstein, Har-Peled and Sidiropoulos~\cite{DBLP:journals/jocg/EppsteinHS20}] 
\label{fact:net_tree_bounded_spread}
Let $G = (V,E)$ be a graph with $n$ vertices, $m$ edges, diameter $\Delta$ and spread $\Phi$. Let $\varepsilon_1 > 0$ and let $r_i = \Delta / (1+\varepsilon_1)^{i-1}$. One can compute in $O(\varepsilon_1 m \log n \log \Phi)$ expected time, for $1 \leq i \leq \lceil \log_{1+\varepsilon_1} \rceil$, a sequence of $r_i$-cluster centres $\{N_i\}$ that satisfy the covering, packing and inheritance properties.
\end{fact}

Second, we turn the $r_i$-clustering in Fact~\ref{fact:net_tree_bounded_spread} into a net-tree in a straightforward way. We make the element of $N_1$ the root of the net-tree. For $2 \leq i \leq \lceil \log_{1+\varepsilon_1} \rceil$ the nodes on level~$i$ of the net-tree are the elements of~$N_i$. For a node $v \in N_{i}$, define $\parent(v) \in N_{i-1}$ to be the cluster centre of the $r_{i-1}$-cluster containing~$v$. All parents of $N_{i}$ can be computed in a single pass of Dijkstra's algorithm, starting simultaneously from the source vertices~$N_i$. The net-tree has $O(n)$ nodes per level and $O(\log \Phi)$ levels, so its total size is $O(n \log \Phi)$. Computing $\{N_i\}$ takes $O(\varepsilon_1^{-1} m \log n \log \Phi)$ expected time by Fact~\ref{fact:net_tree_bounded_spread}. Computing the parents of $N_i$ takes $O(m \log n)$ time per level, so $O(m \log n \log \Phi)$ in total. Therefore, the overall running time for computing the net-tree is $O(\varepsilon_1 (n + m) \log n \log \Phi)$ expected time.

In Figure~\ref{figure:net-tree1}, we provide an example of an $r_i$-clustering and its corresponding net-tree. We have the vertices~$N_1 = \{v_8\}$ marked as red, the vertices~$N_2 \setminus N_1 = \{v_3,v_4,v_8\}$ marked as orange, the vertices~$N_3 \setminus N_2 = \{v_1,v_9,v_{13}\}$ marked as green, the vertices~$N_4 \setminus N_3 = \{v_5,v_{12}\}$ marked as blue, and $N_5 \setminus N_4 = \{v_2,v_6,v_7,v_{10}\}$ marked as purple. There are $O(\log \Phi)$ levels, and each vertex can appear at most once in each level, so the total size is $O(n \log \Phi)$.

\begin{figure}[ht]
    \centering
    \includegraphics[width=\textwidth]{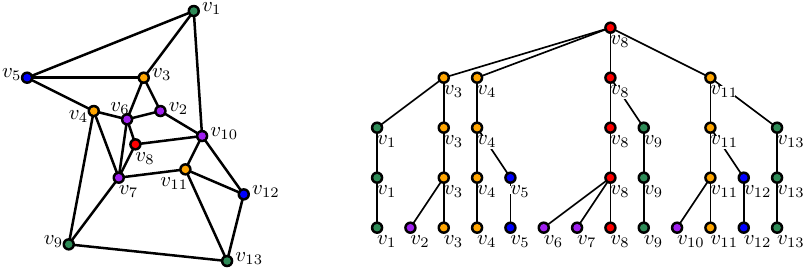}
    \caption{An $r_i$-clustering (left) and its corresponding net-tree (right). The colour of a vertex indicates the smallest~$i$ where it is an element of~$N_i$, and corresponds to the highest level it appears in the net-tree.}
    \label{figure:net-tree1}
\end{figure}

Third, we review the result of~\cite{DBLP:journals/jocg/EppsteinHS20} in the case where the spread~$\Phi$ is unbounded. In particular, they use a standard method~\cite{DBLP:journals/cjtcs/MendelS09} to eliminate the dependence on the spread. For a resolution~$r_i$, edges longer than~$n r_i$ can be ignored, and edges shorter than~$\varepsilon_1 r_i / n^2$ can be collapsed. Therefore, an edge is only active in $O(\varepsilon_1 \log (n / \varepsilon_1)$ resolutions. Let~$m_i$ be the number of active edges at resolution~$r_i$. The running time of computing the~$r_i$-clustering is~$O(m_i \log m_i)$ expected, by Lemma~2.4 in~\cite{DBLP:journals/jocg/EppsteinHS20}. Therefore, the sequence of~$r_i$-cluster centres~$\{N_i\}$ can be computed in $O(\sum_i m_i \log m_i)$ expected time, which is $O(\varepsilon_1^{-1} m \log n \log (n/\varepsilon_1))$ expected time in total.

\begin{fact}[Theorem~2.8 in Eppstein, Har-Peled and Sidiropoulos~\cite{DBLP:journals/jocg/EppsteinHS20}]
\label{fact:net_tree_unbounded_spread}
    Let $G = (V,E)$ be a graph with $n$ vertices, $m$ edges and diameter $\Delta$. Let $\varepsilon_1 > 0$ and let $r_i = \Delta / (1+\varepsilon_1)^{i-1}$. One can compute in $O(\varepsilon_1^{-1} m \log n \log (n/\varepsilon_1))$ expected time a sequence of $r_i$-cluster centres $\{N_i\}$ that satisfy the covering, packing and inheritance properties. Note that the sequence $\{N_i\}$ may skip over resolutions $r_i$ if the $r_i$-clustering does not change.
\end{fact}

Fourth, we turn the $r_i$-clustering in Fact~\ref{fact:net_tree_unbounded_spread} into a net-tree. We modify the net-tree construction from step two, so as to eliminate the dependence on the spread. Previously, a vertex could appear in all $O(\log \Phi)$ levels of the net-tree, so the overall size of the net-tree was~$O(n \log \Phi)$. Instead, we now only include one copy of each vertex in the net-tree. We refer to this new net-tree as the compressed net-tree. We make~$N_1$ the root of the compressed net-tree. We place a vertex~$v \in V$ into level~$i$ of the compressed net-tree if~$v \in N_i$ and~$v \not \in N_{i-1}$. For~$v$ on level~$i$, define $\parent(v)$ to be the cluster centre of the $r_{i-1}$-cluster containing~$v$. Note that the parent of a node in level $i$ may no longer be in level~$i-1$, but instead in level~$j$ for some $j < i$. This completes the construction of the compressed net-tree. The overall size of the net-tree is $O(n)$. The overall running time for constructing the net-tree is $O(\varepsilon_1^{-1} m \log n \log (n / \varepsilon_1))$ expected. 

In Figure~\ref{figure:net-tree2}, we provide an example of an $r_i$-clustering and its corresponding compressed net-tree. Similarly to Figure~\ref{figure:net-tree1}, we have ~$N_1 = \{v_8\}$,~$N_2 \setminus N_1 = \{v_3,v_4,v_8\}$,~$N_3 \setminus N_2 = \{v_1,v_9,v_{13}\}$,~$N_4 \setminus N_3 = \{v_5,v_{12}\}$, and $N_5 \setminus N_4 = \{v_6,v_7,v_{10}\}$. However, unlike in Figure~\ref{figure:net-tree1}, each vertex can appear at most once in the compressed net-tree, so its overall size is $O(n)$.

\begin{figure}[ht]
    \centering
    \includegraphics[width=\textwidth]{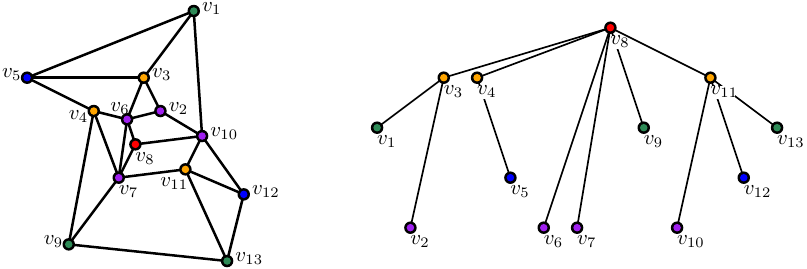}
    \caption{An $r_i$-clustering (left) and its corresponding compressed net-tree (right). The colour of a vertex indicates its level in the net-tree.}
    \label{figure:net-tree2}
\end{figure}

Fifth, we state the properties of our net-tree. For our purposes, it suffices that $\varepsilon_1 = 1$. To this end, we also set $r_i = \Delta/2^{i-1}$ where $\Delta$ is the side length of the root node in the semi-compressed quadtree. Recall that $N_i$ can be reconstructed from the $(\leq i)$-levels. We prove a useful property of our compressed net-tree; that the subtrees rooted at $N_i$ induces a $2r_i$-clustering.

\begin{lemma}
    \label{lemma:net_tree}
    Let $G = (V,E)$ be a graph with $n$ vertices and $m$ edges. Let $G$ be contained in a $\Delta \times \Delta$ square. Let $r_i = \Delta / 2^{i-1}$. Construct a sequence of $r_i$-cluster centres~$\{N_i\}$ using Fact~\ref{fact:net_tree_unbounded_spread}. Construct its compressed net-tree using the algorithm above. Let $N_i = \{s_{ij}\}$. Let $S_{ij}$ be the set of descendants of $s_{ij}$ with $(> i)$-level so that path from $s_{ij}$ to the descendant passes through no other nodes in~$N_i$. Then $S_{ij}$ satisfies:
    \begin{itemize}[noitemsep]
        \item (Covering) The distance from any vertex $v \in S_{ij}$ to  $s_{ij}$ is at most $2r_i$.
        \item (Packing) The distance between $s_{ij}$ and $s_{ik}$ is at least $r_i$, for any $j \neq k$. 
        \item (Partition) The sets $\{s_{ij} \cup S_{ij}\}$ form a partition of $V$. 
    \end{itemize}
    Moreover, the compressed net-tree can be computed in $O(m \log^2 n)$ expected time.
\end{lemma}

\begin{proof}
The packing property of $\{s_{ij}\}$ follows directly from the packing property of $N_i$ in Fact~\ref{fact:net_tree_unbounded_spread}. The partition property follows from $\{s_{ij}\}$ partitioning the $(\leq i)$-levels, and $\{S_{ij}\}$ partitioning the $(> i)$-levels. Next, we prove the covering property. Let the path from $s_{ij}$ to $v$ be $s_{ij} = u_0 \to \ldots \to u_k = v$. Then $u_1$ is on a level~$\geq i+1$, otherwise we would have $u_1 \in N_i$ which is a contradiction. It follows that $u_j$ is on a level $\geq i+j$. Since $\parent(u_j) = u_{j-1}$, we have that $d_G(u_j,u_{j-1}) \leq r_{i+j} = \Delta/2^{i+j-1}$. By the triangle inequality, $d_G(s_{ij},v) \leq \sum_{j=0}^k \Delta/2^{i+j-1} \leq \Delta/2^{i-2} = 2r_i$, as required. 

The running time for computing the net-tree is dominated by computing the sequence of cluster centres~$\{N_i\}$. The overall running time is $O(m \log^2 n)$ by setting $\varepsilon_1 = 1$ in Fact~\ref{fact:net_tree_unbounded_spread}.
\end{proof}

In Figure~\ref{figure:net-tree3}, we provide an example where we apply Lemma~\ref{lemma:net_tree}. In the example, we have~$i=2$ and $N_2 = \{v_3,v_4,v_8,v_{11}\}$. Therefore, $s_{2,1} = v_3$, $s_{2,2} = v_4$, $s_{2,3} = v_8$, and $s_{2,4} = v_{11}$. By considering the subtrees rooted at $s_{2,j}$ for $1 \leq j \leq 4$, we partition the compressed net-tree into four clusters~$S_{2,j}$ for $1 \leq j \leq 4$. In particular, the cluster rooted at $v_3$ is $S_{2,1} = \{v_1,v_2\}$, the cluster rooted at $v_4$ is $S_{2,2} = \{v_5\}$, the cluster rooted at $v_8$ is $S_{2,3} = \{v_6,v_7,v_9\}$, and the cluster rooted at $v_{11}$ is $\{v_{10},v_{12},v_{13}\}$.

\begin{figure}[ht]
    \centering
    \includegraphics[width=\textwidth]{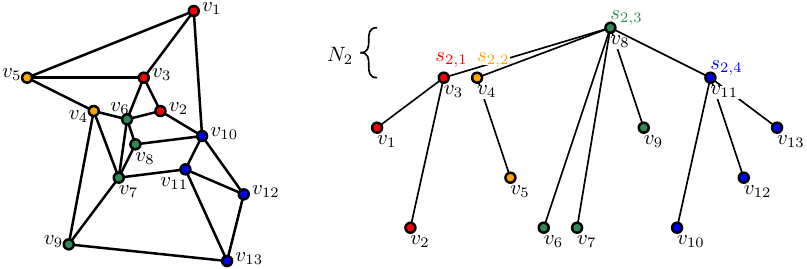}
    \caption{An $r_i$-clustering (left) and its corresponding compressed net-tree (right). The elements of~$N_i$ are on the $(\leq i)$-level of the net-tree. The subtrees rooted at~$N_i$ induce a clustering partition. The colour of a vertex indicates the $r_i$-clustering it belongs to.}
    \label{figure:net-tree3}
\end{figure}

\subsection{Putting it together}
\label{subsection:putting_it_together}

Now that we have the Euclidean WSPD in Lemma~\ref{lemma:euclidean_wspd_construction} and the net-tree in Lemma~\ref{lemma:net_tree}, we will combine them to obtain the Graph WSPD. Recall that, for each~$(A,B)$ in the Euclidean WSPD, our approach is to locally cluster~$A$ into $O(\lambda \sqrt{\semisize(V \cap 6A)})$ clusters and~$B$ into $O(\lambda \sqrt{\semisize(V \cap 6B)})$ clusters so that each cluster has a small graph diameter relative to its quadtree node size. Our approach for computing the local clustering for~$A$ or~$B$ is to retrieve the global $r$-clustering where $r$ corresponds to the side length of~$A$ or~$B$, and to clip the global clustering to the square~$A$ or~$B$. We formalise our approach in Theorem~\ref{theorem:wspd_construction} to obtain the main theorem of this section.

\begin{theorem}
    \label{theorem:wspd_construction}
    Let $G = (V,E)$ be a $\lambda$-low-density graph with $n$ vertices. For all $\varepsilon > 0$, there is a $(1/\varepsilon)$-WSPD for $V$ in the graph metric $G$ with $O(n \lambda^2 \varepsilon^{-4} \log n)$ pairs. The construction takes $O(n \lambda^2 \varepsilon^{-4} \log^2 n)$ expected time.
\end{theorem}

\begin{proof}
The number of pairs in the Graph WSPD follows from Theorem~\ref{theorem:unbounded_spread_wspd}. It remains to state the algorithm for computing the Graph WSPD, prove its correctness, and analyse its running time. 

We state our Graph WSPD algorithm. Construct a Euclidean WSPD using Lemma~\ref{lemma:euclidean_wspd_construction}. Construct a compressed net-tree using Lemma~\ref{lemma:net_tree}. Construct a three-dimensional orthogonal range searching data structure containing $(x(v),y(v),i(v)) \in \mathbb R^3$ for all $v \in V$, where $x(v)$ and $y(v)$ are the $x$- and $y$-coordinates of the vertex~$v \in \mathbb R^2$, and $i(v)$ is the level of $v$ in the compressed net-tree. Note that each vertex~$v \in V$ appears exactly once in the net-tree, so~$i(v)$ is well defined. For each well-separated pair of quadtree nodes~$(A,B)$ in the Euclidean WSPD, define $i(A) = \log_2(\Delta/\ell(A)) - 1$, where~$\ell(A)$ is the side length of $A$. Construct the set of cluster centres $\{s_j\}$ by querying for all points in the three dimensional range $[x_{min}(2A),x_{max}(2A)] \times [y_{min}(2A),y_{max}(2A)] \times [i(A),\infty]$, where $x_{min}(2A)$, $x_{max}(2A)$, $y_{min}(2A)$, $y_{max}(2A)$ denote the minimum and maximum $x$- and $y$-coordinates of the square $2A$. Recall that $2A$ is $A$ expanded by a factor of 2 from its centre. Construct~$\{t_{k}\}$ similarly by querying for all points in the range $[x_{min}(2B),x_{max}(2B)] \times [y_{min}(2B),y_{max}(2B)] \times [i(B),\infty]$. Let $S_j$ be the descendants of $s_j$ with $(>i(A))$-level so that the path from $s_j$ to the descendant passes through no other nodes in the $(\leq i(A))$-levels. Add~$s_j$ to the set~$S_j$. Let $T_k$ be the descendants of $t_k$ with $(>i(B))$-level so that the path from $t_k$ to the descendant passes through no other nodes in the~$(\leq i(B))$-levels. Add~$t_k$ to the set $T_k$. Let $C_j = S_j \cap A$ and $D_k = T_k \cap B$. Note that computing the elements in the sets $C_j$ and $D_k$ explicitly would be too expensive, so instead we represent~$C_j$ implicitly with the pair $(s_j,A)$ and~$D_k$ implicitly with pair $(t_k,B)$. Add the well-separated pairs $\{(C_j,D_k)\}$ to the WSPD for all $j$ and $k$, and repeat this for every pair of well-separated quadtree nodes~$(A,B)$ in the Euclidean WSPD. This completes the description of the Graph WSPD algorithm.

In Figure~\ref{figure:net-tree4}, we provide an example where we construct the clusters $\{C_j\}$ from a set~$A$. In this example, we have $A = \{v_3,v_6\}$ and $2A = \{v_2,v_3,v_4,v_6,v_8\}$. We assume that $i(A) = 2$, and $N_2 = \{v_3,v_4,v_8,v_{11}\}$. Therefore, the three dimensional range query $[x_{min}(2A),x_{max}(2A)] \times [y_{min}(2A),y_{max}(2A)] \times [i(A),\infty]$ returns the set of points $\{v_3,v_4,v_8\}$, which we relabel with $\{s_1,s_2,s_3\}$. From the net-tree we obtain $S_1 = \{v_1,v_2,v_3\}$, $S_2 = \{v_4,v_5\}$, and $S_3 = \{v_6,v_7,v_8,v_9\}$. We clip these sets to the bounding box of~$A$ to obtain the clusters~$\{C_j\}$. In particular, we obtain~$C_1 = \{v_3\}$, $C_2 = \{\}$, and $C_3 = \{v_6\}$. In general, computing the clusters~$C_j$ explicitly is too expensive, so instead we represent these sets implicitly. In particular, the implicit representations are $C_1 = (v_3,A)$, $C_2 = (v_4,A)$, and $C_3 = (v_8,A)$. This completes the example. 

\begin{figure}[ht]
    \centering
    \includegraphics[width=\textwidth]{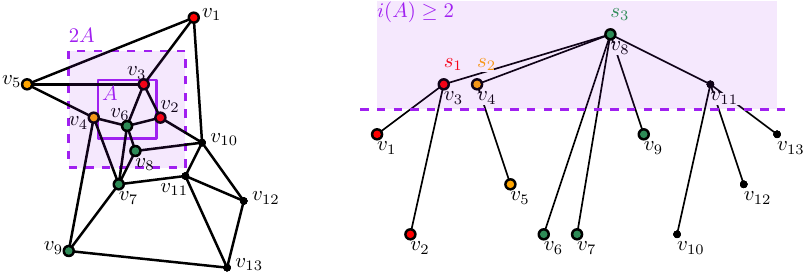}
    \caption{Given a quadtree cell~$A$, we query a three dimensional range searching data structure to obtain vertices~$\{s_j\}$ that lie in~$2A$, and have a level $\geq i(A)$. The clusters $\{C_j\}$ are represented implicitly by $(s_j,A)$.}
    \label{figure:net-tree4}
\end{figure}

Next, we prove the correctness of our Graph WSPD algorithm. The correctness of the Euclidean WSPD construction follows from Lemma~\ref{lemma:unbounded_spread_euclidean_wspd}, so it suffices to prove the correctness of the clustering. The clustering requirements of Theorem~\ref{theorem:unbounded_spread_wspd} are given in Lemma~\ref{lemma:clustering_strong}. The requirements are, first, that~$\{C_j\}$ forms a partition of $V \cap A$, second, that the diameters of~$\{C_j\}$ are at most $\ell(A)$, and third, that there are at most $O(\lambda \sqrt{\semisize(V \cap 6A)})$ clusters in $\{C_j\}$. 

First, we show that $\{C_j\}$ forms a partition of $V \cap A$. Let $r_{i(A)} = \Delta/2^{i(A) - 1} = \ell(A) / 4$. By Lemma~\ref{lemma:net_tree}, we get that querying $[x_{min}(2A),x_{max}(2A)] \times [y_{min}(2A),y_{max}(2A)] \times [i(A),\infty]$ returns the cluster centres of a~$2r_{i(A)}$-clustering that lie in $2A$. The $2r_{i(A)}$-clustering would be a partition of~$V$, except that we have removed all cluster centres that lie outside $2A$. However, since the cluster radius is $\ell(A)/2$, we cannot have removed any clusters that intersect $A$. Therefore, the $2r_{i(A)}$-clustering would still cover all vertices in $V \cap A$. The net-tree construction in Lemma~\ref{lemma:net_tree} states that the sets~$S_j$ form a $2r_i$-clustering, and therefore a partition of $V \cap A$, as required. 

Second, we show that the diameters of $\{C_j\}$ are at most $\ell(A)$. From Lemma~\ref{lemma:net_tree}, the radius of the cluster $S_j$ is at most $2r_i = \ell(A)/2$, so the diameter of $S_j$ is at most $\ell(A)$. Therefore, the diameter of $C_j$ is also at most $\ell(A)$, as required. 

Third, we show that there are at most $O(\lambda \sqrt{\semisize(V \cap 6A)})$ clusters in $\{C_j\}$. Here, we refer to Lemma~\ref{lemma:clustering_weak}. To show the bound on the number of clusters, we only require the length bound of Lemma~\ref{lemma:low_density_implies_packed_strong} (or Fact~\ref{fact:low_density_implies_packed_weak}) and that the distance between clusters is at least $\ell(A)/c_1$ for some constant $c_1$. The lower bound on the distance between clusters comes from the packing property in Lemma~\ref{lemma:net_tree}. Therefore, the clustering in our Graph WSPD algorithm satisfies the requirements of Lemma~\ref{lemma:clustering_strong}. The correctness of the remainder of the algorithm follows from Lemma~\ref{lemma:clustering_strong} and Theorem~\ref{theorem:unbounded_spread_wspd}, and we have completed the proof of correctness.

Finally, we analyse the running time of our Graph WSPD algorithm. Constructing the Euclidean WSPD using Lemma~\ref{lemma:euclidean_wspd_construction} takes~$O(n \log n + n \varepsilon^{-2})$ time. Constructing the compressed net-tree takes~$O(n \lambda \log^2 n)$ expected time by Lemma~\ref{lemma:net_tree}. Constructing the three-dimensional orthogonal range searching data structure takes~$O(n\log^2 n)$ time. Using fractional cascading~\cite{DBLP:books/lib/BergCKO08}, each orthogonal range query takes $O(\log^2 n + k)$ time, where $k$ is the size of the output. When querying for quadtree node~$A$, Lemma~\ref{lemma:clustering_strong} implies that the number of clusters is $k = O(\lambda \sqrt{\semisize(V \cap 6A)})$. Therefore, performing all orthogonal range queries takes $O(n \log^2 n + \lambda \sum (\sqrt{\semisize(V \cap 6A_i)} + \sqrt{\semisize(V \cap 6B_i)}))$ time in total, which is at most~$O(n \lambda^2 \varepsilon^{-4} \log^2 n)$. The overall running time is $O(n \lambda^2 \varepsilon^{-4} \log^2 n)$ deterministic plus $O(n \lambda \log^2 n)$ expected, which combines to the stated running time. 
\end{proof}

\section{Membership and distance oracle}
\label{section:wspd_oracle}

Similar to classic WSPD constructions, it is too expensive to enumerate the vertex pairs contained in a Graph WSPD pair, as there are $\Omega(n^2)$ vertex pairs in total to enumerate. In this section, we provide the next best thing; a membership oracle that, given any vertex pair, answers in $O(1)$ time the Graph WSPD that the vertex pair belongs to. In Section~\ref{subsection:membership_oracle}, we construct the membership oracle, which has near-linear size. In Section~\ref{subsection:approximate_distance_oracle}, we use the membership oracle to construct an approximate distance oracle, which also has near-linear size. Note that for queries, we require the word RAM model of computation.

\subsection{Membership oracle}
\label{subsection:membership_oracle}

Given a pair of vertices $a,b \in V$, a membership oracle returns the Graph WSPD pair $(C_{j}, D_{k})$ satisfying $a \in C_{j}$ and $b \in D_{k}$. Recall from our construction in Section~\ref{subsection:putting_it_together} that $C_{j}$ is represented implicitly by $(s_j,A)$, where~$s_j$ is a net-tree node and~$A$ is a quadtree node, and $D_k$ is represented implicitly by $(t_k,B)$ where $t_k$ is a net-tree node and~$B$ is a quadtree node. Let~$S_j$ be the set of descendants with $(>i)$-level so that the path from~$s_j$ to the descendant passes through no other nodes in the $(\leq i)$-level. Define $T_k$ similarly. Recall from Section~\ref{subsection:putting_it_together} that $a \in C_j \iff a \in A \cap S_j$, and $b \in D_k \iff b \in B \cap T_k$. Our goal is to compute~$(s_j,A)$ and~$(t_k,B)$. 

We start by computing the quadtree nodes $A$ and $B$. Recall that~$(A,B)$ is a Euclidean WSPD pair. Moreover, $(A,B)$ is the unique Euclidean WSPD pair that contains~$(a,b)$. Recall from the proof of Lemma~\ref{lemma:unbounded_spread_euclidean_wspd}a that all Euclidean WSPD pairs are congruent and $(1/\varepsilon)$-maximal. We use the following corollary of Lemma~\ref{lemma:quadtree_ancestor} to compute the unique pair of congruent and $(1/\varepsilon)$-maximal quadtree nodes~$A$ and~$B$ that are ancestors of~$a$ and $b$ respectively.

\begin{corollary}
\label{corollary:congruent_and_maximal}
Given a quadtree and a pair of vertices $a,b \in V$, one can compute in $O(1)$ time a pair of congruent and $(1/\varepsilon)$-maximal quadtree nodes $A$ and $B$ so that $a \in A$ and $b \in B$. 
\end{corollary}

\begin{proof}
Same as Lemma~\ref{lemma:quadtree_ancestor}, except that we replace the $(2/\varepsilon)$-well-separated pair $(A_i,B_i)$ with the pair of vertices~$(a,b)$. In particular, the side lengths of $A$ and $B$ must be in the range~$[c_1\varepsilon d_2(a,b), c_2\varepsilon \cdot d_2(a,b)]$ for constants $c_1$ and $c_2$, so there are only $O(1)$ ancestors of~$a$ and~$b$ to check. Constructing these ancestors takes~$O(1)$ time, since the floor function can be computed in~$O(1)$ time. Checking if a pair of ancestors is~$(1/\varepsilon)$-well-separated takes~$O(1)$ time. Finally, we return the largest of these pairs of congruent ancestors that are~$(1/\varepsilon)$-well-separated.
\end{proof}

This completes the query for quadtree nodes $A$ and $B$. It remains to query the net-tree nodes~$s_j$ and~$t_k$. Recall from our algorithm in Section~\ref{subsection:putting_it_together} that~$s_j$ and~$t_k$ are on the $(\leq i(A))$-level of the net-tree, where $i(A) = \log_2(\Delta / \ell(A)) - 1$, $\ell(A)$ is the side length of~$A$, and~$\Delta$ is the side length of the root node in the quadtree. Furthermore, $a \in C_j$ implies that $a \in S_j$, where $S_j$ is the set of descendants of~$s_j$ so that the path from~$s_j$ to the descendant passes through no other nodes in the $(\leq i(A))$-level. In other words,~$s_j$ is the first $(\leq i(A))$-level ancestor on the path from $a$ to the root. Similarly, $t_k$ is the first ($\leq i(B))$-level ancestor on the path from $b$ to the root.

In the case where the spread~$\Phi$ is bounded, our uncompressed net-tree has size~$O(n \log \Phi)$, see Figure~\ref{figure:net-tree1}. Since the net-tree is uncompressed, the $(\leq i(A))$-level ancestor of any leaf node is simply the $i(A)$-level ancestor, which is also exactly $i(A)$ hops from the root. Therefore, we can use a standard level ancestor query~\cite{DBLP:journals/tcs/BenderF04,DBLP:journals/jcss/BerkmanV94,DBLP:conf/wads/Dietz91} to obtain the net-tree ancestors~$s_j$ and~$t_k$.

\begin{fact}[Level ancestor query~\cite{DBLP:journals/tcs/BenderF04,DBLP:journals/jcss/BerkmanV94,DBLP:conf/wads/Dietz91}]
\label{fact:level_ancestor}
Given a rooted tree $T$ with~$n$ nodes, one can preprocess the tree in~$O(n)$ time and space, so that given a query node~$u$ and a query depth~$d$, one can return in $O(1)$ query time the ancestor of~$u$ with depth~$d$. The depth of a node is the number of hops from the root to the node.
\end{fact}

In the case where the spread~$\Phi$ is unbounded, our compressed net-tree has size~$O(n)$, see Figure~\ref{figure:net-tree2}. We run into two obstacles that prevent us from using the level ancestor query in Fact~\ref{fact:level_ancestor} to query~$s_j$ and~$t_k$. The first obstacle is that we can no longer guarantee that a copy of the~$(\leq i(A))$-level ancestor will always be on the~$i(A)$-level. The second obstacle is that we no longer have that the parent of any node on the $i$-level is on the~$(i-1)$-level, and therefore, we can no longer guarantee that a node on the $i(A)$-level is exactly~$i(A)$ hops from the root. 

One approach to avoid both obstacles is to instead use a weighted ancestor query to compute~$s_j$ and~$t_k$. The weighted ancestor query returns the first $(\leq i(A))$-level ancestor on the path from~$a$ to the root. However, the downside of using the weighted level ancestor query is that it requires additional preprocessing time, see Fact~\ref{fact:weighted_ancestor}. 

\begin{fact}[Weighted ancestor query~\cite{DBLP:conf/cpm/FarachM96}]
\label{fact:weighted_ancestor}
Given a rooted tree~$T$ with~$n$ nodes, each node is given a weight in $\{1,\ldots,n\}$, and the weight of a child is greater than the weight of its parent. One can preprocess $T$ in $O(n^{1+\delta})$ time and space, so that given a query node~$u$ and weight~$w$, one can return in $O(1/\delta)$ time the first ancestor on the path from $u$ to the root with weight~$\leq w$. 
\end{fact}

Ideally, we would like to avoid the~$O(n^{1+\delta})$ preprocessing time and space solution in Fact~\ref{fact:weighted_ancestor}, and instead obtain a linear time and space solution with constant query time. Unfortunately, the weighted ancestor query is a generalisation of the predecessor query, which prevents any linear space solution with constant query time.

\begin{fact}[Predecessor query~\cite{DBLP:conf/esa/GawrychowskiLN14,DBLP:conf/stoc/PatrascuT06}]
\label{fact:predecessor}
Given a set $S$ of~$n$ integers in $\{1,\ldots,U\}$ where~$U \geq n$, any data structure that preprocesses the set $S$ in $O(n \polylog U)$ space requires at least $\Omega(\log \log U)$ time to answer predecessor queries. A predecessor query is, given an integer $x$, to find the largest integer in $S$ that is $\leq x$.
\end{fact}

Surprisingly, the predecessor query is equivalent to the weighted ancestor query. An efficient data structure for the predecessor query implies an efficient data structure for the weighted ancestor query. After linear preprocessing, a weighted ancestor query can be answered using a predecessor query plus constant query time. Next, we state the reduction together with a sketch of its proof. Then, we will modify this proof so that it can be applied to our special case to obtain constant query time. 

\begin{fact}[Predecessor query $\to$ weighted ancestor query~\cite{DBLP:conf/esa/GawrychowskiLN14,DBLP:conf/soda/KopelowitzL07}]
\label{fact:predecessor_to_weighted_ancestor}
Given a rooted tree~$T$ with~$n$ nodes, suppose that each node is given a weight in~$\{1,\ldots,U\}$, so that $U \geq n$ and the weight of a child is greater than the weight of its parent. Then, one can construct a data structure on~$T$ in~$\predecessor_p(n,U) + O(n)$ time and space, so that weighted ancestor queries on~$T$ can be answered in $\predecessor_q(n,U)+O(1)$ time. Here, $\predecessor_p(n,U)$ and $\predecessor_q(n,U)$ denote the preprocessing and query times for a predecessor data structure over a set of~$n$ integers in~$\{1,\ldots,U\}$. 
\end{fact}

\begin{proof}[Proof (Sketch)]
    We follow the proof sketch in Section~5 of~\cite{DBLP:conf/cpm/BelazzouguiKPR21}. Perform a heavy-light decomposition~\cite{DBLP:conf/stoc/SleatorT81} of~$T$. A heavy-light decomposition is a decomposition of the tree into heavy paths. An edge~$(u,v)$ connecting a child~$v$ to its parent~$u$ is heavy if $\size(v) > \size(u)/2$, where $\size(w)$ denotes the number of nodes in the subtree rooted at~$w$; a path is heavy if all edges on the path are heavy. Lemma~11 of~\cite{DBLP:conf/esa/GawrychowskiLN14} states that one can construct a data structure in~$O(n)$ time and space, so that given a query node~$u$ and a weight~$w$, one can return in~$O(1)$ time the heavy path that contains the first ancestor of~$u$ with weight~$\leq w$. At preprocessing time, build a predecessor data structure on the set of heavy paths, and at query time, return the predecessor with weight~$\leq w$ along the heavy path. The preprocessing time is $O(n) + \predecessor_p(n,U)$, and the query time is~$O(1) + \predecessor_q(n,U)$, as required. Note that a similar result is provided in Theorem~7.1 of~\cite{DBLP:conf/soda/KopelowitzL07}, however, their result is weaker since a constant number of predecessor queries are required. 
\end{proof}

Motivated by Fact~\ref{fact:predecessor_to_weighted_ancestor}, we propose a special case of the weighted ancestor query that admits a linear space data structure with linear preprocessing time and constant query time. Our special case is: we either return a weighted ancestor of exactly weight~$w$, or we return that no such ancestor exists. In Lemma~\ref{lemma:weighted_ancestor_special_case}, we construct a data structure for our special case. It is straightforward to see that our special case is no longer a generalisation of a predecessor query, so the lower bound in Fact~\ref{fact:predecessor} no longer applies. 

\begin{lemma}
\label{lemma:weighted_ancestor_special_case}
Given a rooted tree~$T$ with~$n$ nodes, each node is given a weight in~$\{1,\ldots,n\}$, and the weight of a child is greater than the weight of its parent. One can preprocess~$T$ in $O(n)$ time and space, so that given a node~$u$ and a weight~$w$, one can return in~$O(1)$ time the ancestor of~$u$ with weight~$w$, or that no such ancestor exists, under the word ram model of computation.
\end{lemma}

\begin{proof}
We use Fact~\ref{fact:predecessor_to_weighted_ancestor}, where our predecessor data structure is a lookup table. If there is an ancestor on the heavy path with weight exactly~$w$, we can return it from our lookup table, otherwise we can return that there is no such ancestor. Under the word RAM model of computation, we can implement a lookup table with~$O(n)$ preprocessing time and space and~$O(1)$ query time. Therefore, $\predecessor_p(n,U) = O(n)$ and $\predecessor_q(n,U) = O(1)$. By Fact~\ref{fact:predecessor_to_weighted_ancestor}, the overall preprocessing space and time is~$O(n)$ and the overall query time is~$O(1)$.
\end{proof}

Now, we are ready to prove the main theorem of this section. We will query the net-tree nodes~$s_j$ and~$t_k$ while avoiding the~$O(n^{1+\delta})$ preprocessing solution in Fact~\ref{fact:weighted_ancestor}. Recall that the two obstacles related to querying the net-tree nodes~$s_j$ and~$t_k$ are, first, the $(\leq i(A))$-level ancestor may not be on the~$i(A)$-level, and second, a node on the~$i(A)$ level may not be exactly~$i(A)$ hops from the root. We will resolve the first obstacle by constructing a semi-compressed net-tree. Our semi-compressed net-tree ensures that $s_j$ is on the $i(A)$-level and $t_k$ is on the $i(B)$-level for all Graph WSPD pairs~$(C_j, D_k)$, where~$C_j = (s_j,A)$ and~$D_k = (t_k,B)$. We will resolve the second obstacle by applying Lemma~\ref{lemma:weighted_ancestor_special_case} to solve a special case of the weighted ancestor query problem. We describe our algorithm and analyse its running time in Theorem~\ref{theorem:membership_oracle}.

\begin{theorem}
    \label{theorem:membership_oracle}
    Let $G = (V,E)$ be a $\lambda$-low-density graph with $n$ vertices. For $\varepsilon > 0$, one can construct a $(1/\varepsilon)$-WSPD for~$V$ in the graph metric~$G$ with $O(n \lambda^2 \varepsilon^{-4} \log n)$ pairs, and a membership oracle for the WSPD of $O(n \lambda^2 \varepsilon^{-4} \log n)$ size. Given a pair of vertices $a,b \in V$, the membership oracle returns in $O(1)$ time the unique WSPD pair that contains $(a,b)$, under the word RAM model of computation. The construction time for the membership oracle is~$O(n \lambda^2 \varepsilon^{-4} \log^2 n)$ expected time.
\end{theorem}

\begin{proof}
    We state the preprocessing procedure. Construct a Euclidean WSPD using Lemma~\ref{lemma:euclidean_wspd_construction}. Construct a compressed net-tree using Lemma~\ref{lemma:net_tree}. Construct a $(1/\varepsilon)$-WSPD using Theorem~\ref{theorem:wspd_construction}. For each WSPD pair~$(C_j,D_k)$, where~$C_j = (s_j,A)$ and~$D_k = (t_k,B)$, insert the pairs $(s_j,i(A))$ and $(t_k,i(B))$ into a list~$L$. Next, we will use~$L$ to construct a semi-compressed net-tree. The semi-compressed net-tree will have the property that for each $(v,i) \in L$, we have a copy of the vertex~$v$ as a node on the $i$-level. 
    
    We state our construction for the semi-compressed net-tree. Construct a sequence of~$r_i$-cluster centres~$\{N_i\}$ using Fact~\ref{fact:net_tree_unbounded_spread}, where~$r_i = \Delta/2^{i-1}$. We make $N_1$ the root of the net-tree. Recall that for the compressed net-tree in Lemma~\ref{lemma:net_tree}, we place a vertex~$v$ into level~$i$ if $v \in N_i \setminus N_{i-1}$. For the semi-compressed net-tree, we instead place a vertex~$v$ into level~$i$ if either $v \in N_i \setminus N_{i-1}$, or $(v,i) \in L$. For a node $v \in N_i$, define $\parent(v)$ to be a vertex on the $(< i)$-level that is the cluster centre of the $r_{i-1}$-cluster that contains~$v$. If there are multiple copies of the cluster centre on the $(<i)$-level, choose the one that is closest to the~$(i-1)$-level. This completes the construction of the semi-compressed net-tree. The semi-compressed net-tree satisfies the covering, packing, and partition properties in Lemma~\ref{lemma:net_tree}, using the same proof. Moreover, by construction we have that a copy of the vertex~$v$ is on the~$i$-level for every $(v,i) \in L$. Construct the data structure in Lemma~\ref{lemma:weighted_ancestor_special_case} on the semi-compressed quadtree, where the weight of a node is given by its level. This completes the preprocessing procedure.

    We state the query procedure. Given a pair of vertices~$a,b \in V$, use Corollary~\ref{corollary:congruent_and_maximal} to query the quadtree nodes~$A$ and~$B$. Given~$a$ and~$A$, use Lemma~\ref{lemma:weighted_ancestor_special_case} to query the net-tree node~$s_j$, which is the ancestor of~$a$ with weight exactly~$i(A)$. Note that if we define~$s_j$ to be the first $(\leq i(A))$-level ancestor on the path from~$a$ to the root, then $(s_j,i(A)) \in L$, so $s_j$ is the $i(A)$-level ancestor of~$a$, as required. Repeat the query for~$b$ and~$B$ to obtain~$t_k$. This completes the query procedure.

% A figure here would be nice

    We analyse the preprocessing time. Constructing the Euclidean WSPD using Lemma~\ref{lemma:euclidean_wspd_construction} takes $O(n \log n + n \varepsilon^{-2})$ time. Constructing the $(1/\varepsilon)$-WSPD using Theorem~\ref{theorem:wspd_construction} takes $O(n \lambda^2 \varepsilon^{-4} \log^2 n)$ expected time. Constructing the list~$L$ takes $O(n \lambda^2 \varepsilon^{-4} \log n)$ time. Constructing the semi-compressed net-tree takes the same time as constructing the compressed net-tree, which by Lemma~\ref{lemma:net_tree} takes $O(n \lambda \log^2 n)$ time. Therefore, preprocessing takes $O(n \lambda^2 \varepsilon^{-4} \log^2 n)$ expected time. 

    The data structure size is the size of the semi-compressed quadtree plus the size of the semi-compressed net-tree. The semi-compressed quadtree has size~$O(n \varepsilon^{-2})$ by Lemma~\ref{lemma:unbounded_spread_euclidean_wspd}, and the semi-compressed net-tree has its size dominated by~$L$, which has the same size as the Graph WSPD. The size of the Graph WSPD is $O(n \lambda^2 \varepsilon^{-4} \log n)$ by Theorem~\ref{theorem:unbounded_spread_wspd}. Therefore, the overall size of the data structure is $O(n \lambda^2 \varepsilon^{-4} \log n)$.

    Querying the quadtree nodes~$A$ and~$B$ using Corollary~\ref{corollary:congruent_and_maximal} takes~$O(1)$ time. Querying the net-tree nodes~$s_j$ and~$t_k$, which are the $i(A)$-level ancestor of~$a$ and the $i(B)$-level ancestor of~$b$, takes~$O(1)$ time by Lemma~\ref{lemma:weighted_ancestor_special_case}. The overall query time is~$O(1)$.     
\end{proof}

\subsection{Distance oracle}
\label{subsection:approximate_distance_oracle}

Given a pair of vertices $u,v \in G$, an approximate distance oracle returns a $(1+\varepsilon)$-approximation of~$d_G(u,v)$ in $O(1)$ query time, see Definition~\ref{definition:ado}. In Section~\ref{subsection:putting_it_together}, we constructed a Graph WSPD for $G$, using Theorem~\ref{theorem:wspd_construction}. In Section~\ref{subsection:membership_oracle}, we constructed a membership oracle for the WSPD, using Theorem~\ref{theorem:membership_oracle}. In this section, we extend the membership oracle to an approximate distance oracle. 

Our approach is inspired by Gao and Zhang~\cite{DBLP:conf/stoc/GaoZ03}. First, we will construct an exact distance data structure for~$G$ with $O(\lambda \sqrt n)$ query time. Second, we will construct a lookup table, where the keys are Graph WSPD pairs, and its corresponding value is the distance between some representative pair for the WSPD. Third, we will argue that the distance between the representative pair is a $(1+\varepsilon)$-approximation for the distance between any pair in the Graph WSPD. 

First, we will describe our exact distance data structure. Our idea is to use a separator hierarchy~\cite{DBLP:conf/focs/LiptonT77}. To construct the separator hierarchy, we will use the separator theorem of Le and Than~\cite{DBLP:conf/soda/LeT22} for $\tau$-lanky graphs.

\begin{definition}[restate of Definition~\ref{definition:lanky}]
    Let $G=(V,E)$ be a graph embedded in $\mathbb R^2$, and let $\tau \in \mathbb N$. We say that $G$ is $\tau$-lanky if, for all $r \in \mathbb R^+$ and all balls $B$ of radius $r$ centred at a vertex~$v \in V$, there are at most $\tau$ edges in~$E$ that have length at least~$r$, have one endpoint inside~$B$, and have one endpoint outside~$B$.
\end{definition}

It is straightforward to verify that any $\lambda$-low-density graph is also $\lambda$-lanky. Next, we define a balanced separator, and then state the separator theorem of~\cite{DBLP:conf/soda/LeT22}.

\begin{definition}[Balanced separator]
    Let~$G = (V,E)$ be a graph with~$n$ vertices. A balanced separator is a subset of vertices~$S \subset V$ so that removing the vertices~$S$ from~$G$ separates~$G$ into two disjoint subgraphs, each of size at most~$cn$, for some constant~$c$.
\end{definition}

\begin{fact}[Theorem~1 in~\cite{DBLP:conf/soda/LeT22}]
    \label{fact:balanced_separator_lanky}
    Let $G = (V,E)$ be a $\tau$-lanky graph with $n$ vertices. Then, $G$ admits a balanced separator of size $O(\tau \sqrt n)$ that can be computed in $O(\tau n)$ expected time. 
\end{fact}

Now, we are ready to construct the exact distance data structure.

\begin{lemma}
    \label{lemma:exact_distance_oracle}
    Let $G = (V,E)$ be a $\lambda$-low-density graph with $n$ vertices. There is a data structure of~$O(\lambda n \sqrt n )$ size that, given a pair of points~$u, v \in V$, returns in $O(\lambda \sqrt n)$ query time the distance $d_G(u,v)$. The preprocessing time is $O(\lambda^2 n \sqrt n \log n)$ expected. 
\end{lemma}

\begin{proof}
    The data structure is a separator hierarchy~\cite{DBLP:conf/focs/LiptonT77}. Constructing a separator heirarchy from a balanced separator is standard~\cite{DBLP:conf/compgeom/BuchinBG0W24}, so we only provide a very brief sketch of the technique. The construction procedure is use Fact~\ref{fact:balanced_separator_lanky} to recursively separate the graph. The result is a binary tree with a balanced separator stored at the internal nodes of the binary tree. Construct a shortest path tree rooted at each vertex of each of the balanced separators. The query procedure for vertices~$u$ and~$v$ is to compute their lowest common ancestor (LCA), and to retrieve the balanced separators between the LCA and the root. Then, we query the shortest path trees to compute $\min_{w \in B} d_G(u,w) + d_G(w,v)$, where $B$ is the union of the balanced separators. It is straightforward to verify that the data structure size is $O(\lambda n \sqrt n)$, the preprocessing time is~$O(\lambda^2 n \sqrt n \log n)$ expected, and the query time is $O(\lambda \sqrt n)$. 
\end{proof}

Second, we construct a lookup table. Construct a Graph WSPD using Theorem~\ref{theorem:wspd_construction} with a separation constant of~$4/\varepsilon$. The keys of the lookup table are the well-separated pairs in the Graph WSPD. Let~$(C_j, D_k)$ be one of these well-separated pairs. Then $C_j = (s_j,A)$ and $D_k = (t_k,B)$, where $A,B$ are quadtree nodes and~$s_j,t_k$ are net-tree nodes. Recall that net-tree nodes are also vertices in the graph. In the lookup table, for the key~$(C_j, D_k)$, let its corresponding value be the graph distance~$d_G(s_j,t_k)$. This completes the description of the lookup table.

Third, we will argue that the distance $d_G(s_j,t_k)$ is a $(1+\varepsilon)$-approximation of the distance between any pair of vertices contained in $(C_j,D_k)$. Let $u \in C_j$ and $v \in D_k$. Then~$u \in A$ and~$v \in B$. Recall from Theorem~\ref{theorem:wspd_construction} that~$C_j$ is a cluster centred at~$s_j$ with radius~$\ell(A)/2$, where $\ell(A)$ denotes the side length of~$A$. Therefore, $d_G(u,s_j) \leq \ell(A)/2$, and similarly, $d_G(t_k,v) \leq \ell(B)/2$.

Since the quadtree nodes~$A$ and~$B$ are $(4/\varepsilon)$-well-separated, we have:
\[
    d_G(u,v) \geq d_2(u,v) \geq d_2(A,B) \geq \frac 4 \varepsilon \cdot \max(\ell(A),\ell(B)) \geq \frac 4 \varepsilon \cdot (d_G(u,s_j)+d_G(t_k,v)).
\]
We can now rearrange the above inequality to show that the distance $d_G(s_j,t_k)$ is a $(1+\varepsilon)$-approximation of $d_G(u,v)$. We have:
\[
    (1+\frac \varepsilon 4) \cdot d_G(u,v) \geq d_G(s_j,u) + d_G(u,v) + d_G(v,t_k) \geq d_G(s_j,t_k),
\]
\[
    d_G(s_j,t_k) \geq d_G(u,s_j) + d_G(s_j,t_k) + d_G(t_k,v) - \frac \varepsilon 4 \cdot d_G(u,v) \geq (1-\frac \varepsilon 4) \cdot d_G(u,v).
\]
Finally, $(1+\varepsilon/4)/(1-\varepsilon/4) < (1+\varepsilon/4)^2 < (1+\varepsilon)$, so $d_G(s_j,t_k)$ is a $(1+\varepsilon)$-approximation of $d_G(u,v)$.

Now we are ready to prove the main theorem of this section. 

\begin{theorem}
    \label{theorem:ado_construciton}
    Let $G=(V,E)$ be a $\lambda$-low-density graph with~$n$ vertices. For all $\varepsilon > 0$, there is an approximate distance oracle of $O(n \lambda^2 \varepsilon^{-4} \log n)$ size, that, given vertices $u,v \in V$, returns in $O(1)$ time a $(1+\varepsilon)$-approximation of $d_G(u,v)$, under the word RAM model of computation. The preprocessing time is $O(n \sqrt n \lambda ^3 \varepsilon^{-4} \log n)$ expected.
\end{theorem}

\begin{proof}
    We state the preprocessing procedure. Construct a Graph WSPD with separation constant $4/\varepsilon$, using Theorem~\ref{theorem:wspd_construction}. Construct a membership oracle for the Graph WSPD using Theorem~\ref{theorem:membership_oracle}. Construct an exact distance data structure for~$G$ using Lemma~\ref{lemma:exact_distance_oracle}. Construct a lookup table with \mbox{key-value} pairs $((C_j,D_k),d_G(s_j,t_k))$, where $(C_j,D_k)$ is a well-separated pair in the Graph WSPD, $C_j = (s_j,A)$ and~$D_k = (t_k,B)$. This completes the preprocessing procedure.

    We state the query procedure. Given vertices $u,v \in V$, query the membership oracle to obtain the Graph WSPD pair~$(C_j,D_k)$ that contains~$(u,v)$. Query the lookup table using the key~$(C_j,D_k)$ to obtain the distance~$d_G(s_j,t_k)$. Return the distance~$d_G(s_j,t_k)$. The correctness of the $(1+\varepsilon)$-approximation follows from the analysis prior to Theorem~\ref{theorem:ado_construciton}.

    We analyse the data structure size. The query only requires the membership oracle and the lookup table. The size of the membership oracle is~$O(n \lambda^2 \varepsilon^{-4} \log n)$ by Theorem~\ref{theorem:membership_oracle}, and the size of the lookup table is~$O(n \lambda^2 \varepsilon^{-4} \log n)$ by Theorem~\ref{theorem:wspd_construction}. This yields the stated data structure size.

    We analyse the preprocessing time. Constructing the Graph WSPD and the membership oracle takes $O(n \lambda^2 \varepsilon^{-4} \log^2 n)$ expected time by Theorems~\ref{theorem:wspd_construction} and~\ref{theorem:membership_oracle}. Constructing the exact distance data structure takes $O(\lambda^2 n \sqrt n \log n)$ expected time by Lemma~\ref{lemma:exact_distance_oracle}. For each key in the lookup table, we query the exact distance data structure to obtain its corresponding value. There are $O(n \lambda^2 \varepsilon^{-4} \log n)$ keys, by Theorem~\ref{theorem:wspd_construction}, and each query takes $O(\lambda \sqrt n)$ time, by Lemma~\ref{lemma:exact_distance_oracle}. Therefore, constructing the lookup table takes $O(n \sqrt n \lambda ^3 \varepsilon^{-4} \log n)$ time. Overall, the preprocessing time is $O(n \sqrt n \lambda ^3 \varepsilon^{-4} \log n)$ expected.

    We analyse the query time. Querying the membership oracle takes $O(1)$ time by Theorem~\ref{theorem:membership_oracle}. Querying the lookup table takes $O(1)$ time by Definition~\ref{definition:word_ram}. The overall query time is $O(1)$.
\end{proof}

\section{Conclusion}

In this paper, we constructed a $(1/\varepsilon)$-well-separated pair decomposition for a $\lambda$-low-density graph. Our well-separated pair decomposition has $O(n \lambda^2 \varepsilon^{-4} \log n)$ pairs. It admits a membership oracle of $O(n \lambda^2 \varepsilon^{-4} \log n)$ size, so that given a pair of vertices~$(a,b)$, returns in $O(1)$ time the WSPD pair containing~$(a,b)$. We use our well-separated pair decomposition and its membership oracle to construct an approximate distance oracle for low density graphs, that has $O(n \lambda^2 \varepsilon^{-4} \log n)$ size, and answers $(1+\varepsilon)$-approximate distance queries in~$O(1)$ time. Our queries assume the word RAM model of computation.

We argued that low density graphs are a realistic and useful graph class for studying road networks. We believe that low density provides a theoretical explanation as to why real-world road networks admit efficient shortest path data structures and small separators. We hope that by studying low density graphs, we can gain insights into the structure of road networks, which hopefully will lead to more efficient algorithms on road networks in the future. 

We conclude with a list of open problems. Can the size of a well-separated pair decomposition or the size of an approximate distance oracle be improved to~$O(n)$, for low density graphs? Even for subsets of the $\sqrt n \times \sqrt n$ grid graph, which is a special case of both low density and unit disc graphs, this problem remains unresolved. Otherwise, can a lower bound be shown? Low density graphs are a special case of lanky graphs and graphs with constant crossing degeneracy; can we extend our results to more general graph classes? Can the preprocessing time of our approximate distance oracle be made determistic? Using the same techniques as Gao and Zhang~\cite{DBLP:conf/stoc/GaoZ03}, our well-separated pair decomposition implies randomised $\tilde O(n^{3/2})$ time approximation algorithms for the diameter, radius, and stretch of a low density graph; can these be improved? Can the linear time single source shortest path algorithm of Henzinger, Klein, Rao and Subramanian's~\cite{DBLP:journals/jcss/HenzingerKRS97} be adapted to low density graphs? The result of~\cite{DBLP:journals/jcss/HenzingerKRS97} only applies to graph classes that are minor closed.

Our Euclidean WSPD construction, using a semi-compressed quadtree, may be of independent interest. One advantage of our Euclidean WSPD over previous constructions are that it admits a membership oracle with $O(1)$ query time, assuming arithmetic operations and the floor function can be performed in $O(1)$ time. Another advantage of our Euclidean WSPD is that each node in the semi-compressed quadtree participates in only a constant number of WSPD pairs, although, a WSPD with this property was already essentially known (see Footnote~3 in Chapter~3.5 of~\cite{har2011geometric}).

Finally, we believe that the SELG graph class may be of independent interest. Any algorithmic results for realistic graph classes, under our definition of realistic in Section~\ref{subsection:other_graph_classes}, also applies to SELG graphs. This inludes low density graphs, bounded growth graphs, disc neighbourhood systems, and graphs with constant crossing degeneracy. As a result, we believe that the SELG graph class may be a good testing ground for algorithmic results for road networks.

\subsubsection*{Acknowledgements}

The second author would like to thank Jacob Holm for useful discussions related to Fact~\ref{fact:predecessor_to_weighted_ancestor} and Lemma~\ref{lemma:weighted_ancestor_special_case}.

\newpage
\bibliographystyle{plain}
\bibliography{bib.bib}

\begin{thebibliography}{100}

\bibitem{DBLP:journals/jacm/AbrahamDFGW16}
Ittai Abraham, Daniel Delling, Amos Fiat, Andrew~V. Goldberg, and Renato~F.
  Werneck.
\newblock Highway dimension and provably efficient shortest path algorithms.
\newblock {\em J. {ACM}}, 63(5):41:1--41:26, 2016.

\bibitem{DBLP:conf/icalp/AlstrupH00}
Stephen Alstrup and Jacob Holm.
\newblock Improved algorithms for finding level ancestors in dynamic trees.
\newblock In {\em Proceedings of the 27th International Colloquium on Automata,
  Languages and Programming, (ICALP)}, pages 73--84. Springer, 2000.

\bibitem{DBLP:conf/stoc/AryaDMSS95}
Sunil Arya, Gautam Das, David~M. Mount, Jeffrey~S. Salowe, and Michiel H.~M.
  Smid.
\newblock Euclidean spanners: short, thin, and lanky.
\newblock In {\em Proceedings of the 27th Annual {ACM} Symposium on Theory of
  Computing (STOC)}, pages 489--498. {ACM}, 1995.

\bibitem{DBLP:journals/jacm/AryaMNSW98}
Sunil Arya, David~M. Mount, Nathan~S. Netanyahu, Ruth Silverman, and Angela~Y.
  Wu.
\newblock An optimal algorithm for approximate nearest neighbor searching fixed
  dimensions.
\newblock {\em J. {ACM}}, 45(6):891--923, 1998.

\bibitem{DBLP:journals/tcs/BadkobehCKP22}
Golnaz Badkobeh, Panagiotis Charalampopoulos, Dmitry Kosolobov, and Solon~P.
  Pissis.
\newblock Internal shortest absent word queries in constant time and linear
  space.
\newblock {\em Theor. Comput. Sci.}, 922:271--282, 2022.

\bibitem{DBLP:conf/icalp/BartalFN19}
Yair Bartal, Nova Fandina, and Ofer Neiman.
\newblock Covering metric spaces by few trees.
\newblock In {\em Proceedings of the 46th International Colloquium on Automata,
  Languages, and Programming (ICALP)}, volume 132 of {\em LIPIcs}, pages
  20:1--20:16. Schloss Dagstuhl - Leibniz-Zentrum f{\"{u}}r Informatik, 2019.

\bibitem{DBLP:conf/alenex/BastFMSS07}
Hannah Bast, Stefan Funke, Domagoj Matijevic, Peter Sanders, and Dominik
  Schultes.
\newblock In transit to constant time shortest-path queries in road networks.
\newblock In {\em Proceedings of the 9th Workshop on Algorithm Engineering and
  Experiments (ALENEX)}. {SIAM}, 2007.

\bibitem{DBLP:conf/cpm/BelazzouguiKPR21}
Djamal Belazzougui, Dmitry Kosolobov, Simon~J. Puglisi, and Rajeev Raman.
\newblock Weighted ancestors in suffix trees revisited.
\newblock In {\em Proceedings of the 32nd Annual Symposium on Combinatorial
  Pattern Matching (CPM)}, volume 191 of {\em LIPIcs}, pages 8:1--8:15. Schloss
  Dagstuhl - Leibniz-Zentrum f{\"{u}}r Informatik, 2021.

\bibitem{DBLP:journals/tcs/BenderF04}
Michael~A. Bender and Martin Farach{-}Colton.
\newblock The level ancestor problem simplified.
\newblock {\em Theor. Comput. Sci.}, 321(1):5--12, 2004.

\bibitem{DBLP:journals/jcss/BerkmanV94}
Omer Berkman and Uzi Vishkin.
\newblock Finding level-ancestors in trees.
\newblock {\em J. Comput. Syst. Sci.}, 48(2):214--230, 1994.

\bibitem{DBLP:journals/comgeo/BerrettyOS98}
Robert{-}Paul Berretty, Mark~H. Overmars, and A.~Frank van~der Stappen.
\newblock Dynamic motion planning in low obstacle density environments.
\newblock {\em Comput. Geom.}, 11(3-4):157--173, 1998.

\bibitem{DBLP:conf/swat/BilleNP24}
Philip Bille, Yakov Nekrich, and Solon~P. Pissis.
\newblock Size-constrained weighted ancestors with applications.
\newblock In {\em Proceedings of the 19th Scandinavian Symposium and Workshops
  on Algorithm Theory (SWAT)}, volume 294 of {\em LIPIcs}, pages 14:1--14:12.
  Schloss Dagstuhl - Leibniz-Zentrum f{\"{u}}r Informatik, 2024.

\bibitem{blog:uber}
Uber blog.
\newblock {How Uber Engineers an Efficient Route}.
\newblock Retrieved from
  \url{https://www.uber.com/en-DK/blog/engineering-routing-engine/}, accessed
  April 2025.

\bibitem{DBLP:conf/iwpec/000119}
Johannes Blum.
\newblock Hierarchy of transportation network parameters and hardness results.
\newblock In {\em Proceedings of the 14th International Symposium on
  Parameterized and Exact Computation (IPEC)}, volume 148 of {\em LIPIcs},
  pages 4:1--4:15. Schloss Dagstuhl - Leibniz-Zentrum f{\"{u}}r Informatik,
  2019.

\bibitem{DBLP:phd/basesearch/Blum23}
Johannes Blum.
\newblock {\em A Parameterized View on Transportation Networks: Algorithms,
  Hierarchy, and Complexity}.
\newblock PhD thesis, University of Konstanz, Germany, 2023.

\bibitem{DBLP:journals/jco/0001FS21}
Johannes Blum, Stefan Funke, and Sabine Storandt.
\newblock Sublinear search spaces for shortest path planning in grid and road
  networks.
\newblock {\em J. Comb. Optim.}, 42(2):231--257, 2021.

\bibitem{DBLP:conf/cocoon/BlumS18}
Johannes Blum and Sabine Storandt.
\newblock Computation and growth of road network dimensions.
\newblock In {\em Proceedings of the 24th International Conference on Computing
  and Combinatorics (COCOON)}, volume 10976 of {\em Lecture Notes in Computer
  Science}, pages 230--241. Springer, 2018.

\bibitem{DBLP:conf/aips/BlumS18}
Johannes Blum and Sabine Storandt.
\newblock Scalability of route planning techniques.
\newblock In {\em Proceedings of the 28th International Conference on Automated
  Planning and Scheduling (ICAPS)}, pages 20--28. {AAAI} Press, 2018.

\bibitem{boeing2020planarity}
Geoff Boeing.
\newblock Planarity and street network representation in urban form analysis.
\newblock {\em Environment and Planning B: Urban Analytics and City Science},
  47(5):855--869, 2020.

\bibitem{DBLP:journals/networks/BoyaciDL22}
Burak Boyaci, Thu~Huong Dang, and Adam~N. Letchford.
\newblock On matchings, {T}-joins, and arc routing in road networks.
\newblock {\em Networks}, 79(1):20--31, 2022.

\bibitem{DBLP:conf/compgeom/BuchinBG0W24}
Kevin Buchin, Maike Buchin, Joachim Gudmundsson, Aleksandr Popov, and Sampson
  Wong.
\newblock Map-matching queries under {F}r{\'{e}}chet distance on low-density
  spanners.
\newblock In {\em Proceedings of the 40th International Symposium on
  Computational Geometry (SoCG)}, volume 293 of {\em LIPIcs}, pages
  27:1--27:15. Schloss Dagstuhl - Leibniz-Zentrum f{\"{u}}r Informatik, 2024.

\bibitem{DBLP:conf/soda/Cabello17}
Sergio Cabello.
\newblock Subquadratic algorithms for the diameter and the sum of pairwise
  distances in planar graphs.
\newblock In {\em Proceedings of the Twenty-Eighth Annual {ACM-SIAM} Symposium
  on Discrete Algorithms (SODA)}, pages 2143--2152. {SIAM}, 2017.

\bibitem{DBLP:conf/soda/CallahanK93}
Paul~B. Callahan and S.~Rao Kosaraju.
\newblock Faster algorithms for some geometric graph problems in higher
  dimensions.
\newblock In {\em Proceedings of the 4th Annual {ACM/SIGACT-SIAM} Symposium on
  Discrete Algorithms (SODA)}, pages 291--300. {ACM/SIAM}, 1993.

\bibitem{DBLP:journals/jacm/CallahanK95}
Paul~B. Callahan and S.~Rao Kosaraju.
\newblock A decomposition of multidimensional point sets with applications to
  {$k$}-nearest-neighbors and {$n$}-body potential fields.
\newblock {\em J. {ACM}}, 42(1):67--90, 1995.

\bibitem{DBLP:journals/jocg/ChanS19a}
Timothy~M. Chan and Dimitrios Skrepetos.
\newblock Approximate shortest paths and distance oracles in weighted unit-disk
  graphs.
\newblock {\em J. Comput. Geom.}, 10(2):3--20, 2019.

\bibitem{DBLP:journals/algorithmica/ChanS19}
Timothy~M. Chan and Dimitrios Skrepetos.
\newblock Faster approximate diameter and distance oracles in planar graphs.
\newblock {\em Algorithmica}, 81(8):3075--3098, 2019.

\bibitem{DBLP:conf/soda/ChangCC0PP25}
Hsien{-}Chih Chang, Vincent Cohen{-}Addad, Jonathan Conroy, Hung Le, Marcin
  Pilipczuk, and Michal Pilipczuk.
\newblock Embedding planar graphs into graphs of treewidth ${O}(\log^3 n)$.
\newblock In {\em Proceedings of the 2025 Annual {ACM-SIAM} Symposium on
  Discrete Algorithms (SODA)}, pages 88--123. {SIAM}, 2025.

\bibitem{DBLP:conf/focs/ChangCLMST23}
Hsien{-}Chih Chang, Jonathan Conroy, Hung Le, Lazar Milenkovic, Shay Solomon,
  and Cuong Than.
\newblock Covering planar metrics (and beyond): {$O(1)$} trees suffice.
\newblock In {\em Proceedings of the 64th {IEEE} Annual Symposium on
  Foundations of Computer Science (FOCS)}, pages 2231--2261. {IEEE}, 2023.

\bibitem{DBLP:conf/soda/ChangCLMST24}
Hsien{-}Chih Chang, Jonathan Conroy, Hung Le, Lazar Milenkovic, Shay Solomon,
  and Cuong Than.
\newblock Shortcut partitions in minor-free graphs: Steiner point removal,
  distance oracles, tree covers, and more.
\newblock In David~P. Woodruff, editor, {\em Proceedings of the 2024 {ACM-SIAM}
  Symposium on Discrete Algorithms (SODA)}, pages 5300--5331. {SIAM}, 2024.

\bibitem{DBLP:journals/algorithmica/ChaudhuriZ00}
Shiva Chaudhuri and Christos~D. Zaroliagis.
\newblock Shortest paths in digraphs of small treewidth. part {I:} sequential
  algorithms.
\newblock {\em Algorithmica}, 27(3):212--226, 2000.

\bibitem{DBLP:conf/alenex/ChenDGNW11}
Daniel Chen, Anne Driemel, Leonidas~J. Guibas, Andy Nguyen, and Carola Wenk.
\newblock Approximate map matching with respect to the {F}r{\'{e}}chet
  distance.
\newblock In {\em Proceedings of the 13th Workshop on Algorithm Engineering and
  Experiments (ALENEX)}, pages 75--83. {SIAM}, 2011.

\bibitem{DBLP:journals/siamcomp/CohenHKZ03}
Edith Cohen, Eran Halperin, Haim Kaplan, and Uri Zwick.
\newblock Reachability and distance queries via 2-hop labels.
\newblock {\em {SIAM} J. Comput.}, 32(5):1338--1355, 2003.

\bibitem{DBLP:conf/focs/Cohen-AddadLPP23}
Vincent Cohen{-}Addad, Hung Le, Marcin Pilipczuk, and Michal Pilipczuk.
\newblock Planar and minor-free metrics embed into metrics of polylogarithmic
  treewidth with expected multiplicative distortion arbitrarily close to 1.
\newblock In {\em 64th {IEEE} Annual Symposium on Foundations of Computer
  Science (FOCS)}, pages 2262--2277. {IEEE}, 2023.

\bibitem{DBLP:conf/soda/ColletteI24}
S{\'{e}}bastien Collette and John Iacono.
\newblock Distances and shortest paths on graphs of bounded highway dimension:
  simple, fast, dynamic.
\newblock In {\em Proceedings of the 35th {ACM-SIAM} Symposium on Discrete
  Algorithms (SODA)}, pages 2657--2678. {SIAM}, 2024.

\bibitem{DBLP:books/lib/BergCKO08}
Mark de~Berg, Otfried Cheong, Marc~J. van Kreveld, and Mark~H. Overmars.
\newblock {\em Computational geometry: algorithms and applications, 3rd
  Edition}.
\newblock Springer, 2008.

\bibitem{DBLP:conf/compgeom/BergKSV97}
Mark de~Berg, Matthew~J. Katz, A.~Frank van~der Stappen, and Jules Vleugels.
\newblock Realistic input models for geometric algorithms.
\newblock In {\em Proceedings of the 13th Annual Symposium on Computational
  Geometry (SoCG)}, pages 294--303. {ACM}, 1997.

\bibitem{DBLP:conf/ipps/DellingGRW11}
Daniel Delling, Andrew~V. Goldberg, Ilya~P. Razenshteyn, and Renato Fonseca~F.
  Werneck.
\newblock Graph partitioning with natural cuts.
\newblock In {\em Proceedings of the 25th {IEEE} International Symposium on
  Parallel and Distributed Processing (IPDPS)}, pages 1135--1146. {IEEE}, 2011.

\bibitem{DBLP:conf/wads/DeryckereGRSW25}
Lindsey Deryckere, Joachim Gudmundsson, Andr{\'{e}} van Renssen, Yuan Sha, and
  Sampson Wong.
\newblock A wspd, separator and small tree cover for c-packed graphs.
\newblock In {\em 19th International Symposium on Algorithms and Data
  Structures, {WADS} 2025}, volume 349 of {\em LIPIcs}, pages 21:1--21:15.
  Schloss Dagstuhl - Leibniz-Zentrum f{\"{u}}r Informatik, 2025.

\bibitem{DBLP:journals/jea/DibbeltSW16}
Julian Dibbelt, Ben Strasser, and Dorothea Wagner.
\newblock Customizable contraction hierarchies.
\newblock {\em {ACM} J. Exp. Algorithmics}, 21(1):1.5:1--1.5:49, 2016.

\bibitem{DBLP:conf/wads/Dietz91}
Paul~F. Dietz.
\newblock Finding level-ancestors in dynamic trees.
\newblock In {\em Proceedings of the 2nd Workshop on Algorithms and Data
  Structures (WADS)}, 1991.

\bibitem{DBLP:journals/dcg/DriemelHW12}
Anne Driemel, Sariel Har{-}Peled, and Carola Wenk.
\newblock Approximating the {{F}r\'{e}chet} distance for realistic curves in
  near linear time.
\newblock {\em Discret. Comput. Geom.}, 48(1):94--127, 2012.

\bibitem{DBLP:journals/jct/DvorakN19}
Zdenek Dvor{\'{a}}k and Sergey Norin.
\newblock Treewidth of graphs with balanced separations.
\newblock {\em J. Comb. Theory {B}}, 137:137--144, 2019.

\bibitem{DBLP:conf/gis/EppsteinG08}
David Eppstein and Michael~T. Goodrich.
\newblock Studying (non-planar) road networks through an algorithmic lens.
\newblock In {\em Proceedings of the 16th {ACM} {SIGSPATIAL} International
  Symposium on Advances in Geographic Information Systems (ACM-GIS)}, page~16.
  {ACM}, 2008.

\bibitem{DBLP:conf/gis/Eppstein017}
David Eppstein and Siddharth Gupta.
\newblock Crossing patterns in nonplanar road networks.
\newblock In {\em Proceedings of the 25th {ACM} {SIGSPATIAL} International
  Conference on Advances in Geographic Information Systems ({ACM-GIS})}, pages
  40:1--40:9. {ACM}, 2017.

\bibitem{DBLP:journals/jocg/EppsteinHS20}
David Eppstein, Sariel Har{-}Peled, and Anastasios Sidiropoulos.
\newblock Approximate greedy clustering and distance selection for graph
  metrics.
\newblock {\em J. Comput. Geom.}, 11(1):629--652, 2020.

\bibitem{DBLP:conf/cpm/FarachM96}
Martin Farach and S.~Muthukrishnan.
\newblock Perfect hashing for strings: Formalization and algorithms.
\newblock In {\em Proceedings of the 7th Annual Symposium on Combinatorial
  Pattern Matching (CPM)}, 1996.

\bibitem{DBLP:conf/soda/FeldmannF25}
Andreas~Emil Feldmann and Arnold Filtser.
\newblock Highway dimension: a metric view.
\newblock In {\em Proceedings of the 2025 Annual {ACM-SIAM} Symposium on
  Discrete Algorithms (SODA)}, pages 3267--3276. {SIAM}, 2025.

\bibitem{DBLP:conf/icalp/FeldmannFKP15}
Andreas~Emil Feldmann, Wai~Shing Fung, Jochen K{\"{o}}nemann, and Ian Post.
\newblock A (1+{\(\epsilon\)})-embedding of low highway dimension graphs into
  bounded treewidth graphs.
\newblock In {\em Proceedings of the 42nd International Colloquium on Automata,
  Languages, and Programming (ICALP)}, volume 9134 of {\em Lecture Notes in
  Computer Science}, pages 469--480. Springer, 2015.

\bibitem{DBLP:journals/siamcomp/Frederickson87}
Greg~N. Frederickson.
\newblock Fast algorithms for shortest paths in planar graphs, with
  applications.
\newblock {\em {SIAM} J. Comput.}, 16(6):1004--1022, 1987.

\bibitem{DBLP:conf/isaac/FunkeS15}
Stefan Funke and Sabine Storandt.
\newblock Provable efficiency of contraction hierarchies with randomized
  preprocessing.
\newblock In {\em Proceedings of the 26th International Symposium on Algorithms
  and Computation (ISAAC)}, volume 9472 of {\em Lecture Notes in Computer
  Science}, pages 479--490. Springer, 2015.

\bibitem{DBLP:conf/stoc/GaoZ03}
Jie Gao and Li~Zhang.
\newblock Well-separated pair decomposition for the unit-disk graph metric and
  its applications.
\newblock In {\em Proceedings of the 35th Annual {ACM} Symposium on Theory of
  Computing (STOC)}, pages 483--492. {ACM}, 2003.

\bibitem{DBLP:conf/esa/GawrychowskiLN14}
Pawel Gawrychowski, Moshe Lewenstein, and Patrick~K. Nicholson.
\newblock Weighted ancestors in suffix trees.
\newblock In {\em Proceedings of the 22th Annual European Symposium on
  Algorithms (ESA)}, volume 8737 of {\em Lecture Notes in Computer Science},
  pages 455--466. Springer, 2014.

\bibitem{geisberger2015route}
Robert Geisberger.
\newblock Route planning, November~3 2015.
\newblock US Patent 9,175,972.

\bibitem{DBLP:conf/wea/GeisbergerSSD08}
Robert Geisberger, Peter Sanders, Dominik Schultes, and Daniel Delling.
\newblock Contraction hierarchies: Faster and simpler hierarchical routing in
  road networks.
\newblock In {\em Proceedings of the 7th International Workshop on Experimental
  Algorithms (WEA)}, volume 5038 of {\em Lecture Notes in Computer Science},
  pages 319--333. Springer, 2008.

\bibitem{DBLP:journals/siamdm/GottliebK13}
Lee{-}Ad Gottlieb and Robert Krauthgamer.
\newblock Proximity algorithms for nearly doubling spaces.
\newblock {\em {SIAM} J. Discret. Math.}, 27(4):1759--1769, 2013.

\bibitem{grammenos2002residential}
Fanis Grammenos, Sevag Pogharian, and Julie Tasker-Brown.
\newblock Residential street pattern design.
\newblock {\em Socio-economic Series}, 75:22, 2002.

\bibitem{DBLP:journals/tcs/GuX19}
Qian{-}Ping Gu and Gengchun Xu.
\newblock Constant query time (1+$\varepsilon$-approximate distance oracle for
  planar graphs.
\newblock {\em Theor. Comput. Sci.}, 761:78--88, 2019.

\bibitem{DBLP:conf/cocoon/GudmundssonHW23}
Joachim Gudmundsson, Zijin Huang, and Sampson Wong.
\newblock Approximating the {\(\lambda\)}-low-density value.
\newblock In {\em Proceedings of the 29th International Conference on Computing
  and Combinatorics (COCOON)}, volume 14422 of {\em Lecture Notes in Computer
  Science}, pages 71--82. Springer, 2023.

\bibitem{DBLP:journals/talg/GudmundssonLNS08}
Joachim Gudmundsson, Christos Levcopoulos, Giri Narasimhan, and Michiel H.~M.
  Smid.
\newblock Approximate distance oracles for geometric spanners.
\newblock {\em {ACM} Trans. Algorithms}, 4(1):10:1--10:34, 2008.

\bibitem{gudmundsson2024map}
Joachim Gudmundsson, Martin~P. Seybold, and Sampson Wong.
\newblock Map matching queries on realistic input graphs under the
  {F}r{\'e}chet distance.
\newblock {\em ACM Transactions on Algorithms}, 20(2):1--33, 2024.

\bibitem{DBLP:conf/icalp/0002KV19}
Siddharth Gupta, Adrian Kosowski, and Laurent Viennot.
\newblock Exploiting hopsets: Improved distance oracles for graphs of constant
  highway dimension and beyond.
\newblock In {\em Proceedings of the 46th International Colloquium on Automata,
  Languages, and Programming (ICALP)}, volume 132 of {\em LIPIcs}, pages
  143:1--143:15. Schloss Dagstuhl - Leibniz-Zentrum f{\"{u}}r Informatik, 2019.

\bibitem{DBLP:conf/alenex/Gutman04}
Ronald~J. Gutman.
\newblock Reach-based routing: {A} new approach to shortest path algorithms
  optimized for road networks.
\newblock In {\em Proceedings of the 6th Workshop on Algorithm Engineering and
  Experiments (ALENEX)}, pages 100--111. {SIAM}, 2004.

\bibitem{har2011geometric}
Sariel Har-Peled.
\newblock {\em Geometric approximation algorithms}.
\newblock American Mathematical Society, USA, 2011.

\bibitem{DBLP:journals/siamcomp/Har-PeledM06}
Sariel Har{-}Peled and Manor Mendel.
\newblock Fast construction of nets in low-dimensional metrics and their
  applications.
\newblock {\em {SIAM} J. Comput.}, 35(5):1148--1184, 2006.

\bibitem{DBLP:journals/siamcomp/Har-PeledQ17}
Sariel Har{-}Peled and Kent Quanrud.
\newblock Approximation algorithms for polynomial-expansion and low-density
  graphs.
\newblock {\em {SIAM} J. Comput.}, 46(6):1712--1744, 2017.

\bibitem{har2025well}
Sariel Har-Peled, Benjamin Raichel, and Eliot~W Robson.
\newblock Well-separated pairs decomposition revisited.
\newblock {\em arXiv preprint arXiv:2509.05997}, 2025.

\bibitem{DBLP:journals/jcss/HenzingerKRS97}
Monika~Rauch Henzinger, Philip~N. Klein, Satish Rao, and Sairam Subramanian.
\newblock Faster shortest-path algorithms for planar graphs.
\newblock {\em J. Comput. Syst. Sci.}, 55(1):3--23, 1997.

\bibitem{DBLP:conf/gis/HuaXT18}
Hua Hua, Hairuo Xie, and Egemen Tanin.
\newblock Is {Euclidean} distance really that bad with road networks?
\newblock In {\em Proceedings of the 11th {ACM} {SIGSPATIAL} International
  Workshop on Computational Transportation Science (IWCTS@SIGSPATIAL)}, pages
  11--20. {ACM}, 2018.

\bibitem{DBLP:conf/focs/HuangJLW18}
Lingxiao Huang, Shaofeng~H.{-}C. Jiang, Jian Li, and Xuan Wu.
\newblock Epsilon-coresets for clustering (with outliers) in doubling metrics.
\newblock In {\em Proceedings of the 59th {IEEE} Annual Symposium on
  Foundations of Computer Science (FOCS)}, pages 814--825. {IEEE} Computer
  Society, 2018.

\bibitem{DBLP:conf/soda/JayaprakashS22}
Aditya Jayaprakash and Mohammad~R. Salavatipour.
\newblock Approximation schemes for capacitated vehicle routing on graphs of
  bounded treewidth, bounded doubling, or highway dimension.
\newblock In {\em Proceedings of the 33rd {ACM-SIAM} Symposium on Discrete
  Algorithms (SODA)}, pages 877--893. {SIAM}, 2022.

\bibitem{DBLP:conf/icalp/KawarabayashiKS11}
Ken{-}ichi Kawarabayashi, Philip~N. Klein, and Christian Sommer.
\newblock Linear-space approximate distance oracles for planar, bounded-genus
  and minor-free graphs.
\newblock In {\em Proceedings of the 38th International Colloquium on Automata,
  Languages and Programming (ICALP)}, volume 6755 of {\em Lecture Notes in
  Computer Science}, pages 135--146. Springer, 2011.

\bibitem{DBLP:conf/soda/KawarabayashiST13}
Ken{-}ichi Kawarabayashi, Christian Sommer, and Mikkel Thorup.
\newblock More compact oracles for approximate distances in undirected planar
  graphs.
\newblock In {\em Proceedings of the 24th Annual {ACM-SIAM} Symposium on
  Discrete Algorithms (SODA)}, pages 550--563. {SIAM}, 2013.

\bibitem{DBLP:conf/soda/Klein02}
Philip~N. Klein.
\newblock Preprocessing an undirected planar network to enable fast approximate
  distance queries.
\newblock In {\em Proceedings of the 13th Annual {ACM-SIAM} Symposium on
  Discrete Algorithms (SODA)}, pages 820--827. {ACM/SIAM}, 2002.

\bibitem{DBLP:journals/talg/KociumakaKRRW20}
Tomasz Kociumaka, Marcin Kubica, Jakub Radoszewski, Wojciech Rytter, and Tomasz
  Walen.
\newblock A linear-time algorithm for seeds computation.
\newblock {\em {ACM} Trans. Algorithms}, 16(2):27:1--27:23, 2020.

\bibitem{DBLP:conf/soda/KopelowitzL07}
Tsvi Kopelowitz and Moshe Lewenstein.
\newblock Dynamic weighted ancestors.
\newblock In {\em Proceedings of the 18th Annual {ACM-SIAM} Symposium on
  Discrete Algorithms(SODA)}, pages 565--574. {SIAM}, 2007.

\bibitem{DBLP:conf/soda/KosowskiV17}
Adrian Kosowski and Laurent Viennot.
\newblock Beyond highway dimension: Small distance labels using tree skeletons.
\newblock In {\em Proceedings of the 28th Annual {ACM-SIAM} Symposium on
  Discrete Algorithms (SODA)}, pages 1462--1478. {SIAM}, 2017.

\bibitem{DBLP:conf/soda/Le23}
Hung Le.
\newblock Approximate distance oracles for planar graphs with subpolynomial
  error dependency.
\newblock In {\em Proceedings of the 2023 {ACM-SIAM} Symposium on Discrete
  Algorithms (SODA)}, pages 1877--1904. {SIAM}, 2023.

\bibitem{DBLP:conf/soda/LeT22}
Hung Le and Cuong Than.
\newblock Greedy spanners in euclidean spaces admit sublinear separators.
\newblock In {\em Proceedings of the 33rd {ACM-SIAM} Symposium on Discrete
  Algorithms (SODA)}, pages 3287--3310. {SIAM}, 2022.

\bibitem{DBLP:conf/focs/LeW21}
Hung Le and Christian Wulff{-}Nilsen.
\newblock Optimal approximate distance oracle for planar graphs.
\newblock In {\em Proceedings of the 62nd {IEEE} Annual Symposium on
  Foundations of Computer Science (FOCS)}, pages 363--374. {IEEE}, 2021.

\bibitem{DBLP:conf/focs/LiptonT77}
Richard~J. Lipton and Robert~Endre Tarjan.
\newblock Application of a planar separator theorem.
\newblock In {\em Proceedings of the 18th Annual Symposium on Foundations of
  Computer Science (FOCS)}, pages 162--170. {IEEE} Computer Society, 1977.

\bibitem{DBLP:conf/icdt/ManiuSJ19}
Silviu Maniu, Pierre Senellart, and Suraj Jog.
\newblock An experimental study of the treewidth of real-world graph data.
\newblock In {\em Proceedings of the 22nd International Conference on Database
  Theory (ICDT)}, volume 127 of {\em LIPIcs}, pages 12:1--12:18. Schloss
  Dagstuhl - Leibniz-Zentrum f{\"{u}}r Informatik, 2019.

\bibitem{DBLP:journals/cjtcs/MendelS09}
Manor Mendel and Chaya Schwob.
\newblock Fast {C-K-R} partitions of sparse graphs.
\newblock {\em Chic. J. Theor. Comput. Sci.}, 2009, 2009.

\bibitem{DBLP:journals/jcss/Miller86}
Gary~L. Miller.
\newblock Finding small simple cycle separators for 2-connected planar graphs.
\newblock {\em J. Comput. Syst. Sci.}, 32(3):265--279, 1986.

\bibitem{DBLP:books/daglib/0017763}
Giri Narasimhan and Michiel H.~M. Smid.
\newblock {\em Geometric spanner networks}.
\newblock Cambridge University Press, 2007.

\bibitem{DBLP:conf/stoc/PatrascuT06}
Mihai P{u{a}}tra{c{s}}cu and Mikkel Thorup.
\newblock Time-space trade-offs for predecessor search.
\newblock In {\em Proceedings of the 38th Annual {ACM} Symposium on Theory of
  Computing (STOC)}, pages 232--240. {ACM}, 2006.

\bibitem{rodrigue2020geography}
Jean-Paul Rodrigue.
\newblock {\em The geography of transport systems}.
\newblock Routledge, 2020.

\bibitem{DBLP:journals/ipl/SchwarzkopfV96}
Otfried Schwarzkopf and Jules Vleugels.
\newblock Range searching in low-density environments.
\newblock {\em Inf. Process. Lett.}, 60(3):121--127, 1996.

\bibitem{DBLP:conf/stoc/SleatorT81}
Daniel Sleator and Robert Tarjan.
\newblock A data structure for dynamic trees.
\newblock In {\em Proceedings of the 13th Annual {ACM} Symposium on Theory of
  Computing (STOC)}, pages 114--122. {ACM}, 1981.

\bibitem{DBLP:conf/stoc/Solomon14}
Shay Solomon.
\newblock From hierarchical partitions to hierarchical covers: optimal
  fault-tolerant spanners for doubling metrics.
\newblock In {\em Proceedings of the 46th Annual {ACM} Symposium on Theory of
  Computing (STOC)}, pages 363--372. {ACM}, 2014.

\bibitem{DBLP:conf/stoc/Talwar04}
Kunal Talwar.
\newblock Bypassing the embedding: algorithms for low dimensional metrics.
\newblock In {\em Proceedings of the 36th Annual {ACM} Symposium on Theory of
  Computing (STOC)}, pages 281--290. {ACM}, 2004.

\bibitem{DBLP:journals/jacm/Thorup04}
Mikkel Thorup.
\newblock Compact oracles for reachability and approximate distances in planar
  digraphs.
\newblock {\em J. {ACM}}, 51(6):993--1024, 2004.

\bibitem{DBLP:conf/stoc/ThorupZ01}
Mikkel Thorup and Uri Zwick.
\newblock Approximate distance oracles.
\newblock In {\em Proceedings on 33rd Annual {ACM} Symposium on Theory of
  Computing (STOC)}, pages 183--192. {ACM}, 2001.

\bibitem{DBLP:books/daglib/0084325}
A.~Frank van~der Stappen.
\newblock {\em Motion planning amidst fat obstacles}.
\newblock University of Utrecht, 1994.

\bibitem{DBLP:journals/dcg/StappenOBV98}
A.~Frank van~der Stappen, Mark~H. Overmars, Mark de~Berg, and Jules Vleugels.
\newblock Motion planning in environments with low obstacle density.
\newblock {\em Discret. Comput. Geom.}, 20(4):561--587, 1998.

\bibitem{wiki:amgm}
{W}ikipedia{,} The Free~Encyclopedia.
\newblock {AM–GM} inequality.
\newblock Retrieved from \url{https://en.wikipedia.org/wiki/AM-GM_inequality},
  accessed April 2025.

\bibitem{wiki:graphhopper}
{W}ikipedia{,} The Free~Encyclopedia.
\newblock Graphhopper.
\newblock Retrieved from \url{https://en.wikipedia.org/wiki/GraphHopper},
  accessed April 2025.

\bibitem{wiki:gridplan}
{W}ikipedia{,} The Free~Encyclopedia.
\newblock Grid plan.
\newblock Retrieved from \url{https://en.wikipedia.org/wiki/Grid_plan},
  accessed April 2025.

\bibitem{wiki:suburbanisation}
{W}ikipedia{,} The Free~Encyclopedia.
\newblock Suburbanization.
\newblock Retrieved from \url{https://en.wikipedia.org/wiki/Suburbanization},
  accessed April 2025.

\bibitem{DBLP:conf/soda/Wulff-Nilsen16}
Christian Wulff{-}Nilsen.
\newblock Approximate distance oracles for planar graphs with improved query
  time-space tradeoff.
\newblock In {\em Proceedings of the 27th Annual {ACM-SIAM} Symposium on
  Discrete Algorithms (SODA)}, pages 351--362. {SIAM}, 2016.

\bibitem{zhang2013cul}
Ming Zhang.
\newblock On the cul-de-sac vs. checker-board street network: Search for
  sustainable urban form.
\newblock {\em International review for spatial planning and sustainable
  development}, 1(1):1--16, 2013.

\end{thebibliography}

\end{document}